\documentclass{lmcs}
\usepackage{xcolor}
\usepackage{tikz}
\usepackage{tikz-cd}
\usepackage{bussproofs}
\usepackage{amssymb}
\usepackage{upgreek}
\usepackage{listings}
\usepackage{mathtools}
\lstdefinestyle{customc}{
  belowcaptionskip=1\baselineskip,
  breaklines=true,
  frame=L,
  xleftmargin=\parindent,
  language=C,
  showstringspaces=false,
  basicstyle=\footnotesize\ttfamily,
  keywordstyle=\bfseries\color{blue!90!black},
  commentstyle=\itshape\color{purple},
  identifierstyle=\color{black},
  stringstyle=\color{orange},
}
\newcommand{\lctrs}{LCTRS}
\newcommand{\lcstrs}{LCSTRS}
\newcommand{\lctrss}{LCTRSs}
\newcommand{\lcstrss}{LCSTRSs}
\newcommand{\simplify}{(Simplify)}
\newcommand{\induct}{(Induct)}
\newcommand{\case}{(Case)}
\newcommand{\generalize}{(Generalize)}
\newcommand{\generalizemarked}{(Generalize$\circledast$)}
\newcommand{\alter}{(Alter)}
\newcommand{\delete}{(Delete)}
\newcommand{\hdelete}{($\hs$-Delete)}
\newcommand{\disprove}{(Disprove)}
\newcommand{\semicons}{(Context-Symb)}
\newcommand{\headvar}{(Context-Var)}
\newcommand{\headvarmarked}{(Context-Var$\circledast$)}
\newcommand{\postulate}{(Postulate)}
\newcommand{\postulatemarked}{(Postulate$\circledast$)}
\newcommand{\calc}{(Calc)}
\newcommand{\axiom}{(Axiom)}
\newcommand{\expand}{(Expand)}
\newcommand{\stepCC}{(Step\textsf{CC})}
\newcommand{\stepCI}{(Step\textsf{CI})}
\newcommand{\stepII}{(Step\textsf{II})}
\newcommand{\completeness}{(Complete)}
\newcommand{\eqdeletion}{(EQ-Deletion)}

\newcommand{\Sorts}{\mathcal{S}}
\newcommand{\thSorts}{\Sorts_{theory}}
\newcommand{\Types}{\mathcal{T}}
\newcommand{\thTypes}{\Types_{theory}}

\newcommand{\Sig}{\Sigma}
\newcommand{\thSig}{\Sig_{theory}}
\newcommand{\termsSig}{\Sig_{terms}}
\newcommand{\Var}{\mathcal{V}}
\newcommand{\Vars}[1]{Var(#1)}
\newcommand{\Val}{\mathcal{V}\textit{al}}
\newcommand{\f}{\mathsf{f}}

\newcommand{\nil}{\mathsf{nil}}
\newcommand{\cons}{\mathsf{cons}}

\newcommand{\init}{\mathsf{init}}
\newcommand{\sumfun}{\mathsf{sumfun}}
\newcommand{\map}{\mathsf{map}}
\newcommand{\fold}{\mathsf{fold}}
\newcommand{\append}{\mathsf{app}}
\newcommand{\rev}{\mathsf{reva}}

\newcommand{\recdown}{\mathsf{recdown}}
\newcommand{\tailup}{\mathsf{tailup}}

\newcommand{\add}{\mathsf{add}}
\newcommand{\mul}{\mathsf{mul}}

\newcommand{\Terms}{T(\Sig, \Var)}
\newcommand{\GroundTerms}{T(\Sig, \emptyset)}
\newcommand{\preTerms}{\mathbb{T}}
\newcommand{\myrule}{\ell \to r \ [\varphi]}
\newcommand{\mgu}[2]{\text{mgu}(#1, #2)}
\newcommand{\bool}{\mathsf{bool}}
\newcommand{\true}{\top}
\newcommand{\false}{\bot}
\renewcommand{\int}{\mathsf{int}}

\newcommand{\lijst}{\mathsf{list}}
\newcommand{\rules}{\mathcal{R}}
\newcommand{\calcrules}{\rules_{calc}}
\newcommand{\Defined}{\mathcal{D}}
\newcommand{\Cons}{\mathcal{C}}

\newcommand{\semiCons}{\mathcal{SCT}}
\newcommand{\gsemiCons}{\mathcal{SCT}^\emptyset}
\newcommand{\coverset}{\mathcal{C}}
\newcommand{\Q}{\mathcal{Q}}
\newcommand{\reqs}{\mathsf{REQS}}
\newcommand{\eqs}{\mathcal{E}}
\newcommand{\hs}{\mathcal{H}}

\newcommand{\pfst}[2]{\left(#1, #2\right)}

\newcommand{\deduces}[2]{#1 \vdash^{*} #2}
\newcommand{\oldeqs}{\mathcal{Q}}
\newcommand{\cflag}{F}

\newcommand{\typeInterpret}[1]{\mathcal{I}_{#1}}
\newcommand{\termInterpret}[1]{[\![ #1 ]\!]}
\newcommand{\rw}{\to_{\rules}}

\newcommand{\rwleft}{\:\prescript{}{\rules}{\leftarrow}\:}
\newcommand{\rwleftstar}{\:\prescript{*}{\rules}{\leftarrow}\:}
\newcommand{\rwleftplus}{\:\prescript{+}{\rules}{\leftarrow}\:}
\newcommand{\rweq}[3]{{#1} \mathrel{\leftrightarrow^{*}_{#3}} {#2}}

\newcommand{\rwseq}[3]{#1 \leftrightarrow_{#3} #2}

\newcommand{\arrztop}[1]{\arrz_{#1,\mathtt{root}}}
\newcommand{\arrzin}[1]{\arrz_{#1,\mathtt{in}}}

\newcommand{\critpairs}[1]{\mathsf{CP}(#1)}
\newcommand{\arrtype}{\to}
\newcommand{\atype}{\sigma}
\newcommand{\btype}{\tau}
\newcommand{\asort}{\iota}
\newcommand{\bsort}{\kappa}
\newcommand{\arrz}{\to}
\newcommand{\afun}{\mathsf{f}}
\newcommand{\bfun}{\mathsf{g}}
\newcommand{\typeof}{\mathit{typeof}}
\newcommand{\strue}{\mathtt{true}}
\newcommand{\sfalse}{\mathtt{false}}
\newcommand{\domain}{\mathit{dom}}
\newcommand{\image}{\mathit{im}}
\newcommand{\avar}{x}
\newcommand{\Pos}{\mathit{Pos}}
\newcommand{\prefix}[1]{\mathsf{[#1]}}
\newcommand{\symb}[1]{\mathsf{#1}}
\newcommand{\arity}{\mathit{ar}}
\newcommand{\alcstrs}{\mathcal{L}}
\newcommand{\eqlabel}[2]{\textbf{(#1}_{#2}\textbf{)}}
\newcommand{\complete}{\mathsf{C}}
\newcommand{\incomplete}{\mathsf{I}}

\usepackage{hyperref}
\theoremstyle{plain} 

\newcommand{\head}{\mathit{head}}
\newcommand{\supterm}{\rhd}
\newcommand{\suptermeq}{\unrhd}

\newcommand{\csucc}[3]{#1 \succ #2\ [#3]}
\newcommand{\csucceq}[3]{#1 \succeq #2\ [#3]}
\newcommand{\csucceqstrong}[3]{#1 \succeq^! #2\ [#3]}
\newcommand{\sterm}{\varsigma}
\newcommand{\tterm}{\uptau}
\newcommand{\hypothesis}{(Hypothesis)}

\definecolor{eqctxblue}{RGB}{0,114,178}
\newcommand{\eqctxcolor}{eqctxblue} 
\newcommand{\eqconwc}[4]{(\textcolor{\eqctxcolor}{\underline{#1}}\ ;\ #2 \approx #3\ ;\  \textcolor{\eqctxcolor}{\underline{#4}})}
\newcommand{\eqconsim}[5]{(\textcolor{\eqctxcolor}{\underline{#1}}\ ;\ #2 \simeq #3\ ;\ \textcolor{\eqctxcolor}{\underline{#4}})\ [#5]}
\newcommand{\eqcon}[5]{\eqconwc{#1}{#2}{#3}{#4}\ [#5]}

\newcommand{\eqsN}{\eqs_N}
\newcommand{\succeqmul}{\succeq_{mul}}
\newcommand{\succmul}{\succ_{mul}}
\newcommand{\converteither}[6]{\leftrightarrow_{#3;#4;#5}^{\{#1,#2\}#6}}
\newcommand{\convert}[5]{\converteither{#1}{#2}{#3}{#4}{#5}{*}}
\newcommand{\convertsingle}[5]{\converteither{#1}{#2}{#3}{#4}{#5}{}}

\newcommand{\bconvert}[5]{\convert{#1}{#2}{#3}{#4}{#5}}

\newcommand{\samereduct}[3]{#1 \rightarrow_{#3}^*\leftarrow\! #2}
\newcommand{\multiset}[1]{\{\!\{#1\}\!\}}
\newcommand{\peak}[4]{#2 \leftarrow #1 \rightarrow #3\ [#4]}
\newcommand{\brightRp}[1]{\mathrel{\xrightarrow{#1\succ*}}_\rules}
\newcommand{\bleftRp}[1]{\mathrel{\xleftarrow{#1\succ*}}_\rules}
\newcommand{\bconvertRp}[1]{\mathrel{\xleftrightarrow{#1\succ*}}_\rules}

\newcommand{\DP}{\mathtt{DP}}

\newcommand{\Defs}{\mathit{Defs}}
\newcommand{\Pstrong}{\mathcal{P}_{\text{strong}}}
\newcommand{\Pweak}{\mathcal{P}_{\text{weak}}}
\newcommand{\cora}{\textsf{Cora}}

\newcommand{\innerenumA}{(\textbf{i})}
\newcommand{\innerenumB}{(\textbf{ii})}
\newcommand{\innerenumC}{(\textbf{iii})}
\newcommand{\innerenumD}{(\textbf{iv})}
\newcommand{\innerenumAprime}{(\textbf{i'})}
\newcommand{\innerenumBprime}{(\textbf{ii'})}

\begin{document}

\title{Bounded Rewriting Induction for \lcstrss}

\author[K.~Hagens]{Kasper Hagens\lmcsorcid{0009-0005-2382-0559}}
\author[C.~Kop]{Cynthia Kop\lmcsorcid{0000-0002-6337-2544}}

\address{Radboud University, Nijmegen}
\email{kasper.hagens@ru.nl, c.kop@cs.ru.nl}  

\thanks{This work is supported by the NWO VIDI project
  ``Constrained Higher-Order Rewriting and Program Equivalence'',
  NWO VI.Vidi.193.075.}

\begin{abstract}
  \noindent Rewriting Induction (RI) is a method to prove inductive theorems, originating from equational reasoning.
  By using Logically Constrained Simply-typed Term Rewriting Systems\ (\lcstrss) as an intermediate language,
  rewriting induction becomes a tool for program verification, 
  with inductive theorems taking the role of equivalence predicates.
  
  Soundness of RI depends on well-founded induction, and one of the core obstacles for obtaining a practically useful proof system is to find suitable well-founded orderings automatically.
  Using naive approaches, all induction hypotheses must be oriented within the well-founded ordering,
  which leads to very strong termination requirements. This, in turn, severely limits the proof capacity of RI. 
  Here, we introduce \emph{Bounded RI}: an adaptation of RI for \lcstrss\ where such termination requirements are minimized.

  Traditionally, RI 
  can be used not only to prove equivalence, but also to establish ground confluence. 
  Moreover, for ground confluent TRSs, RI can be extended to a system for disproving inductive theorems.
  We will show that this is also possible with Bounded RI.
\end{abstract}

\maketitle

\section{Introduction} 
\noindent Rewriting Induction (RI) is a proof system for showing equations to be inductive theorems. 
It was introduced by Reddy~\cite{red:90} 
as a method to validate inductive proof procedures based on Knuth-Bendix completion.
Classically, RI is used in equational reasoning to prove properties of inductively defined mathematical structures
like natural numbers or lists.
For example, in Peano arithmetic one may orient the axioms defining addition as rules
($\symb{add}(x,\symb{0}) \arrz x$ and $\symb{add}(x,\symb{s}(y)) \arrz \symb{s}(\symb{add}(x,y))$),
and use RI to show that the equation $\symb{add}(x,y)\approx \symb{add}(y,x)$ is an inductive theorem.
RI has been extended to constrained rewriting~\cite{fal:kap:12,fuh:kop:nis:17,nak:nis:kus:sak:sak:10}, and recently to higher-order constrained rewriting~\cite{hag:kop:24}. 
These formalisms closely relate to real
programming and therefore have a natural place in the larger toolbox for program verification. 
Programs are represented by term rewriting systems, and equivalence of two functions within a program is modeled by an equation being an inductive theorem.

\paragraph{Why constrained rewriting?}
Using RI for program equivalence somewhat differs from the standard setting in equational reasoning where, for example, the Peano axioms are used to prove statements about the natural numbers. 
In our case, we are not so much interested in proving properties about numbers themselves, but about programs that operate on them.
Of course, we can define functions like $\mathsf{add}$, $\mathsf{mul}$ and
$\mathsf{greater}$, express the Peano axioms as rewrite rules and use this to
define programs on natural numbers.
However, when doing so,
studying program equivalence becomes a cumbersome experience that requires both reasoning about the arithmetic and the program definition itself.
Moreover, in practice, we typically want to reason about integers or even bit vectors rather than natural numbers, which requires correspondingly harder arithmetic reasoning.
Intermingling these two different kinds of reasoning makes it hard to scale analysis.
Ideally, we would want arithmetic to be given ``for free'', as it is in real life programming.
With standard term rewriting this is not possible. 

Constrained term rewriting provides a solution here, as it natively supports primitive data structures, such as integers, bit vectors and floating point numbers. 
This makes it possible to distinguish between the program definition (represented by rewrite rules), and underlying data structures with their operators
(represented by distinguished terms with pre-determined semantical interpretations). 
This allows us to shift some of the proof-burden from the rewriting side to the semantical side, where dedicated SMT solvers can be used.

In constrained rewriting, rewrite rules have a shape $s \to t\ [\varphi]$ where the boolean constraint $\varphi$ acts as a guard, in order to manage control flow over primitive data structures.
Here, we will consider Logically Constrained Simply-typed Term Rewriting Systems (\lcstrss), which consider applicative higher-order rewriting (without $\lambda$-abstractions) and first-order constraints~\cite{guo:kop:24}.
In particular, we will build on our earlier work~\cite{hag:kop:24} where we defined RI for \lcstrss\ 
(but do not assume familiarity with this work or other definitions of RI).

\paragraph{Goals}
Our goals are threefold:
\begin{itemize}[label=$\triangleright$]
\item 
\emph{Redefine RI for \lcstrss\ in such a way that we minimize the termination requirements.}
Induction proofs in RI require induction hypotheses to be oriented in a well-founded ordering.
This has the potential to give very strong termination requirements -- which, as we will see,
is not necessary.  By adapting the definition of RI, we can significantly reduce these termination
requirements and thereby make (automatic) proof search easier.
\item 
\emph{Using RI as a method for proving ground confluence.}
For first-order unconstrained rewriting it has been shown that RI can be combined with critical pairs to obtain a method for proving ground confluence~\cite{aot:toy:16}. 
Here, we extend this result to \lcstrss.    
\item
\emph{Disproving equations.}
For first-order, ground confluent constrained TRSs, RI can also be used to prove that equations are not inductive theorems~\cite{fuh:kop:nis:17},
but the higher-order proof system in~\cite{hag:kop:24} does not yet support this.
Here, we extend this result to \lcstrss.    
\end{itemize}

\paragraph{Termination requirements}
The name \emph{Rewriting Induction} refers to the principle that for a terminating rewrite system $\rules$, the reduction relation $\rw^+$ defines a well-founded order on the set of all terms, and therefore can be used for proofs by well-founded induction.
In many cases, however, we will need a well-founded order $\succ$ which is strictly larger than $\rw^+$.

In particular, the role of induction hypotheses in RI is taken by equations, which,
when applied, must be oriented w.r.t.~$\succ$.
That is, we 
can only
use an induction hypothesis $s \approx t$ if $s \succ t$ or $t \succ s$ holds.
Consequently, termination of $\rules$ itself is not enough, since equations are not usually orientable by \( \rw^+ \).
Instead, we might for instance let \( \succ \ = \ \to^+_{\rules \cup \{s \to t\}}
\), or in the case of multiple induction hypotheses, orient them as rewrite rules,
collect them into a set $\hs$, and use \( \succ \ = \ \to^+_{\rules \cup \hs} \).  
However, doing so leaves us with an
obligation to show termination of $\rules \cup \hs$.
Even if \( \rules \) is known to be terminating, it may not be easy or even
possible to prove the same for \( \rules \cup \hs \)
(think for instance of an induction hypothesis \( \symb{add}(x,y) \approx \symb{add}(y,x) \), which is not orientable in either direction).
In such a situation a RI proof might get stuck.  

As already observed by Reddy~\cite{red:90}, we do not necessarily need every induction hypothesis to be oriented, so long as 
an induction rule $s \to t$ is only applied to terms $\succ$-smaller than $s$. 
For this, it is not required to choose the well-founded ordering $\succ \ = \ \to^+_{\rules \cup \hs}$. 
Reddy proposed using modulo rewriting, to build a well-founded order $\succ$ which may not need to contain all induction rules.  
This approach was investigated by Aoto, who introduced several extensions of RI for first-order unconstrained rewriting~\cite{aot:06,aot:08a,aot:08b}. 
Here, we will follow a strategy along the same idea: by redefining RI we construct a well-founded relation $\succ$ during the RI process, aiming to keep it as small as possible.

\paragraph{Paper setup and contributions}
In \autoref{sec:preliminaries} we recap \lcstrss, equations and inductive theorems.  The subsequent sections present our contributions:
\begin{itemize}[label=$\triangleright$]
\item In \autoref{sec:RI} we introduce Bounded RI for \lcstrss.
  This system builds on the one introduced in~\cite{hag:kop:24} -- and is thus designed for \emph{higher-order}, \emph{constrained} systems -- but provides a more semantic approach and strictly contains it.
  Contributions include:
  \begin{itemize}
  \item We introduce the notion of \emph{bounding pair} $(\succ,\succeq)$, providing the fundamental ingredient by which we express the ordering requirements for induction proofs.
  \item We replace equations by \emph{equation contexts}, containing the extra information of two \emph{bounding terms} which define an upper bound for applying an induction hypothesis. 
    More specifically, the bounding terms provide us with a way to keep track of terms to be used in the ordering requirements (along the equation we are proving) --
    instead of orienting the induction hypotheses themselves -- yielding milder termination requirements.
  \item In contrast to~\cite{hag:kop:24}, we also include derivation rules to deduce non-equivalence (\autoref{subsec:completeness}) for ground confluent \lcstrss. This is thus far only achieved for first-order definitions of RI~\cite{fuh:kop:nis:17}.
  \end{itemize}
\item In \autoref{sec:proofs}, we prove soundness and completeness of Bounded RI.
  Apart from its necessity for the results of \autoref{sec:RI} to be meaningful, there are two additional contributions:
  \begin{itemize}
  \item We show that Bounded RI can be used to prove that equations are \emph{bounded ground convertible}, which is strictly stronger than the
    property of ground convertibility that was used in~\cite{hag:kop:24}, and which explicitly relates to the bounding pair $(\succ, \succeq)$.
  \item We structure our soundness proof in a very modular way, by focusing on two properties that all but one of our deduction
    rules satisfy.  This makes it possible to 
    easily add new derivation rules to the existing system without having to redo the full proof.
  \end{itemize}
\item In \autoref{sec:relativePower} we illustrate how other RI systems are
  captured by Bounded RI, and how previous examples are now handled more easily and with fewer termination requirements.
\item In \autoref{sec:howToFindOrdering} we discuss two strategies to construct a suitable well-founded ordering.
  In particular, we can define the ordering as a reduction relation \( \to^+_{\rules \cup \Q} \) and briefly discuss
  strategies to choose the set $\Q$.  We also show that we can use a version of the dependency pair framework directly
  to obtain more liberal requirements.
\item In \autoref{sec:newDedEasy} we illustrate how the two soundness properties -- \emph{Preserving Bounds} (\autoref{def:boundPreserve})
  and \emph{Base Soundness Property} (\autoref{def:BaseSound}) -- can be exploited to obtain new deduction rules,
  by extending Bounded RI with some practically useful deduction rules.
\item Following~\cite{aot:toy:16,aot:toy:kim:17}, RI can be used to prove ground confluence of a first-order TRS
  without constraints.  In \autoref{sec:GroundConfluence} we show that this result extends to Bounded RI, and to
  ground confluence of an \lcstrs.
  \begin{itemize}
  \item We recap the definition of first-order critical pairs for Logically Constrained TRSs (\lctrss)\ from~\cite{sch:mid:23},
    extend it to \lcstrss\ and prove the Critical Peak Lemma.
    This is a first critical peak / pair definition for higher-order constrained systems.
  \item We introduce the Ground Critical Peak Theorem, and show how this allows us to use RI for proving ground confluence.
    This follows the idea of~\cite{aot:toy:16}, but makes important adaptations to work well with Bounded RI.
  \end{itemize}
\item We have implemented all methods in this work in our tool \cora, to allow for human-guided equivalence proofs
  with fully automatic termination and SMT implication checks.
  In \autoref{sec:implementation} we discuss our implementation.
\end{itemize}

We conclude, and provide directions for future work, in \autoref{sec:conclusionFuture}.

\begin{rem}
A preliminary version of this paper was presented at WPTE 2025~\cite{hag:kop:25},
which considered primarily Sections~\ref{sec:boundedRI} and~\ref{sec:example}.
The contributions in all later sections -- including the ability to prove non-equivalence,
the proofs (and modular proof strategy), the strategies for finding a bounding pair, the results on
ground confluence and the implementation -- are new to the current paper.
\end{rem}

\section{Preliminaries}\label{sec:preliminaries}

\subsection{Logically Constrained Simply-typed Term Rewriting Systems}
We will recap \lcstrss~\cite{guo:kop:24}, a higher-order version of \lctrss~\cite{kop:nis:13}.
This formalism has built-in support for data structures like integers and booleans (or in fact any
theory like bitvectors, floating point numbers or integer arrays) as well as logical constraints to model control flow. 
\lcstrss\ are \emph{applicative} higher-order TRSs (without $\lambda$) with \emph{first-order} constraints. 

\paragraph{Types}
Assume given a set of sorts (base types) $\Sorts$.  The set $\Types$ of types is defined by the
grammar
\(
    \Types
    ::=
    \Sorts
    \mid
    \Types \arrtype \Types
\).
Here, $\arrtype$ is right-associative, so all types may be written as
$\mathit{type}_1 \arrtype \dots \arrtype \mathit{type}_m \arrtype \mathit{sort}$
with $m \geq 0$.

We also assume given a subset $\thSorts \subseteq \Sorts$ of \emph{theory sorts} (e.g., $\int$ and
$\bool$), and define the \emph{theory types} by the grammar
\(
    \thTypes
    ::=
    \thSorts
    \mid
    \thSorts \to \thTypes
\).
Each theory sort $\iota \in \thSorts$ is associated with a non-empty interpretation set $\typeInterpret{\iota}$ (e.g.,
$\typeInterpret{\int} = \mathbb{Z}$, the set of all integers). 
We define $\typeInterpret{\asort \arrtype \atype}$ as the set of all total functions from $\typeInterpret{\asort}$ to $\typeInterpret{\atype}$.

\paragraph{Terms}

We assume given a signature $\Sig$ of \emph{function symbols} and a disjoint set $\Var$ of
variables, and a function $\typeof$ from $\Sig \cup \Var$ to $\Types$. We require that there
are infinitely many variables of all types.
The set of terms $\Terms$ over $\Sig$ and $\Var$ consists of the expressions in
$\preTerms$ -- defined by the grammar
\(
    \preTerms
    ::=
    \Sig
    \mid
    \Var
    \mid
    (\preTerms\ \preTerms)
\) -- that are \emph{well-typed}, using
\innerenumA\ $a :: \typeof(a)$ for $a \in \Sig \cup \Var$, and
\innerenumB\ if $s :: \atype \arrtype \btype$ and $t :: \atype$ then $(s\ t) :: \btype$.
Application is left-associative and we omit brackets where no ambiguity can
arise, which allows all terms to be written in a form
$a\ t_1 \cdots t_n$ with $a \in  \Sig \cup \Var$ and $n \geq 0$.
Writing $t =a\ t_1 \cdots t_n$, we define $\head(t)=a$.
For a term $t$, let $\Vars{t}$ be the set of variables occurring in $t$.
For multiple terms $t_1,\dots,t_n$, let $\Vars{t_1,\dots,t_n}$ denote
$\Vars{t_1} \cup \dots \cup \Vars{t_n}$.
A term $t$ is \emph{ground} if $\Vars{t} = \emptyset$.
A term $t$ is \emph{linear} if no variable occurs more than once in $t$.

We assume that $\Sig$ is the disjoint union $\thSig \uplus \termsSig$, where
$\typeof(\afun) \in \thTypes$ for all $\afun \in \thSig$.
We use infix notation for the binary symbols, or use $\prefix{\afun}$ for prefix or partially
applied notation (e.g., $\prefix{+}\ x\ y$ and $x + y$ are the same).
Each $\afun \in \thSig$ has an interpretation $\termInterpret{\afun} \in \typeInterpret{\typeof(\afun)}$.
For example, a theory symbol $*::\int \to \int \to \int$ may be interpreted as
multiplication on $\mathbb{Z}$.
Symbols in $\termsSig$ do not have an interpretation.
\emph{Values} are theory symbols of base type, i.e.
\(
    \Val
    =
    \{
        v \in \thSig
        \mid
        \typeof(v) \in \thSorts
    \}
\).
We assume there is exactly one value for each element of $\typeInterpret{\asort}$
($\asort \in \thSorts$).
Elements of $T(\thSig, \Var)$ are called \emph{theory terms}.  For \emph{ground} theory terms, we define
$\termInterpret{s\ t} = \termInterpret{s}(\termInterpret{t})$, thus mapping each ground term of type
$\atype$ to an element of $\typeInterpret{\atype}$.
We fix a theory sort $\bool$ with $\typeInterpret{\bool} = \{\true,\false\}$, and
assume that for every theory sort $\asort$, there is an equality symbol
$=_\asort :: \asort \arrtype \asort \arrtype \bool$ in $\thSig$.
A \emph{constraint} is a theory term $s$ of type $\bool$, such that $\typeof(\avar) \in \thSorts$
for all $x \in \Vars{s}$.

\begin{exa}\label{ex:theorysig}
Throughout this paper we will always use $\thSorts = \{\int,\bool\}$ and
$\thSig = \{+,-,*,<,\le,>,\geq,=,\wedge,\vee,\neg,\Longrightarrow,\Longleftrightarrow,\strue,\sfalse\} \cup
\{\symb{n} \mid n \in \mathbb{Z}\}$, with $+,-,*
::\int \to \int \to \int$, $<,\le,>,\geq,=::\int \to \int \to \bool$, 
$\wedge,\vee,\Longrightarrow,\Longleftrightarrow :: \bool \to \bool \to \bool$,
$\neg :: \bool \to \bool$, 
$\strue,\sfalse :: \bool$
and $\symb{n} :: \int$.  
We let $\typeInterpret{\int} = \mathbb{Z}$, $\typeInterpret{\bool} =
\{\true,\false\}$ and interpret all symbols as expected.
The values are $\strue,\sfalse$ and all $\mathsf{n}$.  Theory terms are for instance
$x + \symb{3}$, $\strue$ and $\symb{7} * \symb{0}$.  The latter two are ground.  We have
$\termInterpret{\symb{7} * \symb{0}} = 0$.
The theory term $x > \mathsf{0}$ is a constraint, but the theory term $(f\ x) > \mathsf{0}$ with
$f \in \Var$ of type $\int \arrtype \int$ is not (since $\typeof(f) \notin \thSorts$),
nor is $\prefix{>}\ \mathsf{0} :: \int \arrtype \bool$ (since constraints must have type $\bool$).
\end{exa}

\begin{rem}
Most programming languages have pre-defined (non-recursive) data structures and operators, e.g. the integers with a multiplication operator $*$.  
This makes it, for instance, possible to define the factorial function without first having to define multiplication.  
The same is true for \lcstrss : $\thSig$ contains all pre-defined operators, including constants. 
\end{rem}

\paragraph{Substitutions}

A substitution is a type-preserving mapping $\gamma: \Var \to \Terms$.
The domain of a substitution is defined as $\domain(\gamma) = \{x \in \Var \mid \gamma(x) \ne x\}$,
and its image as $\image(\gamma) = \{\gamma(x) \mid x \in \domain(\gamma)\}$.
A substitution on finite domain $\{\avar_1,\dots,\avar_n\}$ is often denoted $[\avar_1:=s_1,\dots,
\avar_n:=s_n]$.
A substitution $\gamma$ is extended to a function $s \mapsto s\gamma$ on terms by placewise
substituting variables in the term by their image:
\innerenumA\ $t \gamma = t$ if $t \in \Sig$,
\innerenumB\ $t \gamma = \gamma(t)$ if $t \in \Var$, and
\innerenumC\ $(t_0\ t_1) \gamma = (t_0 \gamma)\ (t_1 \gamma)$.
If $M \subseteq \Terms$ then $\gamma(M)$ denotes 
$\{t \gamma \mid t \in M \}$.

A \emph{ground substitution} is a substitution \( \gamma \) such that for all
variables \( x \) in its domain, \( \gamma(x) \) is a ground term.
A \emph{unifier} of terms $s,t$ is a substitution $\gamma$ such that $s\gamma = t\gamma$.
A \emph{most general unifier} or \emph{mgu} is a unifier $\gamma$ such that all other unifiers are
instances of $\gamma$; that is, if $s\delta = t\delta$ as well, then there is a substitution $\epsilon$ such that $\delta(x) =
(x\gamma)\epsilon$, for all $x \in \domain(\delta)$.
For unifiable terms, an mgu always exists, and is unique modulo variable renaming.

We say that a substitution \( \gamma \) \emph{respects} a constraint \( \varphi \) if
$\gamma(\Vars{\varphi}) \subseteq \Val$ and $\termInterpret{\varphi\gamma} = \true$.
We say that a constraint \( \varphi \) is \emph{satisfiable} if there exists a substitution
\( \gamma \) that respects \( \varphi \). 
It is \emph{valid} if \( \termInterpret{\varphi\gamma} =
\top \) for all substitutions \( \gamma \) such that \( \gamma(\Vars{\varphi}) \subseteq \Val\).

\paragraph{Contexts and subterms}
Let $\square_1,\dots,\square_n$ be fresh, typed constants ($n \geq 1$).
A context $C[\square_1,\dots,\square_n]$ (or just: $C$) is a term in $T(\Sig \cup \{\square_1, \ldots, \square_n\}, \Var)$
in which each $\square_i$ occurs exactly once.
(Note that $\square_i$ may occur at the head of an application.)
The term obtained from $C$ by replacing each $\square_i$ by a term $t_i$ of the same type is denoted by $C[t_1,\dots,t_n]$.

We say that $t$ is a \emph{subterm} of $s$, notation $s \suptermeq t$,
if there is a context $C[\square]$ such that $s = C[t]$.
We say that $t$ is a \emph{strict} subterm of $s$, notation $s \supterm t$, if $s \suptermeq t$ and $s \neq t$.

\paragraph{Rewrite rules}
A rule is an expression $\myrule$.
Here $\ell$ and $r$ are terms of the same type, $\ell$ has a form $\afun\ \ell_1 \cdots \ell_k$ with
$\afun \in \Sig$ and $k \geq 0$, $\varphi$ is a constraint and
\( \Vars{r} \subseteq \Vars{\ell} \cup \Vars{\varphi} \).
If $\varphi = \strue$, we may denote the rule as just $\ell \to r$.

Fixing a signature $\Sig$, we define the set of \emph{calculation rules} as:
\[
\calcrules 
=
\left\{ 
\afun\ x_1 \cdots x_m \to y \ [ y = \afun\ x_1 \cdots x_m ]\ 
\middle|
\begin{array}{l}
\afun \in \thSig \setminus \Val,
\text{all}\ x_i\ \text{and}\ y \in \Var,\\
\typeof(\afun) = \asort_1 \arrtype \dots \arrtype \asort_m \arrtype \bsort
\end{array}
\right\}
\]

We furthermore assume given a set of rules \( \rules \) that satisfies the following properties:
\begin{itemize}[label=$\triangleright$]
\item for all \( \myrule \in \rules \): \( \ell \) is not a theory term
  (such rules only occur in \( \calcrules \))
\item for all \( \afun\ \ell_1 \cdots \ell_k \to r \ [\varphi],\ \bfun\ \ell_1' \cdots \ell_n' \to r'\ [\psi] \in \rules \cup \calcrules \):
  if \( \afun = \bfun \) then \( k = n \)
\end{itemize}

The latter restriction blocks us for instance from having both a rule $\symb{append}\ \symb{nil} \to \symb{id}$
and a rule $\symb{append}\ (\symb{cons}\ x\ y)\ z \to \symb{cons}\ x\ (\symb{append}\ y\ z)$.  While such rules
would normally be allowed in higher-order rewriting, we need to impose this limitation for the notion of
\emph{quasi-reductivity} to make sense, as discussed in \cite{hag:kop:24}.
This does not really limit expressivity, since we can 
pad both sides 
with variables, e.g.,
replacing the first rule above by $\symb{append}\ \symb{nil}\ x \to \symb{id}\ x$.

\begin{rem}
The definition of a rule used here differs slightly from the original definition of rules in
\cite{guo:kop:24,kop:nis:13}: we require that $\Vars{r} \subseteq \Vars{\ell} \cup
\Vars{\varphi}$, whereas in \cite{guo:kop:24,kop:nis:13} $\Vars{r}$ may contain fresh variables.
The reason for this change is simplicity: the reduction relation in~\cite{guo:kop:24,kop:nis:13}
forces all such fresh variables to be instantiated by values, and consequently replacing each such
rule by $\ell \arrz r\ [\varphi \wedge \bigwedge_{x \in \Vars{r} \setminus \Vars{\ell,\varphi}}
x = x]$ yields exactly the same reduction relation.
As this transformation (or a slight variation) is needed for several techniques, doing it
directly in the definition simplifies several definitions.
\end{rem}

For a fixed signature and rules $\rules$ as above, we define three classes of function symbols:
\begin{itemize}[label=$\triangleright$]
\item elements of $\Defined = \{ \afun \in \Sig \mid$ there is a rule $\afun\ \ell_1 \cdots \ell_k \to r\ [\varphi] \in \rules \}$ are called \emph{defined symbols};
\item elements of $\Cons = \Val \cup (\termsSig \setminus \Defined)$ are called \emph{constructors};
\item elements of $\Sig_{calc}=\thSig \setminus \Val$ are called \emph{calculation symbols}.
\end{itemize}

For every symbol $\afun :: \atype_1 \arrtype \dots \arrtype \atype_m \arrtype \asort$ with $\asort
\in \Sorts$, we let $\arity(\afun)$ either be the unique number $0 \le k \le m$ such that there is
a rule \( \afun\ \ell_1 \cdots \ell_k \to r\ [\varphi] \), or $\infty$ if there is no such rule.
By the restrictions above, such $k$ always exists and is unique for defined and calculation symbols;
for constructors $\afun$, necessarily $\arity(\afun) = \infty$.
Hence, $\arity(\afun)$ intuitively indicates the number of arguments that should be given to $\afun$
to obtain a reducible term.

We say that a substitution $\gamma$ respects $\myrule$ if it respects $\varphi$.

\begin{rem}
Note that it is theoretically possible for a rule to have a calculation symbol at the head of
its left-hand side, so a function symbol \emph{can} be both a calculation symbol and a defined
symbol.  However, this would rarely occur in practice -- it is allowed because methods like
rewriting induction might create ``rules'' of unusual forms to be used in a termination proof, and
it does not harm our method or the difficulty of the proofs to admit such rules.
\end{rem}

\paragraph{Reduction relation}
For a fixed signature $\Sig$ and set of rules $\rules$ as discussed above, the reduction relation
$\rw$ is defined by:
\[
C[\ell \gamma]
\rw
C[r \gamma]
\text{ if }
\myrule \in
\rules \cup \calcrules
\text{ and }
\gamma
\text{ respects }
\varphi
\]
Note that by definition of context, reductions may occur at the head of an application.
For example, if $\symb{append}\ \symb{nil} \to \symb{id} \in \rules$, then we could reduce
$\symb{append}\ \symb{nil}\ s\ \rw \symb{id}\ s$.
We say that $s$ \emph{has normal form} $t$ if $s \rw^* t$ and $t$ cannot be reduced.
If we want to emphasize that reduction $s \rw t$ is performed with a rule in $\calcrules$, we write $s \to_{\calcrules} t$ instead. 

\paragraph{\lcstrs}
A \emph{Logically Constrained Simply-typed Term Rewriting System (\lcstrs)} is a pair
$(\Terms,\arrz_\rules)$ generated by $(\Sorts,\thSorts,\termsSig,\thSig,\Var,\typeof,\typeInterpret{},\termInterpret{\cdot},\rules)$.
We often refer to an \lcstrs\ by $\alcstrs = (\Sig, \rules)$, or just $\rules$, leaving the rest implicit.
\begin{exa}\label{example:runningExample}
Let $\rules$ consist of the following rules 
\[
\begin{array}{lrcll}
\textbf{(R1)} & \recdown\ f\ n\ i\ a & \to & a & [i<n] \\
\textbf{(R2)} & \recdown\ f\ n\ i\ a & \to &
  f\ i\ (\recdown\ f\ n\ (i-\symb{1})\ a) & [i \ge n] \\
\textbf{(R3)} & \tailup\ f\ i\ m\ a & \to & a & [i > m] \\
\textbf{(R4)} & \tailup\ f\ i\ m\ a & \to &
  \tailup\ f\ (i+\symb{1})\ m\ (f\ i\ a) & [i \le m] \\
\end{array}
\]
We have     
$\Sorts = \thSorts = \{\int, \bool\}$,
$\termsSig = 
\{
\recdown,\ \tailup :: 
(\int \to \int \to \int)
\to 
\int 
\to
\int 
\to
\int 
\to 
\int
\}$ 
and   
$\thSig$ the same as in \autoref{ex:theorysig}. 
We have $\Sig_{calc}=\{+, -, *, <, \le, >,  \ge, =, \wedge, \vee\}$,
$\Defined = \termsSig$ and $\Cons = \Val = \{\strue,\sfalse\} \cup \{\symb{n} \mid n \in \mathbb{Z}\}$.  
The substitution $\gamma = [n:=\symb{0},\ i:=\symb{1}]$ respects \textbf{(R2)},
and induces a reduction
\[
\begin{array}{llll}
\recdown\ f\ \symb{0}\ \symb{1}\ a
  & \to_{\rules} &
f\ \symb{1}\ (\recdown\ f\ \symb{0}\ (\symb{1}-\symb{1})\ a)
  & \to_{\calcrules} \\
f\ \symb{1}\ (\recdown\ f\ \symb{0}\ \symb{0}\ a)
  & \to_{\rules} &
f\ \symb{1}\ 
(
    f\ \symb{0}\
    (
    \recdown\ f\ \symb{0}\ (\symb{0}\symb{-1})\ a
    )
)
  & \to_{\calcrules} \\
f\ \symb{1}\ 
(
    f\ \symb{0}\
    (
    \recdown\ f\ \symb{0}\ (\symb{-1})\ a
    )
)
  & \to_{\rules} &
f\ \symb{1}\ 
(
    f\ \symb{0}\
    a
)
\end{array}
\] 
It is easy to check that
$
(\tailup\ f\ n\ i\ a) \gamma 
=
\tailup\ f\ \symb{0}\ \symb{1}\ a
$
also reduces to 
$ 
f\ \symb{1}\ (f\ \symb{0}\ a)
$.

Intuitively, $\recdown$ and $\tailup$ define recursors which capture a class of simple programs that compute a return value via a recursive or tail-recursive procedure.
This considers programs using a loop index $i$, which is decreased or increased by $1$ at each recursive call, until $i$ is below the lower bound $n$ or above the upper bound $m$.
Then, the computation terminates and returns 
$a$.
With this, we can for example represent the two programs below
\hfill
\begin{minipage}[l]{0.45\textwidth}
\lstset{style=customc}
\begin{lstlisting}[language=C, style=customc]
int factRec(int x){
   if (x >= 1) 
      return(x*factRec(x-1));
   else 
      return 1; } 
\end{lstlisting}
\end{minipage}
\hfill
\begin{minipage}[l]{0.5\textwidth}
\lstset{style=customc}
\begin{lstlisting}
int factTail(int x){
   int a = 1, i = 1;
   while (i<=x){ a = i*a; i = i+1; }
   return a; }
\end{lstlisting}
\end{minipage}
\\
by introducing rewrite rules $
\symb{factRec}\ x
\to 
\recdown\ \prefix{*}\ \symb{1}\ x\ \symb{1}
$ 
and  
$
\symb{factTail}\ x
\to 
\tailup\ \prefix{*}\ \symb{1}\ x\ \symb{1}
$. 
Specifically, $\recdown\ \prefix{*}\ n\ i\ a$ computes
$(\prod_{k=n}^i k) \cdot a$ and $\tailup\ \prefix{*}\ j\ m\ b$ computes $(\prod_{k=j}^m k) \cdot b$. 
Hence, all ground instances of $\recdown\ \prefix{*}\ n\ i\ a$ and $\tailup\ \prefix{*}\ n\ i\ a$ produce the same result. 
In \autoref{sec:example} we will prove this with Bounded RI 
for arbitrary 
$f :: \int \to \int \to \int$. 
\end{exa}

We consider several properties an \lcstrs\ $\alcstrs = (\Sig, \rules)$ can exhibit:
\begin{itemize}[label = $\triangleright$]
\item $\alcstrs$ is \emph{terminating} if there is no infinite reduction sequence $s_0 \rw s_1 \rw s_2 \rw \cdots$ for any $s_0 \in \Terms$
\item $\alcstrs$ is \emph{weakly normalizing} if every term has at least one normal form \\
  (note that termination implies weak normalisation, but not the other way around)
\item $\alcstrs$ is \emph{confluent} if for all $s,t,q \in \Terms$: if $s \rw^*t$ and $s \rw^*q$ then there is some $w \in \Terms$ such that $t\rw^*w$ and $q\rw^*w$
\item $\alcstrs$ is \emph{ground confluent} if it is confluent on ground terms:
  for all $s,t,q \in T(\Sigma, \emptyset)$: if $s \rw^*t$ and $s \rw^*q$ then there is some $w \in T(\Sigma, \emptyset)$ such that $t\rw^*w$ and $q\rw^*w$
\end{itemize}
Weak normalisation and ground confluence together imply that every ground term has exactly one
normal form -- essentially, that every term computes a unique result.

\subsection{Rewriting induction prerequisites}

Finally, we will recap the notions of equations, inductive theorems, and the restrictions that an
\lcstrs\ must satisfy to be able to apply rewriting induction.
For this, we follow the definitions from~\cite{hag:kop:24}.

\paragraph{Equations and inductive theorems}
An \emph{equation} is a triple $s \approx t \ [\varphi]$ with
$\typeof(s) = \typeof(t)$ 
and $\varphi$ a constraint.
If the constraint $\varphi$ is $\strue$, we will simply write the equation as $s \approx t$.
A substitution $\gamma$ respects $s \approx t \ [\varphi]$ if $\gamma$ respects $\varphi$
and $\Vars{s,t} \subseteq \domain(\gamma)$.

An equation $s \approx t \ [\varphi]$ is an \emph{inductive theorem} (aka \emph{ground convertible})
if $\rweq{s \gamma}{t \gamma}{\rules}$ for every ground substitution $\gamma$ that respects
$s \approx t \ [\varphi]$.
Here $\leftrightarrow_\rules \ =\  \rightarrow_\rules \cup \leftarrow_\rules$,
and $\leftrightarrow_\rules^*$ is its transitive, reflexive closure. 

For a set of equations $\eqs$, we define $\rwseq{}{}{\eqs}$ as follows:
\[
    \rwseq{C[s \gamma]}{ C[t \gamma]}{\eqs}
    \text{ if }
    s \approx t\ [\varphi] \in \eqs
    \text{ or }
    t \approx s\ [\varphi] \in \eqs
    \text{ and }
    \gamma
    \text{ respects }
    \varphi
\]
Note that $\eqs$ consists of inductive theorems if and only if $\rwseq{}{}{\eqs}\ \subseteq\  \leftrightarrow^*_{\rules}$ holds on ground terms (i.e. when restricting to ground substitutions $\gamma$ and ground contexts $C$). 

\begin{rem}
In higher-order rewriting, there are multiple ways to define inductive theorems.
In particular one could choose to take into account extensionality~\cite{aot:yam:toy:04}, which
equates two functions if they are equal on all their arguments.
For example, in our definition $\prefix{+}\ \symb{0} \approx \prefix{*}\ \symb{1}$ is not
an inductive theorem (since the two sides are distinct ground normal forms), while it \emph{would}
be an inductive theorem in an extensional definition (since 
$\prefix{+}\ \symb{0}\ x \approx \prefix{*}\ \symb{1}\ x$ is ground convertible).
This choice was extensively discussed in~\cite{hag:kop:24}.
The primary considerations are that
\innerenumA\ naively defining extensional inductive theorems loses \emph{monotonicity} (i.e.,
$s \approx t\ [\varphi]$ might be an inductive theorem while $C[s] \approx C[t]\ [\varphi]$ is not),
so a more sophisticated (and non-obvious) definition is needed; and
\innerenumB\ no matter how we define such a notion, a ground convertible equation should always be
an extensional inductive theorem.
Hence, by pragmatically defining an inductive theorem to be a ground convertible equation, a method
that proves this property automatically extends to extensional definitions.  (However, a proof
method to show non-inductive theorems does not automatically extend.)

One may also consider \emph{extensibility}~\cite{aot:yam:chi:11}, which limits inductive theorems
to those equations that are still ground convertible in ``any reasonable extension'' of $\rules$.
We will briefly discuss this in \autoref{sec:conclusionFuture}.
\end{rem}

\begin{exa}\label{example:equation}
The \lcstrs\ from \autoref{example:runningExample} admits an equation $\recdown\ f\ n\ i\ a \linebreak
\approx 
\tailup \ f\ n\ i\ a$. 
Since it has constraint $\strue$, any substitution on domain $\supseteq \{f,n,i,a\}$ respects it. 
In \autoref{sec:example} we will prove that this equation is an inductive theorem, meaning that $\rweq{(\recdown\ f\ n\ i\ a
) \gamma}{(\tailup \ f\ n\ i\ a)\gamma}{\rules}$ for any ground substitution $\gamma$. 
\end{exa}
We will limit interest to \emph{quasi-reductive} \lcstrss\ (defined below), which is needed to guarantee soundness of RI. 
Intuitively, this property expresses that pattern matching on ground terms is exhaustive. 
For example, the 
system $\rules = \{\textbf{(R1)}, \textbf{(R2)}\}$ is quasi-reductive because $i<n$ and $i \ge n$ together cover all ground instances of $\recdown\ f\ n\ i\ a$.
But if, say, we replace $\textbf{(R2)}$ by $\recdown\ f\ n \ i\ a \to f\ i\ (\recdown\ f\ n\ (i-\symb{1})\ a)\ [i>n]$ then it is not,
since 
for example $\recdown\ \prefix{*}\ \symb{0} \ \symb{0}\ \symb{0}$ does not reduce anymore.

For first-order \lctrss, quasi-reductivity is achieved by requiring that there are no other ground normal forms than the ground constructor terms $T(\Cons, \emptyset)$. 
For higher-order \lcstrss, however, this definition does not work as we can have ground normal forms with partially applied defined symbols
(for example, $\recdown\ \prefix{+}$).
We therefore use the higher-order generalization of the notion of constructor terms that was introduced in~\cite{hag:kop:24}.

\paragraph{Semi-constructor terms}
Let $\alcstrs=(\Sig,\rules)$ be an \lcstrs.
The \emph{semi-constructor terms} over $\alcstrs$, notation $\semiCons_\alcstrs$, are defined by
\innerenumA\ 
  $\Var \subseteq \semiCons_\alcstrs$;
\innerenumB\ 
  if $\f \in \Sig$ with $\f :: \atype_1 \arrtype \dots \arrtype \atype_m \arrtype \asort$, $\asort \in \Sorts$
and $s_1 :: \atype_1, \ldots, s_n :: \atype_n \in \mathcal{SCT}_\alcstrs$ with $n \leq m$ and $n < \arity(\f)$, then
$\f\ s_1 \cdots s_n \in \mathcal{SCT}_\alcstrs$. 

Semi-constructor terms are always normal forms. 
Furthermore, as $\arity(\afun) = \infty$ for every 
constructor $\afun$, the constructor terms $T(\Cons, \Var)$ are contained in $\semiCons_\alcstrs$.
However, also terms with partial applications, such as $\recdown\ \prefix{+}\ \symb{3}$, are included.
The set $\gsemiCons_\alcstrs$ refers to \emph{ground} semi-constructor terms, built without \innerenumA.
A \emph{ground semi-constructor substitution} (abbreviated to \emph{gsc substitution}) is a
substitution such that $\image(\gamma)\subseteq \gsemiCons_{\alcstrs}$. 

\paragraph{Quasi-reductivity}
An \lcstrs\ $\alcstrs=(\Sig, \rules)$ is quasi-reductive if
for every $t \in T(\Sig, \emptyset)$ we have $t \in \semiCons_{\alcstrs}^{\emptyset}$ or $t$ reduces with
$\to_{\rules}$.  Put differently, the only irreducible ground terms are semi-constructor terms.
Weak normalization and quasi-reductivity together ensure that every ground term reduces to a semi-constructor term.
Note that, if $s_1, \ldots, s_n$ are ground normal forms and $\afun \in \Sigma$, then $\afun\ s_1 \cdots s_n$ is a ground normal form if and only if $n < \arity(\afun)$.

We say that an \lcstrs\ has \emph{inextensible theory sorts} if the only
constructors with a type \( \atype_1 \arrtype \dots \arrtype \atype_n \arrtype
\asort \) with \( \asort \in \thSorts \) are values (in which case \( n = 0 \)).
Hence, it is not allowed to have, for instance, a constructor $\symb{error} ::
\int$.  While we do not explicitly limit interest to \lcstrss\ with inextensible
theory sorts, it is satisfied in all our examples, because in practice it is
very difficult to achieve quasi-reductivity without this property.

\section{Rewriting Induction}\label{sec:RI}
RI was introduced by Reddy~\cite{red:90} in the year 1990 as a deduction system for proving inductive theorems, using unconstrained first-order term rewriting systems.
Since then, many variations on the system have appeared (e.g., \cite{aot:06,aot:08a,aot:08b,fal:kap:12,fuh:kop:nis:17,hag:kop:24}). 
All are based on well-founded induction, using some well-founded relation $\succ$. 
Depending on how $\succ$ is being constructed these versions of RI can be categorized into two approaches.
\begin{itemize}[label = $\triangleright$]
\item 
  The first, being used in~\cite{red:90,fal:kap:12,fuh:kop:nis:17,hag:kop:24},
  employs a fixed strategy to construct a
  terminating rewrite system $A \supseteq \rules$ and then chooses $\succ \ = \ \to_{A}^+$.  
\item
  The second, used in~\cite{aot:06,aot:08a,aot:08b}, employs a well-founded relation $\succ$ that satisfies
  certain requirements (like monotonicity and stability, but also ground totality), and constructs the proof
  in a more targeted way.
  This relation may either be fixed beforehand (e.g., the lexicographic path ordering), or constructed
  during or after the proof, as the proof process essentially accumulates termination requirements.
\end{itemize}
In practice, the former approach leads to quite heavy termination requirements because it forces $A$
to include all (oriented) induction hypotheses (represented by rewrite rules), while the latter is
designed to keep termination requirements as mild as possible; for example by 
orienting some requirements
using a relation $\succeq$ rather than 
$\succ$.
However, the latter approach imposes more bureaucracy, since derivation rules rely on several steps
being done at once -- for example, by reasoning \emph{modulo} the set of induction hypotheses.
This makes it quite hard to use especially when 
$\succ$ is not fixed beforehand but
rather constructed on the fly.

In \autoref{sec:boundedRI}, we will introduce Bounded RI, which aims to combine the best of both
worlds. 
We reduce termination requirements by using a pair $(\succ,\succeq)$, which may either be fixed in advance, or
constructed as part of the proof process.
In particular, we do not orient induction hypotheses themselves: we only require that a particular instance of the induction hypothesis is allowed to be applied whenever it is strictly dominated by some efficiently chosen bounding term, being associated to the particular equation under consideration.

Importantly, we do not impose the ground totality requirement on $\succ$ (which would be extremely restrictive
in higher-order rewriting!), and thus allow for $\succ$ to for instance be a relation $(\to_{A}
\cup\ \supterm)^+$, or built by a construction based on dependency pairs (see~\autoref{sec:howToFindOrdering}).

First, we will define the properties that our pair $(\succ,\succeq)$ should satisfy:

\begin{defi}[Ordering and Bounding Pair]\label{defi:ord-pair}\label{def:boundpair}
For a fixed set \( \mathcal{A} \),
an \emph{ordering pair} is a pair \( (\succ,\succeq) \) of a \emph{well-founded partial
ordering} \( \succ \) on \( \mathcal{A} \) (that is, \( \succ \) is a transitive, anti-symmetric,
irreflexive and well-founded relation) and a \emph{quasi-order} \( \succeq \) (that is,
\( \succeq \) is a transitive and reflexive relation) such that \( \succ\;\subseteq\;\succeq\), and
both \( \succ \cdot \succeq\;\subseteq\;\succ \) and \( \succeq \cdot \succ\;\subseteq\;\succ \)
(that is, for \( a,b,c \in \mathcal{A} \), \( a \succ b \succeq c \) and \( a \succeq b \succ c \)
imply \( a \succ c \)).
A \emph{bounding pair} for an \lcstrs\ \( \alcstrs = (\Sig, \rules) \) is an ordering pair
\( (\succ,\succeq) \) 
on \( T(\Sig,\emptyset) \) such that \( s \succeq t \) whenever \( s \rw t \)
or \( s \supterm t \).

We extend a bounding pair \( (\succ,\succeq) \) to non-ground terms with a constraint as follows.
A substitution $\gamma$ respects $s,t,\psi$ if $\Vars{s,t} \subseteq \domain(\gamma)$
and $\gamma$ respects $\psi$.  We define
\[
\begin{array}{rcl}
\csucc{s}{t}{\psi} & \Longleftrightarrow &
s \gamma \succ t \gamma
\text{ for all ground substitutions }\gamma\text{ that respect }s,t,\psi\\
\csucceq{s}{t}{\psi} & \Longleftrightarrow &
s \gamma \succeq t \gamma
\text{ for all ground substitutions }\gamma\text{ that respect }s,t,\psi
\end{array}
\]
\end{defi}

\subsection{Bounded Rewriting Induction}\label{sec:boundedRI}
Traditionally, RI is a deduction system on proof states, which are pairs 
$\pfst{\eqs}{\hs}$, where (in the existing literature),
$\eqs$ is a set of equations, describing all proof goals, and $\hs$ is the set
of induction hypotheses that have been assumed.  At the start $\eqs$ consists of
all equations that we want to prove to be inductive theorems, and
$\hs = \emptyset$.
With a deduction rule we may transform a proof state $(\eqs, \hs)$ into another
proof state $(\eqs', \hs')$, denoted $(\eqs, \hs) \vdash (\eqs', \hs')$. 
We write $\vdash^*$ for the reflexive-transitive closure of $\vdash$.  
Soundness of RI is guaranteed by the following principle:
``If $(\eqs, \emptyset) \vdash^* (\emptyset, \hs)$ for some set $\hs$, then every equation in $\eqs$ is an inductive theorem''.
Thus, we aim to make $\eqs$ empty.

In Bounded RI, we deviate in one respect: instead of letting $\eqs$ be a set of equations, we will
use a set of \emph{equation contexts}.  This new notion lets us avoid the bureaucracy of
combining steps by keeping track of two \emph{bounding terms}, which together really dictate a bound: we are only allowed to apply induction hypotheses below this bound.

\begin{defi}[Equation context]\label{defi:eq-cont}
Let \( (\succ,\succeq) \) be a bounding pair.
Let $\bullet$ be a fresh symbol.\linebreak
Define $\bullet \succeq \bullet$, and $\bullet \succ s$ and $\bullet \succeq s$ for all
$s \in \Terms$.
An \emph{equation context} $\eqcon{\sterm}{s}{t}{\tterm}{\psi}$
is a tuple of two \emph{bounding terms} $\sterm,\tterm \in \Terms \cup
\{\bullet\}$, two terms $s,t$ and a constraint $\psi$. 
A substitution \( \gamma \) \emph{respects} 
\( \eqcon{\sterm}{s}{t}{\tterm}{\psi} \) if \( \gamma \) respects \( \psi \)
and $\Vars{\sterm,s,t,\tterm} \subseteq \domain(\gamma)$.
\end{defi}

An equation context couples an equation with a bound on the induction: when considering an
equation context \( \eqcon{\sterm}{s}{t}{\tterm}{\psi} \) and a ground substitution \( \gamma \)
that respects the context, we implicitly use the induction hypothesis: ``all ground instances of
an equation in \( \hs \) that are strictly smaller than \( \sterm\gamma \approx \tterm\gamma \) are
convertible''.  Or phrased more precisely: for all \( \ell \approx r\ [\varphi] \in \hs \), all
ground substitutions \( \delta \) that respect this equation: if $\{\sterm\gamma,\tterm\gamma\}
\succ_{\mathrm{mul}} \{\ell\delta,r\delta\}$ then $\rweq{\ell\delta}{r\delta}{\rules}$
(using the multiset extension of $\succ$).

For example, in a proof that two implementations of the factorial function are equivalent,
we may wish to apply an induction hypothesis \( \symb{fact}_1\ m \approx \symb{fact}_2\ m\ [m \geq 0]
\) to an equation context \( \eqcon{\symb{fact}_1\ n}{\symb{fact}_1\ k}{\symb{fact}_2\ k}{
\symb{fact}_2\ n}{n > 0 \wedge n = k + 1} \).  
For an appropriately chosen \( \succ \) we have both
\( \csucc{\symb{fact}_1\ n}{\symb{fact}_1\ k}{n > 0 \wedge n = k + 1} \) and
\( \csucc{\symb{fact}_2\ n}{\symb{fact}_2\ k}{n > 0 \wedge n = k + 1} \).
Since moreover \( n > \symb{0} \wedge n = k + \symb{1} \) implies \( k \geq \symb{0} \),
we can apply the induction hypothesis: for any ground substitution \( \gamma \) that respects
the equation context, \( \rweq{(\symb{fact}_1\ k)\gamma}{(\symb{fact}_2\ k)\gamma}{\rules} \).

\begin{defi}[Proof state]
A proof state is a tuple $(\eqs, \hs)$  with $\eqs$ a set of equation contexts and $\hs$ a set of equations. 
\end{defi}  

From \autoref{defi:eq-cont} we see that $\bullet$ acts as an infinity-term with respect to $\succ$ and $\succeq$. 
As expressed below in \autoref{theorem:soundnessRI}: when using Bounded RI to prove a set of equations,
we put them into a set $\eqs$ of equation contexts using infinite bounds $\sterm = \tterm = \bullet$. 
The bounds remain infinite until we start an induction process, at which point the equation is added to
the proof state as an induction hypothesis, and the bounds are lowered accordingly.

\smallskip
The derivation rules for Bounded RI are provided in \autoref{fig:boundedRIrules}, and will be
explained in detail in \autoref{sec:example}.  The following notation is used by several derivation
rules:

\begin{defi}[$\models^\delta$]\label{par:ConstrainedReductions}
Let $\delta$ be a substitution and $\varphi$, $\psi$ be constraints. 
We write \(\psi \models^\delta \varphi\)
if $\delta(\Vars{\varphi}) \subseteq \Val \cup \Vars{\psi}$, and
$\psi \Longrightarrow \varphi\delta$ is a valid constraint.
\end{defi} 


\newcommand{\NAMERULE}[1]{\begin{flushleft}\textbf{#1}\end{flushleft}}
\newcommand{\DEDUCRULE}[3]{
  \[
  \AxiomC{$#1$}
  \RightLabel{\scriptsize\quad \(\begin{aligned}&#2\end{aligned}\)}
  \UnaryInfC{$#3$}
  \DisplayProof
  \]}

\newcommand{\RULEsimplify}{
  \DEDUCRULE{(\eqs \uplus \{\eqconsim{\sterm}{C[\ell\delta]}{t}{\tterm}{\psi}\}, \hs)}
            {\ell \to r\ [\varphi] \in \rules \cup \calcrules \text{ and } \psi \models^\delta \varphi}
            {(\eqs \cup \{\eqconsim{\sterm}{C[r \delta]}{t}{\tterm}{\psi}\}, \hs)}
}

\newcommand{\RULEcasebase}[1]{
  \DEDUCRULE{(\eqs \uplus \{\eqcon{\sterm}{s}{t}{\tterm}{\psi}\}, \hs)}
            {\coverset \text{ a cover set of } s \approx t\ [\psi] #1}
            {(\eqs \cup \{\eqcon{\sterm\delta}{s\delta}{t\delta}{\tterm\delta}{\psi\delta \wedge \varphi} \mid (\delta,\varphi) \in \coverset\}, \hs)}
}

\newcommand{\RULEcase}{\RULEcasebase{\\ &\text{(see \autoref{def:coverset})}}}
\newcommand{\RULEcasenoref}{\RULEcasebase{}}

\newcommand{\RULEdelete}{
  \DEDUCRULE{(\eqs \uplus \{\eqcon{\sterm}{s}{t}{\tterm}{\psi}\}, \hs)}
            {\psi \text{ unsatisfiable, or }s=t}
            {(\eqs, \hs)}
}

\newcommand{\RULEsemicons}{
  \DEDUCRULE{(\eqs \uplus \{ \eqcon{\sterm}{f\ s_1 \cdots s_n}{f\ t_1 \cdots t_n}{\tterm}{\psi} \}, \hs)}
            {n > 0 \text{ and } (f \in \Var/n<\arity(f))}
            {(\eqs \cup \left\{ \eqcon{\sterm}{s_i}{t_i}{\tterm}{\psi} \mid 1 \le i \le n \right\}, \hs)}
}

\newcommand{\RULEcontextsymb}{
  \DEDUCRULE{(\eqs \uplus \{ \eqcon{\sterm}{\afun\ s_1 \cdots s_n}{\afun\ t_1 \cdots t_n}{\tterm}{\psi} \}, \hs)}
            {n > 0 \text{ and } \afun \in \Sig \text{ and } n<\arity(f))}
            {(\eqs \cup \left\{ \eqcon{\sterm}{s_i}{t_i}{\tterm}{\psi} \mid 1 \le i \le n \right\}, \hs)}
}

\newcommand{\RULEcontextvar}{
  \DEDUCRULE{(\eqs \uplus \{ \eqcon{\sterm}{f\ s_1 \cdots s_n}{f\ t_1 \cdots t_n}{\tterm}{\psi} \}, \hs)}
            {n > 0 \text{ and } f \in \Var}
            {(\eqs \cup \left\{ \eqcon{\sterm}{s_i}{t_i}{\tterm}{\psi} \mid 1 \le i \le n \right\}, \hs)}
}

\newcommand{\RULEinductbounded}{
  \DEDUCRULE{(\eqs \uplus \{\eqcon{\sterm}{s}{t}{\tterm}{\psi}\}, \hs)}
            {}
            {(\eqs \cup \{\eqcon{s}{s}{t}{t}{\psi}\}, \hs \cup \{s \approx t\ [\psi]\})}
}

\newcommand{\RULEinductgeneral}{
  \DEDUCRULE{(\eqs \uplus \{\eqcon{\sterm}{s}{t}{\tterm}{\psi}\}, \hs)}
            {\multiset{\sterm,\tterm} \succeqmul \multiset{s,t}\ [\psi]}
            {(\eqs \cup \{\eqcon{s}{s}{t}{t}{\psi}\}, \hs \cup \{s \approx t\ [\psi]\})}
}

\newcommand{\RULEhypothesisbounded}{
  \DEDUCRULE{(\eqs \uplus \{\eqconsim{\sterm}{C[\ell\delta]}{t}{\tterm}{\psi}\}, \hs)}
            {(\ell \approx r\ [\varphi]\in \hs \text{ or } r \approx \ell\ [\varphi] \in \hs) \text{ and } \psi \models^\delta \varphi \text{ and} \\
              &\csucceq{\sterm}{C[r\delta]}{\psi} \text{ and } (\exists \pi \in \{\sterm,\tterm\}.\csucc{\pi}{\ell\delta}{\psi} \text{ and } \csucc{\pi}{r\delta}{\psi})}
            {(\eqs \cup \{\eqconsim{\sterm}{C[r \delta]}{t}{\tterm}{\psi}\}, \hs)}
}

\newcommand{\RULEhypothesisgeneral}{
  \DEDUCRULE{(\eqs \uplus \{\eqconsim{\sterm}{C[\ell\delta]}{t}{\tterm}{\psi}\}, \hs)}
            {(\ell \simeq r\ [\varphi]\in \hs \text{ or } r \approx \ell\ [\varphi] \in \hs) \text{ and } \psi \models^\delta \varphi \text{ and} \\
              &\multiset{\sterm,\tterm} \succmul \multiset{\ell\delta,r\delta}\ [\psi]}
            {(\eqs \cup \{\eqconsim{\sterm}{C[r \delta]}{t}{\tterm}{\psi}\}, \hs)}
}

\newcommand{\RULEhdelete}{
  \DEDUCRULE{(\eqs \uplus \{\eqconsim{\sterm}{C[\ell\delta]}{C[r\delta]}{\tterm}{\psi}\}, \hs)}
            {\ell \approx r\ [\varphi]\in \hs \text{ and } \psi \models^\delta \varphi \text{ and} \\
              &\csucc{\sterm}{\ell\delta}{\psi} \text{ or } \csucc{\tterm}{r\delta}{\psi}}
            {(\eqs, \hs)}
}

\newcommand{\RULEgeneralizealterbase}[1]{
  \DEDUCRULE{(\eqs \uplus \{\eqcon{\sterm}{s}{t}{\tterm}{\psi}\}, \hs)}
            {\eqcon{\sterm^\prime}{s^\prime}{t^\prime}{\tterm^\prime}{\varphi} \text{ #1 } \\
             &\eqcon{\sterm}{s}{t}{\tterm}{\psi} \text{ (see \autoref{def:generalize})}}
            {(\eqs \cup \{\eqcon{\sterm^\prime}{s^\prime}{t^\prime}{\tterm^\prime}{\varphi}\}, \hs)}
}

\newcommand{\RULEgeneralizealter}{\RULEgeneralizealterbase{generalizes/alters}}

\newcommand{\RULEgeneralizealterboundedbase}[2]{
  \DEDUCRULE{(\eqs \uplus \{\eqcon{\sterm}{s}{t}{\tterm}{\psi}\}, \hs)}
            {\eqcon{\sterm^\prime}{s^\prime}{t^\prime}{\tterm^\prime}{\psi^\prime} \text{ #1 } \eqcon{\sterm}{s}{t}{\tterm}{\psi} \\
             &#2
              \csucceq{\sterm^\prime}{s^\prime}{\psi^\prime} \text{ and } \csucceq{\tterm^\prime}{t^\prime}{\psi^\prime}}
            {(\eqs \cup \{\eqcon{\sterm^\prime}{s^\prime}{t^\prime}{\tterm^\prime}{\psi^\prime}\}, \hs)}
}

\newcommand{\RULEgeneralizealterbounded}{\RULEgeneralizealterboundedbase{generalizes/alters}{\text{(see \autoref{def:generalize}), and }}}

\newcommand{\RULEpostulate}{
  \DEDUCRULE{(\eqs,\hs)}
            {}
            {(\eqs \cup \{\eqcon{\bullet}{s}{t}{\bullet}{\psi}\},\hs)}
}

\begin{figure}[tp] 
\caption{Derivation rules for Bounded Rewriting Induction, given a bounding pair $(\succ,\succeq)$.}
\label{fig:boundedRIrules}
\noindent\fbox{\begin{minipage}{\textwidth}
\NAMERULE{\simplify}\RULEsimplify\smallskip
\NAMERULE{\case}\RULEcase\smallskip
\NAMERULE{\delete}\RULEdelete\smallskip
\NAMERULE{\headvarmarked/\semicons}\RULEsemicons\smallskip
\NAMERULE{\induct}\RULEinductbounded\smallskip
\NAMERULE{\hypothesis}\RULEhypothesisbounded\smallskip
\NAMERULE{\hdelete}\RULEhdelete\smallskip
\NAMERULE{\generalizemarked/\alter}\RULEgeneralizealterbounded\smallskip
\NAMERULE{\postulatemarked}\RULEpostulate
\end{minipage}}
\end{figure}


Hence, if \( \psi \models^\delta \varphi \) and \( \gamma \) is a substitution that respects
\( \psi \), then \( \gamma \) also respects \( \varphi\delta \).
This is for example used in the deduction rule \simplify\ to ensure that every ground instance of an equation
$C[\ell\delta] \approx t\ [\psi]$ can be reduced with a rewrite rule
$\ell \to r\ [\varphi]$.

\begin{exa}
Consider the equation $\recdown\ f\ i'\ i\ a \approx t\ [i'=i+1]$ and rewrite rule 
$\recdown\ f\ n\ i\ a \to a \ [i<n]$, and let $\delta=[n:=i']$.
We have $(i<n) \delta=i<i'$ and $(i'=i+1) \models^\delta (i<n)$ holds, because
\begin{itemize}
[label = $\triangleright$]
\item $\delta(\Vars{i<n})=\delta(\{i,n\})=\{i,i'\}\subseteq \Vars{i'=i+1}=\{i, i'\}$
\item $(i'=i+1) \Longrightarrow (i<n) \delta$ is valid 
\end{itemize}   
\end{exa} 

The notation $\eqconsim{\sterm}{s}{t}{\tterm}{\psi}$ in the \simplify, \hypothesis\ and \hdelete\ 
rules denotes either an equation context $\eqcon{\sterm}{s}{t}{\tterm}{\psi}$ or an
equation context $\eqcon{\tterm}{t}{s}{\sterm}{\psi}$.
This is used to indicate that a reduction (using a rewrite rule or an induction hypothesis) can be
done on either side of the equation.  The orientation of the equation context in the conclusion, if
any, is intended to be the same as in the premise.  (Although it does not really matter: all
derivation rules treat both sides of the equation context the same.)

The marker $\circledast$ indicates \emph{incomplete} rules, as will be explained in
\autoref{subsec:completeness}.

The following theorem, which is proved in \autoref{subsec:soundnessProofBounded}, shows how to
use Bounded RI.

\newcommand{\theoremsoundnessRI}{
  Let $\alcstrs$ be a weakly normalizing, quasi-reductive \lcstrs,
  let $\mathcal{A}$ be a set of equations,
  and let \( \eqs \) be the set of equation contexts
  \( \{ \eqcon{\bullet}{s}{t}{\bullet}{\psi} \mid s \approx t\ [\psi] \in \mathcal{A} \} \).
  Suppose  $(\eqs,\emptyset) \vdash^* (\emptyset, \hs)$, for some $\hs$ using the derivation
  rules in \autoref{fig:boundedRIrules} with some bounding pair $(\succ, \succeq)$.
  Then every equation in $\mathcal{A}$ is an inductive theorem.
}
\begin{thm}[Soundness of Bounded RI]\label{theorem:soundnessRI}
\theoremsoundnessRI
\end{thm}

\begin{rem}
To avoid confusion we will clarify what we mean by soundness. 
Those familiar with equational reasoning may expect a statement expressing the implication 
\[
(\star)\quad \quad 
\text{RI-provability} \Longrightarrow \text{equality in every possible model of }\rules
\]
This is \emph{not} what \autoref{theorem:soundnessRI} refers to. 
Our notion of soundness does not consider all possible models, but rather one
particular model: we fix our semantics to inductive theorems, i.e.
$\rules$-ground convertibility. 
Thus, 
\emph{soundness} should be interpreted 
as
the following implication 
\[
\text{RI-provability} \Longrightarrow \rules\text{-ground convertibility}
\]
Those interested in a more semantical discussion of RI -- and in particular the 
implication  $(\star)$ -- may consider the recent
publication~\cite{aot:nis:sch:24}, which studies this question for \lctrss.
\end{rem}

While the derivation rules (and \autoref{theorem:soundnessRI}) assume a fixed bounding pair as
given, in practice we can leave this pair undecided, use the proof system to collect requirements
about $(\succ,\succeq)$, and then select a suitable bounding pair at the end.
Note that, unlike earlier definitions of RI, we do not require $\alcstrs$\ to be
terminating; only \emph{weakly} normalizing.  This is possible due to our use of bounding
terms.  However, as we will see in the examples, in practice it is often \emph{easier} to apply
the derivation rules if $\rw\;\subseteq\;\succ$.

\smallskip
As we will see in \autoref{sec:bounded-eq-ct}, the proof system in \autoref{fig:boundedRIrules}
satisfies a property called \emph{Preserving Bounds}, meaning that it only produces a specific kind of equation contexts:

\begin{defi}[Bounded equation context]\label{defi:bounded-eq-cont}
A \emph{bounded equation context} is an equation context $\eqcon{\sterm}{s}{t}{\tterm}{\psi}$ such that both
$\csucceq{\sterm}{s}{\psi}$ and $\csucceq{\tterm}{t}{\psi}$.
\end{defi}
This property plays an important role in the soundness proof of Bounded RI,
since it allows us to avoid ordering requirements that would otherwise
need to be added to \autoref{fig:boundedRIrules}.

In many practical cases, we will also consider an even stronger restriction:
\begin{defi}[Strongly bounded equation context]\label{defi:strongly-bounded-eq-cont}
An equation context $\eqcon{\sterm}{s}{t}{\tterm}{\psi}$ is \emph{strongly
bounded} if either
$\sterm = s$ or $\csucc{\sterm}{s}{\psi}$, and also either
$\tterm = t$ or $\csucc{\tterm}{t}{\psi}$.
\end{defi}

Clearly, all equation contexts at the start of a RI deduction are strongly
bounded, as all initial equation contexts have \( \sterm = \tterm = \bullet \).
At first it may not be clear why \emph{strongly bounded} is a strictly stronger
property than \emph{bounded}.  To illustrate this, consider the following example
with a well-founded ordering $\succ$ whose restriction to integers is defined as
\[
n \succ m 
\Longleftrightarrow
n>m \wedge n \ge \symb{0}
\]
Then $n \succeq m\ [n \ge m \wedge n \ge \symb{0}]$ holds, but
we have neither $n=m$ nor $n \succ m\ [n \ge m \wedge n \ge \symb{0}]$.

In practice, we highly benefit from strong boundedness because many of
the required inequalities in \autoref{fig:boundedRIrules} can be established by
a straightforward syntactic check: whenever $\sterm \ne s$, we immediately
conclude $\sterm \succ s\ [\psi]$.  While strong boundedness is \emph{not}
necessarily preserved by the rules of \autoref{fig:boundedRIrules}, we can
maintain it by a suitable choice of bounding pair and careful application of
the \hypothesis, \generalize, and \alter\ rules.

\subsection{Explanation by example}\label{sec:example}

We will now elaborate on the rules of \autoref{fig:boundedRIrules}, and
illustrate their use through examples.  We will also introduce the definitions
of \emph{cover set} and \emph{generalizing} / \emph{altering}, which are used in
\autoref{fig:boundedRIrules}.
To start, we will consider the \lcstrs\ from \autoref{example:runningExample} applied on the
equation $\recdown\ f\ n\ i\ a \approx \tailup \ f\ n\ i\ a$.

Following \autoref{theorem:soundnessRI}, we aim to show that there is a set $\hs$ such that 
\[
    (\eqs_1, \emptyset)
    \vdash^* 
    (\emptyset, \hs)
    \ 
    \text{ with }
    \eqs_1 
    = 
    \{
    \eqcon{\bullet}{\recdown\ f\ n\ i\ a}{\tailup \ f\ n\ i\ a}{\bullet}{\strue}
    \} 
\]
We do \emph{not} fully choose the bounding pair \( (\succ,\succeq) \) in advance.  Rather, we will
use the proof process to accumulate requirements on the well-founded ordering \( \succ \) to be
used. 
However, we precommit to a bounding pair
where
  \( \succeq \) is the reflexive closure of \( \succ \),
and
  \( s \succ t \) whenever \( s \rw t \) or \( s \supterm t \).
With these assumptions, it is easy to maintain the property that all equation
contexts are \emph{strongly bounded}.  We will ensure this invariant throughout
our RI proof.

\paragraph*{\induct} 
We introduce the following way to start an induction proof.

\RULEinductbounded

\smallskip
Here, two things happen.
First, 
the current equation is added to the set
\( \hs \) of induction hypotheses, making it available for later application of
\hypothesis\ or \hdelete.
Second, the bounding terms \( \sterm,\tterm \) are replaced by \( s,t \).
This ensures that, when an induction hypothesis is applied, it is only on equations that are
strictly smaller than \( s \approx t\ [\psi] \).

\medskip
In our running example, we use \induct\ to obtain \( (\eqs_1, \emptyset) \vdash (\eqs_2, \hs_2) \) where
\[
\begin{array}{rcl}
\eqs_2 & = & \{ \eqcon{\sterm_2}{\recdown\ f\ n\ i\ a}{\tailup \ f\ n\ i\ a}{\tterm_2}{\strue} \} \\
\hs_2 & = & \{ \recdown\ f\ n\ i\ a \approx \tailup \ f\ n\ i\ a \} \\
\end{array}
\]
We will retain \( \sterm_2 = \recdown\ f\ n\ i\ a \) and \( \tterm_2 = \tailup \ f\ n\ i\ a \), for later reference.

\paragraph{\case}
Comparing the equation in $\eqs_2$ to $\rules$, which of the rules should we apply? 
As we will see in \simplify, this requires information about how the variables $i,n$ in the equation are instantiated;   
in this example the reduction behavior depends on which of the two cases, $i<n$ or $i \ge n$, holds. 
This is where \case\ can help us, splitting an equation into multiple cases.
Of course, we have to make sure that the cases together cover the original equation.     

\begin{defi}[Cover set]\label{def:coverset}
A cover set of an equation \( s \approx t\ [\psi] \) is a set \( \coverset \)
of pairs \( (\delta,\varphi) \) where
$\delta$ is a substitution with $\delta(\avar) \subseteq \Var \cup \Val$ for $\avar \in
\Vars{\psi}$ and $\varphi$ is a constraint, such that
for every gsc substitution \( \gamma \) respecting \( s \approx t\ [\psi] \),
there exists \( (\delta,\varphi) \in \coverset \)
and a substitution $\sigma$ such that
\( \gamma(x) = \delta(x)\sigma \) for all \( x \in \domain(\gamma) \), and
\( \sigma \) respects \( \psi\delta \wedge \varphi \). 
(Hence, \( s\gamma \approx t\gamma \) is an instance of
\( s\delta \approx t\delta\ [\psi\delta \wedge \varphi] \).)
\end{defi}

Now, the deduction rule \case\ reads as follows:
\RULEcasenoref

Continuing our example, we observe that the only gsc terms of type \( \int \)
are values (since we have inextensible theory sorts).  Hence, $\coverset 
= 
\{
    ([], i < n),\ 
    ([], i \ge n)
\}$
is a cover set of $\recdown\ f\ n\ i\ a \approx \tailup \ f\ n\ i\ a$.
Using \case, we obtain  $(\eqs_2,\hs_2) \vdash (\eqs_3,\hs_2) $ with  
\[
\eqs_3 = 
\left\{
\begin{aligned}
\eqcon{\sterm_2}{\recdown\ f\ n\ i\ a}{\tailup \ f\ n\ i\ a}{\tterm_2}{i < n} \\
\eqcon{\sterm_2}{\recdown\ f\ n\ i\ a}{\tailup \ f\ n\ i\ a}{\tterm_2}{i \ge n} \\
\end{aligned}
\right\}
\]
Note that the bounding terms \( \sterm_2,\tterm_2 \) are unchanged because the
substitutions in the cover set were both empty.  This is, however,  not true in general
(and we will see an alternative situation later in this section).
Strong boundedness is still satisfied, by \( \sterm = s \) and \( \tterm = t \).

\paragraph{\simplify}
Next,
we can use a rule $\myrule$ to rewrite 
$C[\ell \delta] \approx t\ [\psi]$ or $t \approx C[\ell\delta]\ [\psi]$.

\RULEsimplify
The requirement $\psi \models^\delta \varphi$ makes sure that the $\delta$-instance of $\myrule$ is actually applicable.  
The bounding terms are not affected by the reduction.

\medskip
Continuing our example,
the first equation in $\eqs_3$ has constraint $i<n$, so we apply \simplify\ on both sides, using $\textbf{(R1)}$ and $\textbf{(R3)}$.
For the second equation, we also apply \simplify\ to both sides, using $\textbf{(R2)}$ and $\textbf{(R4)}$. 
We obtain $\deduces{(\eqs_3,\hs_2)}{(\eqs_4,\hs_2)}$ with
\[
\eqs_4 = 
\left\{
\begin{array}{cl}
\ \ \ \eqconwc{\sterm_2}{a}{a}{\tterm_2} & [i < n] \\
\eqconwc{\sterm_2}{f\ i\ (\recdown\ f\ n\ (i-\symb{1})\ a)}{\tailup \ f\ (n+\symb{1})\ i\ (f\ n\ a)}{\tterm_2} & [i \ge n] \\
\end{array}
\right\}
\]
Note that for both equations \( \eqcon{\sterm_2}{s}{t}{\tterm_2}{\psi} \) we now have the property:
\( \csucc{\sterm_2}{s}{\psi} \) and \( \csucc{\tterm_2}{t}{\psi} \), since
\( \rw \) is included in \( \succ \).  Thus, strong boundedness is preserved.

\paragraph{\delete}
The following deduction rule allows us to remove an equation that has an unsatisfiable constraint, or whose two sides are syntactically equal. 
\smallskip 
\RULEdelete

\smallskip 
In our example, we use \delete\ to obtain 
$(\eqs_4,\hs_2) \vdash (\eqs_5, \hs_2)$ with 
\[
\eqs_5 = \left\{
\eqcon{\sterm_2}{f\ i\ (\recdown\ f\ n\ (i-\symb{1})\ a)}{\tailup \ f\ (n+\symb{1})\ i\ (f\ n\ a)}{\tterm_2}{i \ge n}
\right\}
\]

\paragraph{\alter}
It is often useful to rewrite an equation context to another that might be syntactically
different, but has the same ground semi-constructor instances.
Indeed, this may even be necessary, for instance to support the application of a
rewrite rule through \simplify.  This is supported by the deduction rule \alter, which relies on
the following definition:

\begin{defi}\label{def:generalize}
We say that an equation context \( \eqcon{\sterm'}{s'}{t'}{\tterm'}{\psi^\prime} \)
\begin{itemize}[label = $\triangleright$]
\item \emph{generalizes} \( \eqcon{\sterm}{s}{t}{\tterm}{\psi} \) if for every gsc substitution $\gamma$ that respects
  \( \eqcon{\sterm}{s}{t}{\tterm}{\psi} \) there is a ground substitution $\delta$ that respects \( \eqcon{\sterm'}{s'}{t'}{\tterm'}{\psi^\prime} \)
  such that $s\gamma = s'\delta$ and $t\gamma = t'\delta$, and $\sterm\gamma \succeq \sterm'\delta$ and $\tterm\gamma \succeq \tterm'\delta$.
\item \emph{alters} \( \eqcon{\sterm}{s}{t}{\tterm}{\psi} \) if both
    \( \eqcon{\sterm'}{s'}{t'}{\tterm'}{\psi^\prime} \) generalizes \( \eqcon{\sterm}{s}{t}{\tterm}{\psi} \), and
    \( \eqcon{\sterm}{s}{t}{\tterm}{\psi} \) generalizes \( \eqcon{\sterm'}{s'}{t'}{\tterm'}{\psi^\prime} \).  
\end{itemize} 
\end{defi}

\noindent
Now, \alter\ is defined as follows
\newcommand{\RULEalterbounded}{\RULEgeneralizealterboundedbase{alters}{\text{and }}}

\RULEalterbounded

\smallskip
Following the discussion in~\cite{hag:kop:24}, we will use this rule in two primary ways:
\begin{enumerate}[I]
\item\label{alter:equisat} Replacing a constraint by
  a mostly equivalent one:
  we may alter \( \eqcon{\sterm}{s}{t}{\tterm}{\psi} \) into
  \( \eqcon{\sterm}{s}{t}{\tterm}{\psi^\prime} \) if
  \( (\exists \vec{x}. \psi) \Longleftrightarrow (\exists \vec{y}.\psi^\prime) \) is logically valid,
  where \( \{ \vec{x} \} = \Vars{\psi} \setminus \Vars{\sterm,s,t,\tterm} \) and
  \( \{ \vec{y} \} = \Vars{\psi^\prime} \setminus \Vars{\sterm,s,t,\tterm} \).
  If the system does not have inextensible theory sorts, we additionally require that
  $\Vars{\psi} \setminus \{\vec{x}\} = \Vars{\psi'} \setminus \{\vec{y}\}$.\footnote{This is because
    adding a variable into the constraint forces it to be instantiated only by values.  If the sort
    admits other ground constructor terms as instances, this would go against the definition of
    ``generalizes''.}

  A particular example of this case is to replace
  \( \eqcon{\sterm}{s}{t}{\tterm}{\psi} \) by
  \( \eqcon{\sterm}{s}{t}{\tterm}{\psi \wedge \linebreak x_1 = u_1 \wedge \dots \wedge x_n = u_n} \), where
  all \( x_i \) are fresh variables, with \( x_i \notin \Vars{u_j} \) for \( j \leq i \).
\item\label{alter:substitute} Replacing variables by equivalent variables or values: that is, altering
  \( \eqcon{\sterm}{s}{t}{\tterm}{\psi} \) into
  \( \eqcon{\sterm}{s\gamma}{t\gamma}{\tterm}{\psi} \) if
  \( \gamma = [x_1:=u_1,\dots,x_n:=u_n] \) and \( \psi \Longrightarrow x_1 = u_1 \wedge \dots \wedge
  x_n = u_n \) is valid, where \( x_1,\dots,x_n \) are variables and \( u_1,\dots,u_n \) are each
  variables or values.
\end{enumerate}

\newcommand{\shortenedconstr}{\varphi_{i',n'}}

\medskip
Back to our example!  By case \ref{alter:equisat} above, we can use \alter\ to obtain
$(\eqs_5,\hs_2) \vdash (\eqs_6, \hs_2)$:
\[
\eqs_6 = \left\{
\begin{array}{l}
\eqconwc{\sterm_2}{f\ i\ (\recdown\ f\ n\ (i-\symb{1})\ a)}{\tailup \ f\ (n+\symb{1})\ i\ (f\ n\ a)}{\tterm_2}\quad\quad \\
{}\hfill{}[i' = i - \symb{1} \wedge n' = n + \symb{1} \wedge i \ge n] \\
\end{array}
\right\}
\]
To allow this rule to be applied, we must have \( \csucceq{\sterm_2}{f\ i\ (\recdown\ f\ n\ (i-\symb{1})\ a)}{\varphi} \)
and \( \csucceq{\tterm_2}{\tailup \ f\ (n+\symb{1})\ i\ (f\ n\ a)}{\varphi} \) where \( \varphi \) is the constraint
\( i' = i - \symb{1} \wedge n' = n + \symb{1} \wedge i \ge n \).  But this follows immediately from the fact that our
previous equation context was strongly bounded:
if \( \csucc{\sterm}{s}{i \ge n} \) then also \( \csucc{\sterm}{s}{i' = i - \symb{1} \wedge n' = n + \symb{1} \wedge
i \ge n} \), and similar for \( \csucc{\tterm}{t}{i \ge n} \).  Moreover, strong boundedness is clearly still satisfied.

For brevity, we will write \( \shortenedconstr \) instead of \( i' = i - \symb{1} \wedge n' = n + \symb{1} \)
in the remainder of the example.  Hence, the constraint above will be written as \( \shortenedconstr \wedge i \ge n \).

\begin{rem}\label{rem:alter}
The preservation of strong boundedness in our example is not a coincidence:
as we will see in \autoref{lem:stronglyBoundedPreserve}, most of the deduction
rules of \autoref{fig:boundedRIrules} preserve this property automatically if
$\succ$ includes $\rw$ and $\supterm$.
In \alter\ and \generalize\ this is not necessarily the case, but it is when
\alter\ is used in either way \ref{alter:equisat} or
\ref{alter:substitute} above (see \autoref{lem:altersafe}).
\end{rem}

We continue the example by two successive \simplify\ steps, using calculation rules $x-y\to z\ [z = x-y]$ and $x+y \to z\ [z = x+y]$ 
to obtain \( (\eqs_7,\hs_2) \), with 
\[
\eqs_7 = \left\{
\begin{array}{l}
\eqcon{\sterm_2}{f\ i\ (\recdown\ f\ n\ i'\ a)}{\tailup \ f\ n'\ i\ (f\ n\ a)}{\tterm_2}{\shortenedconstr \wedge i \ge n} \\
\end{array}
\right\}
\]
Note that these steps were only possible because of the \alter\ step that preceded them.

(Again, strong boundedness is preserved because \( \succ \) includes \( \rw \).)

\paragraph{\hypothesis} Similar to \simplify, we can use an induction hypothesis to reduce either
side of an equation.  Here, finally, the bounding terms \( \sterm,\tterm \) come into play, as we
need to make sure that we have a decrease of some kind, to apply induction.
\smallskip 
\RULEhypothesisbounded

\medskip
We use \hypothesis\ to reduce the lhs of the equation in \( \eqs_7 \) with the only induction
hypothesis from \( \hs_2 \)
in the direction $\recdown\ f\ n\ i\ a\ \to\  \tailup \ f\ n\ i\ a$,
with substitution \( [i:=i'] \).
This lets us deduce $(\eqs_7,\hs_2) \vdash (\eqs_8,\hs_2)$ with
\[
\eqs_8 = \left\{
\begin{array}{l}
\eqcon{\sterm_2}{f\ i\ (\tailup\ f\ n\ i'\ a)}{\tailup \ f\ n'\ i\ (f\ n\ a)}{\tterm_2}{\shortenedconstr \wedge i \ge n}
\end{array}
\right\}
\]
However, to apply this deduction rule, we must show that the \( \succ \) requirements are satisfied.
This means at least \innerenumA\ \( \recdown\ f\ n\ i\ a \succeq f\ i\ (\tailup\ f\ n\ i'\ a)\ 
[\shortenedconstr \wedge i \ge n] \), and also
\( \innerenumB\ \pi \succ \recdown\ f\ n\ i'\ a\ [\shortenedconstr \wedge i \ge n] \) and
\( \innerenumC\ \pi \succ \tailup\ f\ n\ i'\ a [\shortenedconstr \wedge i \ge n] \) for some
\( \pi \in \{\sterm_2,\tterm_2\} \).
Choosing \( \pi = \sterm_2 = \recdown\ f\ n\ i\ a \), \innerenumB\ is satisfied by the strong
boundnedness property, and \innerenumC\ is an immediate consequence of \innerenumA:
\( \f\ i\ (\tailup\ f\ n\ i'\ a) \supterm \tailup\ f\ n\ i'\ a \) and we have committed
to include \( \supterm \) in \( \succ \).

For \innerenumA, we remember that \textbf{(REQ1)} below still needs to be satisfied:
Here, we replace the $\succeq$ by a $\succ$ to ensure that strong
boundedness is preserved.  Since \( \succeq \) is the reflexive closure of \( \succ \),
this makes no difference: the only way to satisfy \innerenumA\ is to satisfy \textbf{(REQ1)}.
\[
\begin{array}{lrclcc}
\textbf{(REQ1)} & \recdown\ f\ n\ i\ a & \succ & f\ i\ (\tailup\ f\ n\ i'\ a) & [i' = i - \symb{1} \wedge n' = n + \symb{1} \wedge i \ge n] \\
\end{array}
\]

Let \( \sterm_9 = f\ i\ (\tailup\ f\ n\ i'\ a) \), \( \tterm_9 = \tailup \ f\ n'\ i\ (f\ n\ a) \),
and apply \induct\ to \( (\eqs_8,\hs_2) \):
\[
\eqs_9 = \left\{
\begin{array}{l}
\eqcon{\sterm_9}{f\ i\ (\tailup\ f\ n\ i'\ a)}{\tailup \ f\ n'\ i\ (f\ n\ a)}{\tterm_9}{\shortenedconstr \wedge i \ge n}
\end{array}
\right\}
\]
\[
\hs_3 
= 
\left\{
\begin{array}{rcll}
\recdown\ f\ n\ i\ a & \approx & \tailup \ f\ n\ i\ a \\
f\ i\ (\tailup\ f\ n\ i'\ a) & \approx & \tailup \ f\ n'\ i\ (f\ n\ a) & [\shortenedconstr \wedge i \ge n] \\
\end{array}
\right\}
\]
Next, we use \case\ again to split the constraint in \( \eqs_9 \) into \( i = n \) and \( i > n \), giving \( (\eqs_{10},\hs_3) \):
\[
\eqs_{10} = \left\{
\begin{array}{l}
\eqcon{\sterm_9}{f\ i\ (\tailup\ f\ n\ i'\ a)}{\tailup \ f\ n'\ i\ (f\ n\ a)}{\tterm_9}{\shortenedconstr \wedge i = n} \\
\eqcon{\sterm_9}{f\ i\ (\tailup\ f\ n\ i'\ a)}{\tailup \ f\ n'\ i\ (f\ n\ a)}{\tterm_9}{\shortenedconstr \wedge i > n} \\
\end{array}
\right\}
\]
Recall that \( \shortenedconstr \) stands for \( i' = i - \symb{1} \wedge n' = n + \symb{1} \).
Hence, \( \shortenedconstr \wedge i = n \) implies both \( n > i' \) and \( n' > i \), and
\( \shortenedconstr \wedge i > n \) implies both \( n \leq i' \) and \( n' \leq i \).
We use \simplify\ on both sides of the first equation with \textbf{(R3)} and on both sides of the
second equation with \textbf{(R4)}, to deduce \( \deduces{(\eqs_{10},\hs_3)}{(\eqs_{11},\hs_3)} \):
\[
\eqs_{11} = \left\{
\begin{array}{ll}
\eqcon{\sterm_9}{f\ i\ a}{f\ n\ a}{\tterm_9}{\shortenedconstr \wedge i = n} \\
\eqconwc{\sterm_9}{f\ i\ (\tailup\ f\ (n+1)\ i'\ (f\ n\ a))}{\tailup \ f\ (n'+1)\ i\ (f\ n'\ (f\ n\ a))}{\tterm_9} \\
{}\hfill{}[\shortenedconstr \wedge i > n] \\
\end{array}
\right\}
\]
Note that the first equation above does not satisfy the requirements for \delete, even though \( i = n \)
in the constraint makes it look very delete-worthy.  We resolve this by using \alter\ (case
\ref{alter:substitute}), replacing the first equation context by
\( \eqcon{\sterm_9}{f\ n\ a}{f\ n\ a}{\tterm_9}{\shortenedconstr \wedge i = n} \),
after which it can immediately be deleted.
Using \alter\ (case \ref{alter:equisat}) on the second equation, and then \simplify\ with
calculation rules as before, we are left with \( (\eqs_{12},\hs_3) \):

\[
\eqs_{12} = \left\{
\begin{array}{ll}
\eqconwc{\sterm_9}{f\ i\ (\tailup\ f\ n'\ i'\ (f\ n\ a))}{\tailup \ f\ n''\ i\ (f\ n'\ (f\ n\ a))}{\tterm_9} \\
{}\hfill{}[\shortenedconstr \wedge n'' = n' + \symb{1} \wedge i > n] \\
\end{array}
\right\}
\]

\paragraph{\hdelete.} With this deduction rule we may delete any equation that has a subequation which is an instance of an equation in $\hs$. 
The rule looks very similar to \hypothesis, but with lighter requirements on \( \succ \). 
\smallskip 
\RULEhdelete

\medskip
Consider our example.  Renaming the variables to avoid confusion, the second
hypothesis in \( \hs_3 \) reads:
\( g\ x\ (\tailup\ g\ y\ x'\ z) \approx \tailup \ g\ y'\ x\ (g\ y\ z)\ [x'=x-\symb{1} \wedge y'=y+\symb{1} \wedge x \ge y] \).
Let \( \delta \) be the substitution \( [g:=f,x:=i,x':=i',y:=n',y':=n'',z:=f\ n\ a] \).
We can now deduce \( (\eqs_{12},\hs_3) \vdash (\emptyset,\hs_3) \) if one of the following
ordering requirements is satisfied:
\[
\begin{array}{rcl}
\sterm_9 = & f\ i\ (\tailup\ f\ n\ i'\ a) \succ f\ i\ (\tailup\ f\ n'\ i'\ (f\ n\ a)) & [\shortenedconstr \wedge n'' = n = \symb{1} \wedge i \ge n] \\
\tterm_9 = & \tailup \ f\ n'\ i\ (f\ n\ a) \succ \tailup \ f\ n''\ i\ (f\ n'\ (f\ n\ a)) & [\shortenedconstr \wedge n'' = n + \symb{1} \wedge i \ge n] \\
\end{array}
\]
In fact, both are satisfied by the strong boundedness property so there is nothing to check.

\smallskip
Having used \hdelete\ to remove the last remaining equation,
we have shown \( \deduces{(\eqs_1,\emptyset)}{(\emptyset,\hs_3)} \),
so by \autoref{theorem:soundnessRI} the equation
$\recdown\ f\ n\ i\ a \approx \tailup \ f\ n\ i\ a$ is an inductive theorem --
provided we indeed have a suitable bounding pair that satisfies \textbf{(REQ1)}.
But this is easily achieved: we let \( \succ \) be \((\to_{\rules \cup \mathcal{Q}} \cup\ \supterm)^+\) where \( \mathcal{Q} = \{ \recdown\ f\ n\ i\ a \to
f\ i\ (\tailup\ f\ n\ i'\ a)\ [i' = i - \symb{1} \wedge n' = n + \symb{1} \wedge i \ge n] \} \).
This is a bounding pair because \( \to_{\rules \cup \mathcal{Q}} \) is
terminating (which can for instance be proved using static dependency pairs
\cite{guo:hag:kop:val:24}).

\begin{rem}
The choice to take \( \succ \ = (\to_{\rules \cup \mathcal{Q}} \cup\ \supterm)^+\) is
quite natural: in many traditional definitions of rewriting induction
\cite{red:90,fal:kap:12,fuh:kop:nis:17,hag:kop:24} this is the only
choice for \( (\succ,\succeq) \), with \( \mathcal{Q} \) always being a directed version of
the final \( \hs \) (so in the case of this example, \( \hs_3 \)).
However, while such a choice is natural in strategies for rewriting induction, we leave it open in
the definition to allow for alternative orderings, as we will discuss in
\autoref{sec:howToFindOrdering}.
\end{rem}

\paragraph{An example with structural induction}
We continue with another example, both to introduce the remaining deduction rules and to illustrate
that the method can be used not only for induction on integers, but also for structural induction on
terms.
Consider
\( \termsSig = \{ \nil :: \lijst,\ \cons :: \int \to \lijst \to \lijst,\ \append, \rev :: \lijst \to \lijst \to \lijst \} \)
with rules
\[
\begin{array}{rlcrl}
\textbf{(R1)} & \append\ \nil\ ys \to ys & &
\textbf{(R2)} & \append\ (\cons\ x\ xs)\ ys \to \cons\ x\ (\append\ xs\ ys) \\
\textbf{(R3)} & \rev\ \nil\ ys \to ys & &
\textbf{(R4)} & \rev\ (\cons\ x\ xs)\ ys \to \rev\ xs\ (\cons\ x\ ys) \\
\end{array}
\]

\smallskip
Suppose we wish to show that \( \rev\ (\append\ xs\ ys)\ \nil \approx \append\ (\rev\ ys\ \nil)\ (\rev\ xs\ \nil) \)
is an inductive theorem.
We start with \( \eqs_1 \) containing this equation coupled with bullets as bounding terms,
and after an \induct\ step end up with \( (\eqs_2,\hs_2) \) where:
\[
\begin{array}{rcl}
\eqs_2 & = & \{\eqconwc{\sterm_2}{\rev\ (\append\ xs\ ys)\ \nil}{\append\ (\rev\ ys\ \nil)\ (\rev\ xs\ \nil)}{\tterm_2}\} \\
\hs_2 & = & \{ \rev\ (\append\ xs\ ys)\ \nil \approx \append\ (\rev\ ys\ \nil)\ (\rev\ xs\ \nil) \} \\
\sterm_2 & = & \rev\ (\append\ xs\ ys)\ \nil \\
\tterm_2 & = & \append\ (\rev\ ys\ \nil)\ (\rev\ xs\ \nil) \\
\end{array}
\]
(We will omit the constraint from equation contexts when it is just \( [\strue] \).)

Now we apply \case, using the cover set \( \{([xs:=\cons\ a\ as],\strue),\ ([xs:=\nil],\strue) \}\).
This is indeed a cover set because every ground semiconstructor instance of $xs$ must be either
\( \nil \) or headed by the list constructor \( \cons \).  We obtain \( (\eqs_2,\hs_2) \vdash
(\eqs_3,\hs_2) \) with
\[
\eqs_3 = \left\{
\begin{array}{l}
\eqconwc{\sterm_3}{\rev\ (\append\ (\cons\ a\ as)\ ys)\ \nil}{\append\ (\rev\ ys\ \nil)\ (\rev\ (\cons\ a\ as)\ \nil)}{\tterm_3} \\
\eqconwc{\sterm_4}{\rev\ (\append\ \nil\ ys)\ \nil}{\append\ (\rev\ ys\ \nil)\ (\rev\ \nil\ \nil)}{\tterm_4} \\
\end{array}
\right\}
\]
Where:
\[
\begin{array}{rclcrcl}
\sterm_3 & = & \rev\ (\append\ (\cons\ a\ as)\ ys)\ \nil & \quad &
\sterm_4 & = & \rev\ (\append\ \nil\ ys)\ \nil \\
\tterm_3 & = & \append\ (\rev\ ys\ \nil)\ (\rev\ (\cons\ a\ as)\ \nil) & \quad &
\tterm_4 & = & \append\ (\rev\ ys\ \nil)\ (\rev\ \nil\ \nil) \\
\end{array}
\]
Note that here the bounding terms are substituted along with the equation.

After a few \simplify\ steps on both equations, we end up with \( (\eqs_3,\hs_2) \vdash^* (\eqs_4,\hs_2) \):
\[
\eqs_4 = \left\{
\begin{array}{l}
\eqconwc{\sterm_3}{\rev\ (\append\ as\ ys)\ (\cons\ a\ \nil)}{\append\ (\rev\ ys\ \nil)\ (\rev\ as\ (\cons\ a\ \nil))}{\tterm_3} \\
\eqconwc{\sterm_4}{\rev\ ys\ \nil}{\append\ (\rev\ ys\ \nil)\ \nil}{\tterm_4} \\
\end{array}
\right\}
\]

\paragraph{\generalize}
To continue on $\eqs_4$, we could try to continue with the rules as they are,
interleaving \simplify, \induct\ and \case, but doing so continues to yield new equations that
cannot be eliminated easily.  There is no place where we can apply an induction hypothesis.  Instead,
it will prove beneficial to abstract the last equation by generalizing it:
\smallskip 
\newcommand{\RULEgeneralizebounded}{
  \DEDUCRULE{(\eqs \uplus \{\eqcon{\sterm}{s}{t}{\tterm}{\psi}\}, \hs)}
            {\eqcon{\sterm^\prime}{s^\prime}{t^\prime}{\tterm^\prime}{\varphi} \text{ generalizes }
             \eqcon{\sterm}{s}{t}{\tterm}{\psi}\\
             &\text{and } \csucceq{\sterm^\prime}{s^\prime}{\varphi} \text{ and }\csucceq{\tterm^\prime}{t^\prime}{\varphi}}
            {(\eqs \cup \{\eqcon{\sterm^\prime}{s^\prime}{t^\prime}{\tterm^\prime}{\varphi}\}, \hs)}
}
\RULEgeneralizebounded

This rule is quite similar to \alter\ (and in fact, every step that can be done by \alter\ can also
be done by \generalize), but they are used quite differently: \alter\ is designed to set up an
equation for the use of simplification or deletion, while \generalize\ is a form of lemma
generation, very similar to the \postulate\ rule we will see later on.

In the second equation of \(\eqs_4\), we use \generalize\ to abstract the term \( \rev\ ys\ \nil \) into a fresh variable \(zs\),
deriving \( (\eqs_4,\hs_2) \vdash (\eqs_5,\hs_2) \) with:
\[
\eqs_5 = \left\{
\begin{array}{l}
\eqconwc{\sterm_3}{\rev\ (\append\ as\ ys)\ (\cons\ a\ \nil)}{\append\ (\rev\ ys\ \nil)\ (\rev\ as\ (\cons\ a\ \nil))}{\tterm_3} \\
\eqconwc{zs}{zs}{\append\ zs\ \nil}{\append\ zs\ \nil} \\
\end{array}
\right\}
\]
This is a generalization because for every gsc substitution \( \gamma \) that respects \( \strue \)
we can choose the substitution \( \delta = [zs:=\rev\ \gamma(ys)\ \nil] \cup [\gamma(x) \mid x \in
\Var \setminus \{zs\}] \) and have both \( (\rev\ ys\ \nil)\gamma = zs\delta \) and
\( (\append\ (\rev\ ys\ \nil)\ \nil)\gamma = (\append\ zs\ \nil)\delta \), as well as
\( \sterm_4\gamma \succeq zs\delta = (\rev\ ys\ \nil)\gamma \) and
\( \tterm_4\gamma \succeq (\append\ zs\ \nil)\delta \) because \( \eqconwc{\sterm_4}{\rev\ ys\ \nil
}{\append\ (\rev\ ys\ \nil)\ \nil}{\tterm_4} \) is a bounded equation context.  Moreover, the
ordering requirements of the derivation rule, \( \csucceq{zs}{zs}{\strue} \) and
\( \csucceq{\append\ zs\ \nil}{\append\ zs\ \nil}{\strue} \), clearly hold.

\smallskip
Next, we use \induct\ to deduce \( (\eqs_5,\hs_2) \vdash (\eqs_5,\hs_3) \), with:
\[
\hs_3 = \left\{
\begin{array}{rcl}
\rev\ (\append\ xs\ ys)\ \nil & \approx & \append\ (\rev\ ys\ \nil)\ (\rev\ xs\ \nil) \\
zs & \approx & \append\ zs\ \nil \\
\end{array}
\right\}
\]
We use \case\ on the second equation, to split it up into separate cases for \( \nil \) and
\( \cons \):
\[
\eqs_6 = \left\{
\begin{array}{l}
\eqconwc{\sterm_3}{\rev\ (\append\ as\ ys)\ (\cons\ a\ \nil)}{\append\ (\rev\ ys\ \nil)\ (\rev\ as\ (\cons\ a\ \nil))}{\tterm_3} \\
\eqconwc{\nil}{\nil}{\append\ \nil\ \nil}{\append\ \nil\ \nil} \\
\eqconwc{\cons\ x\ xs}{\cons\ x\ xs}{\append\ (\cons\ x\ xs)\ \nil}{\append\ (\cons\ x\ xs)\ \nil} \\
\end{array}
\right\}
\]
The second of these equations is quickly dispatched through applications of \simplify, ending in a
\delete\ step.  For the third equation, we use \simplify\ to arrive at \( (\eqs_7,\hs_3) \):
\[
\eqs_7 = \left\{
\begin{array}{l}
\eqconwc{\sterm_3}{\rev\ (\append\ as\ ys)\ (\cons\ a\ \nil)}{\append\ (\rev\ ys\ \nil)\ (\rev\ as\ (\cons\ a\ \nil))}{\tterm_3} \\
\eqconwc{\cons\ x\ xs}{\cons\ x\ xs}{\cons\ x\ (\append\ xs\ \nil)}{\append\ (\cons\ x\ xs)\ \nil} \\
\end{array}
\right\}
\]

\paragraph{\semicons}
To deal with constructors and partially applied symbols, we introduce:
\smallskip 
\RULEcontextsymb
Note that \( \arity(\afun) = \infty \) when \( \afun \) is a constructor, so certainly \( n < \arity(\afun) \) holds.

\smallskip
In our example, we use \semicons\ to deduce \( (\eqs_7,\hs_3) \vdash (\eqs_8,\hs_3) \):
\[
\eqs_8 = \left\{
\begin{array}{l}
\eqconwc{\sterm_3}{\rev\ (\append\ as\ ys)\ (\cons\ a\ \nil)}{\append\ (\rev\ ys\ \nil)\ (\rev\ as\ (\cons\ a\ \nil))}{\tterm_3} \\
\eqconwc{\cons\ x\ xs}{x}{x}{\append\ (\cons\ x\ xs)\ \nil} \\
\eqconwc{\cons\ x\ xs}{xs}{\append\ xs\ \nil}{\append\ (\cons\ x\ xs)\ \nil} \\
\end{array}
\right\}
\]
The second equation context is immediately deleted.  For the third, note that the equation is an
instance of the induction hypothesis \( zs \approx \append\ zs\ \nil \in \hs \), so we can dispatch
it using \hdelete, provided the ordering requirements are satisfied.  If we again choose a bounding
pair with \( \supterm\;\subseteq\;\succ \), we certainly have \( (\cons\ x\ xs)
\gamma \succ xs\gamma \) for all \( \gamma \), which suffices.
We obtain 
\[
\eqs_9 = \left\{
\eqconwc{\sterm_3}{\rev\ (\append\ as\ ys)\ (\cons\ a\ \nil)}{\append\ (\rev\ ys\ \nil)\ (\rev\ as\ (\cons\ a\ \nil))}{\tterm_3}
\right\}
\]

\paragraph{\headvar} Although not needed for the example, we here also mention the derivation
rule that is almost the same as \semicons, but for variables at the head of both sides:
\RULEcontextvar
This rule is handled separately from \semicons\ because it is less safe: in theory, it is possible
that the equation in the premise is an inductive theorem while one of the equations in the
conclusion is not.  This will be made more concrete in \autoref{subsec:completeness}.

\paragraph{\postulate}
Rather than generalizing a given equation in $\eqs$, it might be preferable to introduce a lemma from scratch. 
For $\eqs_9$ it is useful to introduce lemmas to express the interaction between \( \rev \) and \( \append \).  
For this purpose, we introduce the \postulate\ rule 
\RULEpostulate
The combination of \postulate\ with \induct\ provides a way to introduce induction hypotheses in \( \hs \), which are not necessarily generalizations of existing equations in $\eqs$. 

In $\eqs_9$, we use \postulate\ and \induct\ twice each to obtain
\( (\eqs_{10},\hs_4) \), with \( \eqs_{10} = \)
\[
\left\{
\! \! 
\begin{array}{l}
\eqlabel{E}{0}\quad 
\eqconwc{\sterm_3}{\rev\ (\append\ as\ ys)\ (\cons\ a\ \nil)}{\append\ (\rev\ ys\ \nil)\ (\rev\ as\ (\cons\ a\ \nil))}{\tterm_3} \\
\eqlabel{E}{1}\quad 
\eqconwc{\append\ (\append\ x\!s\ y\!s)\ z\!s}{\append\ (\append\ x\!s\ y\!s)\ z\!s}{\append\ x\!s\ (\append\ y\!s\ z\!s)}{\append\ x\!s\ (\append\ y\!s\ z\!s)} \\
\eqlabel{E}{2}\quad 
\eqconwc{\rev\ xs\ ys}{\rev\ xs\ ys}{\append\ (\rev\ xs\ \nil)\ ys}{\append\ (\rev\ xs\ \nil)\ ys} \\
\end{array}
\! \! 
\right\}
\]
\[
\text{and}\ \hs_4 = \left\{
\begin{array}{lrcl}
\eqlabel{H}{0} & \rev\ (\append\ xs\ ys)\ \nil & \approx & \append\ (\rev\ ys\ \nil)\ (\rev\ xs\ \nil) \\
\eqlabel{H}{1} & zs & \approx & \append\ zs\ \nil \\
\eqlabel{H}{2} & \append\ (\append\ xs\ ys)\ zs & \approx & \append\ xs\ (\append\ ys\ zs) \\
\eqlabel{H}{3} & \rev\ xs\ ys & \approx & \append\ (\rev\ xs\ \nil)\ ys \\
\end{array}
\right\}
\]
Having introduced all our rules, we will not provide the rest of the proof in detail, as it would
require a lot of space without being instructive.  We only summarise the main steps.

$\eqlabel{E}{1}$ is removed through \induct\ followed by a \case\ on the
instantiation of \( xs \), some simplifications, a \delete\ in the \( \nil \) case and
\semicons\ followed by \hdelete\ in the \( \cons \) case.
No new ordering requirements are added since we already set \( \rw \cup\;\supterm\subseteq\;\succ \).

For $\eqlabel{E}{2}$, we use \case\ on \( xs \), simplifications on both sides followed by a
\delete\ to complete the \( \nil \) case, and two applications of \hypothesis\ for the \( \cons \)
case:
the equation \( \rev\ as\ (\cons\ a\ ys) \approx \append\ (\rev\ as\ (\cons\ a\ \nil))\ ys \)
is rewritten on the right side first with \( \eqlabel{H}{3} \), then with \( \eqlabel{H}{2} \),
yielding \( \rev\ as\ (\cons\ a\ ys) \approx \append\ (\rev\ as\ \nil)\ (\append\ (\cons\ a\ \nil)\ 
ys) \).  This is simplified and then removed through \hdelete, again with \( \eqlabel{H}{3} \).

Finally, $\eqlabel{E}{0}$ is removed by successively applying \hypothesis\ with $\eqlabel{H}{3}$,
$\eqlabel{H}{0}$ and $\eqlabel{H}{2}$, all on the left-hand side, followed by \hdelete\ with
$\eqlabel{H}{3}$.

The applications of \hypothesis\ yield one new ordering requirement each, giving:
\[
\begin{array}{rcl}
\append\ (\rev\ (\cons\ a\ as)\ \nil)\ ys & \succ & \append\ (\append\ (\rev\ as\ \nil)\ (\cons\ a\ \nil))\ ys \\
\append\ (\rev\ (\cons\ a\ as)\ \nil)\ ys & \succ & \append\ (\rev\ as\ \nil)\ (\append\ (\cons\ a\ \nil)\ ys) \\
\rev\ (\append\ (\cons\ a\ as)\ ys)\ \nil & \succ & \append\ (\rev\ (\append\ as\ ys)\ \nil)\ (\cons\ a\ \nil) \\
\rev\ (\append\ (\cons\ a\ as)\ ys)\ \nil & \succ & \append\ (\append\ (\rev\ ys\ \nil)\ (\rev\ as\ \nil))\ (\cons\ a\ \nil) \\
\rev\ (\append\ (\cons\ a\ as)\ ys)\ \nil & \succ & \append\ (\rev\ ys\ \nil)\ (\append\ (\rev\ as\ \nil)\ (\cons\ a\ \nil)) \\
\end{array}
\]
These can all be handled easily, for example by choosing for \( \succ \) a lexicographic path
ordering~\cite{baa:nip:98,guo:kop:24} with \( \rev > \append > \cons > \nil \).
The full proof details (as performed in our tool \cora~\cite{cora}), can be found on the companion
webpage for this paper (see \autoref{sec:implementation}).

\subsection{Completeness}\label{subsec:completeness}
Thus far, we have used rewriting induction to prove inductive theorems.
However, we can also use it to derive that some equations are \emph{not} inductive theorems.

\begin{exa}\label{example:nonEquivalence}
Consider 
$
\termsSig 
=\{
    \symb{G} :: (\int \to \int) \to \int \to \int \to \int,\ 
    \symb{H} :: (\int \to \int) \to \int \to \int \to \int \to \int
\}
$ 
with rules
\[
\begin{array}{rcllcrcll}
\symb{G} \ f \ n\ x & \arrz & \symb{G} \ f \ (n-\symb{1})\ (f\ x)  & [n > \symb{0}] & &
\symb{G} \ f \ n\ x & \arrz & x  & [n \le \symb{0}] \\
\symb{H} \ f \ n\ m\ x & \arrz & \symb{H} \ f \ (n-\symb{1})\ m\ (f\ x) & [n > \symb{0}] & &
\symb{H} \ f \ n \ m\ x & \arrz & x & [n \le \symb{0} \wedge m \le \symb{0}] \\
\symb{H} \ f \ n\ m\ x & \arrz & \symb{H} \ f \ (m-\symb{1})\ n\ (f\ x) & [m > \symb{0}] \\
\end{array}
\]
Intuitively, $\symb{G}$ computes the function $(f, n, x) \mapsto f^{n}(x)$ for all $n \ge 0$,
and $\symb{H}$ computes the function $(f, n, m, x) \mapsto f^{(n+m)}(x)$ for all $n,m \ge 0$.
We have an inductive theorem $\symb{G}\ f\ k\ x \approx \symb{H}\ f\ n\ m\ x\ [k=n+m \wedge n \ge \symb{0} \wedge m \ge \symb{0}]$. 
The condition $n,m \ge 0$ is really necessary: 
the equation  
$\symb{G}\ f\ k\ x \approx \symb{H}\ f\ n\ m\ x\ [k=n+m]$ is \emph{not} an inductive theorem. 
For example, 
$\rweq{(\symb{G}\ (\prefix{+}\ \symb{1})\ (-\symb{1})\ i)}{(\symb{H}\ (\prefix{+}\ \symb{1})\ \symb{1}\ (\symb{-2})\ i)}{\rules}$
does not hold for any integer value $i$.
\end{exa}

To prove non-equivalences, we will extend RI with a new deduction rule \disprove.
However, we must take some care: it is possible, in the course of a proof using RI, to end up
with unsound equations even if the original equations are all inductive theorems.  For example,
\( \recdown\ \prefix{+}\ n\ \symb{0}\ a \approx \recdown\ \prefix{+}\ n\ (\symb{-1})\ a \) is an
inductive theorem, but if we use \generalize\ in the RI proof to obtain an equation
\( \recdown\ \prefix{+}\ n\ i\ a \approx \recdown\ \prefix{+}\ n\ i'\ a\ [i' = i - \symb{1}] \), we
will encounter a contradiction.  Hence, we must carefully consider the derivation path.

\begin{defi}[Completeness Property]\label{def:complete}
We say a deduction rule has the Completeness Property if, whenever
\( \rules \) is a weakly normalizing, quasi-reductive, ground confluent \lcstrs, if
we can deduce
\( (\eqs,\hs) \vdash (\eqs',\hs') \) by this rule, then
\[
\leftrightarrow_{\eqs} \cup \leftrightarrow_{\hs}\ \subseteq\ \leftrightarrow^*_{\rules}
\text{ on ground terms}
\quad 
\ \ \Longrightarrow \ \ 
\quad 
\leftrightarrow_{\eqs'} \cup \leftrightarrow_{\hs'}\ \subseteq\ \leftrightarrow^*_{\rules}
\text{ on ground terms}
\]  
Here, \( \leftrightarrow_{\eqs} \) is the relation \( \leftrightarrow_{\mathcal{A}} \) where
\( \mathcal{A} = \{ s \approx t\ [\psi] \mid \eqcon{\sterm}{s}{t}{\tterm}{\psi} \in \eqs \} \).
\end{defi}

\begin{lem}\label{lem:completeprop}
The derivation rules \simplify, \case, \delete, \semicons, \induct, \hypothesis, \hdelete\ and
\alter\ all have the Completeness Property.
\end{lem}

That is, all rules in \autoref{fig:boundedRIrules} not marked with a $\circledast$
have the completeness property.

\begin{proof}
Suppose that \( \rules \) is weakly normalizing, quasi-reductive and ground confluent, and that
\( \leftrightarrow_{\eqs} \cup \leftrightarrow_{\hs}\ \subseteq\ \leftrightarrow^*_{\rules} \)
on ground terms.
Write \( \mathcal{A} = \{ s \approx t\ [\psi] \mid \eqcon{\sterm}{s}{t}{\tterm}{\psi} \in \eqs \} \).

For any \( s \approx t\ [\psi] \in \hs' \) we have either \( s \approx t\ [\psi] \in \hs \) or, if
\induct\ was used, \( s \approx t\ [\psi] \in \mathcal{A} \).  Either way, a step by this equation
can also be done by \( \leftrightarrow_{\hs} \cup \leftrightarrow_{\eqs} \), so indeed
\( \leftrightarrow_{\hs'}\ \subseteq\ \leftrightarrow^*_{\rules} \) on ground terms.
To show that \( \leftrightarrow_{\eqs'}\;\subseteq\;\leftrightarrow^*_\rules \) on ground terms,
let \( \gamma \) be a ground substitution respecting some equation context in \( \eqs' \).
Note that we only need to consider equation contexts in \( \eqs' \setminus \eqs \) because for an
equation context \( \eqcon{\sterm}{s}{t}{\tterm}{\psi} \in \eqs \) we can immediately conclude that
\( \rweq{s\gamma}{t\gamma}{\rules} \) by assumption.
We consider every deduction rule in the lemma.

\medskip\noindent\emph{\simplify, \hypothesis.}
Then \( \eqs' \setminus \eqs = \{ \eqcon{\sterm}{C[r\delta]}{t}{\tterm}{\psi} \} \) while
\( \mathcal{A} \) contains \( C[\ell\delta] \approx t\ [\psi] \) for some
\( \ell \to r\ [\varphi] \in \rules \cup \calcrules \) or some
\( \ell \approx r\ [\varphi] \) or \( r \approx \ell\ [\varphi] \) in \( \hs \),
and \( \delta \) such that \( \psi \models^\delta \varphi \).
Since \( \gamma \) respects \( \psi \) we have $\rweq{C[\ell \delta] \gamma}{t \gamma}{\rules}$ by
assumption on \( \leftrightarrow_{\mathcal{A}} \).
Furthermore, in either case $\rweq{C[r \delta]\gamma}{C[\ell \delta]\gamma}{\rules}$,
whether because \( \leftarrow_{\rules} \) is included in \( \leftrightarrow_\rules^* \) or by assumption on
\( \leftrightarrow_\hs \).

\medskip\noindent\emph{\case.}
Then the only elements of \( \eqs' \setminus \eqs \) have a form
\( \eqcon{\sterm\delta}{s\delta}{t\delta}{\tterm\delta}{\psi\delta \wedge \varphi} \), where
\( s \approx t\ [\psi] \in \mathcal{A} \).
As $\gamma$ respects $\psi \delta \wedge \varphi$, it respects $\psi \delta$.
Therefore, \( \sigma = \delta\gamma \) respects \( \psi \), and is a ground substitution;
\( \domain(\sigma) \supseteq \Vars{s,t} \).
Hence, \( s\delta\gamma = \rweq{s\sigma}{t\sigma}{\rules} = t\delta\gamma \), as desired.

\medskip\noindent\emph{\delete, \hdelete.}
We have \( \eqs' \setminus \eqs = \emptyset \) so there is nothing to prove.

\medskip\noindent\emph{\semicons.}
Then \( \eqs' \setminus \eqs = \{ \eqcon{\sterm}{s_i}{t_i}{\tterm}{\psi} \mid 1 \leq i \leq n \} \)
for some \( \eqcon{\sterm}{\afun\ s_1 \cdots s_n}{\afun\ t_1 \cdots t_n}{\tterm}{\psi} \in \eqs \),
with \( n < \arity(\afun) \).
By assumption, \( \rweq{(\afun\ s_1 \cdots s_n)\gamma}{(\afun\ t_1 \cdots t_n)\gamma}{\rules} \).
As \( \rules \) is ground confluent, there exists \( u \) such that 
\( (\afun\ s_1 \cdots s_n)\gamma \rw^* u \rwleftstar (\afun\ t_1 \cdots t_n)\gamma \).
Since \( n < \arity(\afun) \), \( u \) must have a form \( \afun\ u_1 \cdots u_n \) with
\( s_i\gamma \rw^* u_i \rwleftstar t_i\gamma \) for all \( 1 \le i \le n \).

\medskip\noindent\emph{\induct.}
For all \( \eqcon{\sterm}{s}{t}{\tterm}{\psi} \in \eqs' \) we have \( s \approx t\ [\psi] \in
\mathcal{A} \).

\medskip\noindent\emph{\alter.}
Then \( \eqs' \setminus \eqs = \{ \eqcon{\sterm'}{s'}{t'}{\tterm'}{\psi'} \} \) and there is
some \( \eqcon{\sterm}{s}{t}{\tterm}{\psi} \in \eqs \) which generalizes
\( \eqcon{\sterm'}{s'}{t'}{\tterm'}{\psi'} \).
Since $\rules$ is weakly normalizing and quasi-reductive,
$\gamma^\downarrow = [x:=\gamma(x)\! \!\downarrow_\rules\ \mid x \in
\domain(\gamma)]$ is a well-defined, gsc substitution that respects
\( \eqcon{\sterm'}{s'}{t'}{\tterm'}{\psi'} \).
By definition of generalization, there is a
substitution $\sigma$ that respects \( \eqcon{\sterm}{s}{t}{\tterm}{\psi} \) with
$s\sigma = s'\gamma^\downarrow$ and $t\sigma = t'\gamma^\downarrow$.
Hence, $s'\gamma \rw^* s'\gamma^\downarrow =
\rweq{s\sigma}{t\sigma}{\rules} = t'\gamma^\downarrow \leftarrow_{\rules}^* t'\gamma$.
\end{proof}

With the notion of a complete deduction step, we can define the states within a given derivation
path that are suitable to derive some equation is not an inductive theorem.

\begin{defi}[Complete proof state]\label{def:completestate}
Let $\mathcal{P}_1 \vdash^* \mathcal{P}_n$ be some RI deduction sequence and $1 \le i \le n$. 
Proof state $\mathcal{P}_i=(\eqs_i, \hs_i)$ is complete if one of the following holds
\begin{enumerate}[(a). ]
\item $i=1$;
\item\label{def:completepreserve}
  $\mathcal{P}_{i-1}$ is complete and $\mathcal{P}_{i-1} \vdash \mathcal{P}_i$ by a deduction
  rule with the Completeness Property;
\item\label{def:completerecover}
  there is a complete proof state $\mathcal{P}_{j}$ with $j<i$ and $\eqs_{i} \subseteq \eqs_{j}$.
\end{enumerate}
Cases (a) and (b) together ensure that deductions remain complete as long as they only use steps with
the Completeness Property.
With (c) we can restore completeness, once we have removed all equations originating from deduction
steps that cause loss of completeness. 
\end{defi} 

Note that ``complete proof state'' is \emph{only} defined within the context of a derivation
sequence: it is not a property of a standalone proof state.  This makes it difficult to formulate a
deduction rule for disproving inductive theorems: so far, our deduction rules consider only proof states,
not the history of how they were obtained.  We resolve this issue by defining a \emph{new} proof
system, based on the original one in \autoref{fig:boundedRIrules}.  This new system considers an
extended form of proof states: tuples $(\eqs,\hs, \cflag)$ where $\cflag$ is either ``$\complete$''
(complete), or ``$\incomplete(\oldeqs)$'' (incomplete) for some set of equations $\oldeqs$.
Here, $\oldeqs$ should be understood  as the last set $\eqs_i$ that precedes $(\eqs,\hs)$ in the
derivation sequence for which $(\eqs_i,\hs_i)$ was known to be complete.

\begin{figure}[tp]
\caption{Derivation rules for Bounded RI with completeness tracking and \disprove.}
\label{fig:completenessRules}
\noindent\fbox{\begin{minipage}{\textwidth}
\NAMERULE{(StepCC)}
  \DEDUCRULE{(\eqs, \hs, \complete)}
            {(\eqs,\hs) \vdash (\eqs^\prime,\hs^\prime)\ \text{by a derivation rule} \\
            & \text{with the Completeness Property}}
            {(\eqs^\prime, \hs^\prime, \complete)}
\smallskip
\NAMERULE{(StepCI)}
  \DEDUCRULE{(\eqs, \hs, \complete)}
            {(\eqs,\hs) \vdash (\eqs^\prime,\hs^\prime)\ \text{by any derivation rule}}
            {(\eqs^\prime, \hs^\prime, \incomplete(\eqs))}
\smallskip
\NAMERULE{(StepII)}
  \DEDUCRULE{(\eqs, \hs, \incomplete(\oldeqs))}
            {(\eqs,\hs) \vdash (\eqs^\prime,\hs^\prime)\ \text{by any derivation rule}}
            {(\eqs^\prime, \hs^\prime, \incomplete(\oldeqs))}
\smallskip
\NAMERULE{\completeness}
  \DEDUCRULE{(\eqs,\hs,\incomplete(\oldeqs))}
            {\eqs \subseteq \oldeqs \phantom{and some text to make it line up more}}
            {(\eqs,\hs,\complete)}
\smallskip
\NAMERULE{\disprove}
  \DEDUCRULE{(\eqs \cup \{\eqcon{\sterm}{s}{t}{\tterm}{\psi}\},\hs,\complete)}
            {s \approx t\ [\psi]\text{ is contradictory (see \autoref{def:contradictoryEq})}}
            {\bot}
\smallskip
\end{minipage}}
\end{figure}

The new proof system is given in \autoref{fig:completenessRules}.
\stepCC\ encodes \autoref{def:completestate}(b): a complete state is reduced to another complete
state by a deduction rule with the completeness property.
\stepCI\ is meant to handle the remaining deduction rules, whose application causes the next proof
state to become incomplete.  This derivation rule is formulated in a general way, allowing it to be
applied with any deduction rule, to ensure that we can also make progress when it is unknown whether
a given deduction rule has the Completeness Property or not.  However, in practice \stepCC\ should
always take precedence.
\stepII\ allows progress when the proof state is incomplete: we pass on $\oldeqs$ unmodified, as
this encodes the last set of equations for which the proof state was complete, which does not
change.
The fourth rule, \completeness, encodes \autoref{def:completestate}(c): we can recover completeness
of a proof state if its set $\eqs$ is included in the equations of a prior complete proof state.

Finally, \disprove\ is the reason for defining this new proof system: a new rule for disproving
inductive theorems.  This derivation rule relies on the following definition:

\begin{defi}[Contradictory equation]
\label{def:contradictoryEq}
An equation $s \approx t\ [\psi]$ is contradictory if there exists a
ground substitution $\delta$ that respects $s \approx t\ [\psi]$, such that one of the
following holds:
\begin{enumerate}[(1). ]
\item
  we have $s\delta =\symb{f}\ s_1' \cdots s_n'$ and $t\delta = \symb{g}\ t_1' \cdots t_m'$
  for $\symb{f},\symb{g} \in \Sig$ with \( n < \arity(\symb{f}) \) and \( m < \arity(\symb{g}) \)
\item $s, t \in T(\thSig, \Var)$ and $\termInterpret{(\psi \wedge s \neq t)\delta}=\top$
\end{enumerate}
\end{defi}
Note that if $s$ and $t$ only have base-type variables, then the existence of a
suitable substitution \( \delta \) for case (2) corresponds exactly to satisfiability of
$\psi \wedge (s \neq t)$.
However, when $s$ or $t$ contains higher-order variables, then we first have to choose an
instantiation for every such variable before we can check for satisfiability. 
For example, if we have an equation $x \approx f\ x\ [\psi]$ with $x :: \int$ and
$f :: \int \to \int$ then $\psi \wedge (x \ne f\ x)$ is not a constraint and we first have
to instantiate $f$ as a theory term (e.g. $f:=\prefix{+}\ \symb{1}$) before we can establish
satisfiability. 

\medskip
We can see that, if the \disprove\ rule applies,
$s \approx t\ [\psi]$ cannot be an inductive theorem.  This is expressed by the following
theorem (which we will prove in \autoref{sec:completenessProofRI}). 
\newcommand{\thmCompleteness}{
  Let $\alcstrs$ be a weakly normalizing, quasi-reductive, ground confluent \lcstrs\ and
  let $\eqs$ be a set of equation contexts.
  If
  \( 
      (\eqs, \emptyset, \complete)
      \vdash^*
      \bot
  \) 
  then there is an equation in \(\eqs\) which is not an inductive theorem. 
}
\begin{thm}[Completeness of Bounded RI with Completeness Tracking]\label{thm:completeness}
\thmCompleteness
\end{thm}

We assume weak normalization in \autoref{def:complete} and \autoref{thm:completeness}
because, without this assumption, \alter\ would not have the Completeness Property.
If the derivation $(\eqs, \emptyset, \complete) \vdash^* \bot$ does not use the \alter\ step,
then the conclusion of \autoref{thm:completeness} can still be obtained without
having to require weak normalization.

\begin{exa}
Consider the \lcstrs\ from \autoref{example:nonEquivalence}. 
The system terminates (which implies weak normalization) and is
ground confluent (we will prove this in \autoref{example:groundConfluent}).
It is quasi-reductive as there are no missing cases (e.g., for all values $n,m$, we have at least
one of $n > \symb{0}$ or $m > \symb{0}$ or $n \leq \symb{0} \wedge m \leq \symb{0}$).
Let $\eqs_1=\{\eqcon{\bullet}{\symb{G}\ f\ k\ x}{\symb{H}\ f\ n\ m\ x}{\bullet}{k=n+m}\}$. 
We will derive $(\eqs_1, \emptyset,\complete) \vdash^* \bot$. First, we apply \stepCC\ with
\case\ to obtain 
\[
\eqs_2 
= 
\left\{ 
\begin{aligned}
\eqlabel{E}{0}
&&
\eqconwc{\bullet}{\symb{G}\ f\ k\ x}{\symb{H}\ f\ n\ m\ x}{\bullet}\quad 
&
[
    k<\symb{0} \wedge n > \symb{0} \wedge k=n+m
]\\
\eqlabel{E}{1}
& & 
\eqconwc{\bullet}{\symb{G}\ f\ k\ x}{\symb{H}\ f\ n\ m\ x}{\bullet}\quad 
&
[
    k \ge \symb{0} \wedge k=n+m 
]\\
\eqlabel{E}{2}
& &
\eqconwc{\bullet}{\symb{G}\ f\ k\ x}{\symb{H}\ f\ n\ m\ x}{\bullet}\quad 
&
[
    n \le \symb{0} \wedge k=n+m
]
\end{aligned}
\right\}
\] 
After applying some \simplify\ steps, together with \alter\ (all through \stepCC), we get
\[
\eqs_3 =
\left\{ \eqlabel{E}{1},\eqlabel{E}{2},\ 
\begin{aligned}
\eqcon{\bullet}{x}{\symb{H}\ f\ n'\ m\ (f\ x)}{\bullet}{n' = n-\symb{1} \wedge k<\symb{0} \wedge n > \symb{0} \wedge k=n+m}
\end{aligned}
\right\}
\] 
A new application of \case,\ followed by a number of \simplify-steps with
$\rules$ yields 
\( (\eqs_4,\emptyset,\complete) \) with
\( \eqs_4 = \{ \eqlabel{E}{1}, \eqlabel{E}{2}, \eqlabel{E}{3}, \eqlabel{E}{4} \} \), where:
\[
\begin{aligned}
\eqlabel{E}{3}
&&
&\eqconwc{\bullet}{x}{f\ x}{\bullet}
&
[
    n' \le \symb{0} 
    \wedge 
    n' = n-\symb{1} \wedge k<\symb{0} \wedge n > \symb{0} \wedge k=n+m
]\\
\eqlabel{E}{4}
&&
&\eqconwc{\bullet}{x}{\symb{H}\ f\ n'\ m\ (f\ x)}{\bullet}
&
[
    n' > \symb{0} 
    \wedge 
    n' = n-\symb{1} \wedge k<\symb{0} \wedge n > \symb{0} \wedge k=n+m
]
\end{aligned}
\]
Now, $\eqlabel{E}{3}$ is contradictory by \autoref{def:contradictoryEq}.(2), so with \disprove\ we obtain $\bot$. 
By \autoref{thm:completeness} it follows that $\symb{G}\ f\ k\ x \approx \symb{H}\ f\ n\ m\ x\ [k=n+m]$ is not an inductive theorem. 
\end{exa}

Although we now technically have two different proof systems (one with completeness flag
and one without), in practice we can use both at the same time, through the result:

\begin{lem}
For all $\cflag$: $(\eqs,\hs) \vdash^* (\eqs,\hs)$ i{f}{f} there exists $\cflag'$ with
$(\eqs,\hs,\cflag) \vdash^* (\eqs,\hs,\cflag')$.
\end{lem}

\begin{proof}
The if direction is obtained by simply erasing the flags from each step (and omitting
\completeness\ steps); the only if direction by adding flags using \stepCI\ and \stepII.
\end{proof}

Hence, once we prove \autoref{theorem:soundnessRI} we immediately have the parallel result that
every equation in $\mathcal{A}$ is an inductive theorem if $(\eqs,\emptyset,\complete) \vdash^*
(\emptyset,\hs',\cflag)$ for some $\cflag$.

\section{Proofs}\label{sec:proofs}

In this section, we shall prove the soundness and completeness theorems of \autoref{sec:RI}, in
a form that allows for easily adding new deduction rules.  In addition, we will supply the results
mentioned in \autoref{rem:alter}, which allow us to limit interest to bounded equation contexts
when using Bounded Rewriting Induction.
This section is organized as follows 
\begin{itemize}[label = $\triangleright$]
\item In \autoref{sec:bounded-eq-ct} we show that Bounded RI -- as presented in \autoref{fig:boundedRIrules} -- restricts to a proof system on bounded equation contexts. 
This property is called \emph{Preserving Bounds}. 
\item In \autoref{sec:multisets} we introduce some necessary
  prerequisites on multisets.
\item In \autoref{subsec:boundedconvertibility} we introduce $\rules$-\emph{bounded} ground convertibility: a restriction of $\rules$-ground convertibility which limits intermediate terms in the conversion using the bounding pair $(\succ, \succeq)$.
We use this notion to express the \emph{Base Soundness property}, and prove that
every deduction rule except \induct\ satisfies this property.  This allows us
to obtain a more transparent and modular proof procedure towards proving
\autoref{theorem:soundnessRI}.
\item In \autoref{subsec:soundnessProofBounded} we accumulate the results of the preceding subsections to finally conclude the statement of \autoref{theorem:soundnessRI}. 
In particular, 
\autoref{theorem:boundedGroundConvertibilityRI} shows that every deduction rule that satisfies Preserving Bounds and the Base Soundness property is sound. 
\item
In \autoref{sec:completenessProofRI} we prove the full completeness statement of 
\autoref{thm:completeness} by induction on the length of a derivation sequence that ends in $\bot$. The base case is covered by \autoref{prop:completenessBaseCase}, after which the full proof is obtained 
using \autoref{lem:completeprop}.
\end{itemize}

\subsection{Bounded equation contexts}\label{sec:bounded-eq-ct}
In the following, we assume given a fixed bounding pair \( (\succ,\succeq) \) on $T(\Sig,\emptyset)$.
In this subsection we will show that Bounded RI can always be limited to bounded equation
contexts.

\begin{defi}[Preserving Bounds]\label{def:boundPreserve}
We say that a deduction rule \emph{preserves bounds} if, whenever
\( (\eqs,\hs) \vdash (\eqs',\hs') \) and all equation contexts in \( \eqs \)
are bounded, then also all equation contexts in \( \eqs' \) are bounded.

We say that it \emph{preserves strong bounds} if, whenever
\( (\eqs,\hs) \vdash (\eqs',\hs') \) and all equation contexts in \( \eqs \)
are \emph{strongly} bounded, then also all equation contexts in \( \eqs' \)
are \emph{strongly} bounded.
\end{defi}

\begin{lem}\label{lem:boundedPreserve}
All derivation rules in \autoref{fig:boundedRIrules} preserve bounds.
\end{lem}

\begin{proof}
All elements of \( \eqs' \cap \eqs \) are bounded equation contexts by assumption.  As for the rest,
consider the rule by which \( (\eqs,\hs) \vdash (\eqs',\hs') \):
\begin{description}
\item[\simplify]
  \( \eqs' \setminus \eqs = \{ \eqcon{\sterm}{C[r\delta]}{t}{\tterm}{\psi} \} \) with
  \( \eqcon{\sterm}{C[\ell\delta]}{t}{\tterm}{\psi} \in \eqs \).
  Because the original equation context is bounded, we have both
  \( \csucceq{\tterm}{t}{\psi} \) and \( \csucceq{\sterm}{C[\ell\delta]}{\psi} \).
  Because \( \rw\ \subseteq\ \succeq \), we have
  \( \csucceq{C[\ell\delta]}{C[r\delta]}{\psi} \).
  Then \( \csucceq{\sterm}{C[r\delta]}{\psi} \) by
  transitivity of \( \succeq \).

\item[\case]
  \( \eqs \) contains a bounded equation context \( \eqcon{\sterm}{s}{t}{\tterm}{\psi} \), and
  \( \eqs' \setminus \eqs \) contains equation contexts
  \( \eqcon{\sterm\delta}{s\delta}{t\delta}{\tterm\delta}{\psi\delta \wedge \varphi} \).
  To see that these are bounded as well, consider a ground substitution \( \sigma \) that respects
  the latter context.  This substitution certainly respects \( \psi\delta \), so the
  composition \( \delta\sigma \) respects \( \psi \).  Since by assumption \( \csucceq{\sterm}{s}{
  \psi} \) and \( \csucceq{\tterm}{t}{\psi} \) hold, we have \( \sterm(\delta\sigma) \succeq
  s(\delta\sigma) \) and \( \tterm(\delta\sigma) \succeq t(\delta\sigma) \).

\item[\delete,\hdelete]
  \( \eqs' \setminus \eqs = \emptyset \)

\item[\semicons,\headvar] \( \eqs' \setminus \eqs \) contains only equation contexts
  \( \eqcon{\sterm}{s_i}{t_i}{\tterm}{\psi} \) where \( \eqs \) contains a bounded equation
  context \( \eqcon{\sterm}{f\ s_1 \cdots s_n}{f\ t_1 \cdots t_n}{\tterm}{\psi} \).
  Since \( \sterm\gamma \succeq (f\ s_1 \cdots s_n)\gamma \) implies \( \sterm\gamma
  \succeq s_i\gamma \) (as \( \succeq \) is transitive and includes \( \supterm \)), and we
  similarly have \( \tterm\gamma \succeq t_i\gamma \), all elements of \( \eqs' \setminus
  \eqs \) are bounded.

\item[\induct] \( \eqs' \setminus \eqs = \{ \eqcon{s}{s}{t}{t}{\psi} \} \) and
  \( \csucceq{s}{s}{\psi} \) and \( \csucceq{t}{t}{\psi} \) hold by reflexivity of
  \( \succeq \).

\item[\hypothesis] \( \eqs' \setminus \eqs = \{ \eqcon{\sterm}{C[r\delta]}{t}{\tterm}{\psi} \} \)
  where \( \csucceq{\sterm}{C[r\delta]}{\psi} \) by definition of the deduction rule, and
  \( \csucceq{\tterm}{t}{\psi} \) because the input equation is bounded.

\item[\generalize/\alter] Boundedness of the result is a requirement of the deduction rule.

\item[\postulate] The only fresh equation context is \( \eqcon{\bullet}{s}{t}{\bullet}{
  \psi} \), and \( \bullet \succ s\gamma, t\gamma \) holds.
  \qedhere
\end{description}
\end{proof}

By almost the same proof we can see that the derivation rules -- when applied
in certain ways and with a suitably restricted bounding pair -- can also
preserve \emph{strong} bounds.

\begin{lem}\label{lem:stronglyBoundedPreserve}
Suppose $s \succ t$ whenever $s \rw t$ or $s \supterm t$.
Let $s \succeq^! t\ [\psi]$ denote that either $s = t$ or $\csucc{s}{t}{\psi}$.
Consider the variation of the proof system of \autoref{fig:boundedRIrules} where
any use of \( \succeq \) is replaced by \( \succeq^! \).
All derivation rules in this variant preserve strong bounds.
\end{lem}

Note that \( \succeq \) is only used in the definition of \hypothesis,
\generalize\ and \alter, so all other derivation rules preserve strong bounds
without modification.

\begin{proof}
The proof is a straightforward adaptation of the proof of
\autoref{lem:boundedPreserve}, observing that if $\csucceqstrong{a}{b}{\varphi}$
and $\csucceqstrong{b}{c}{\varphi}$ then $\csucceqstrong{a}{c}{\varphi}$
immediately follows.
\begin{itemize}[label = $\triangleright$]
\item For \simplify, we use that certainly $\csucc{C[\ell\delta]}{C[r\delta]}{\psi}$ because
  $\rw$ is included in $\succ$.
\item For \case, we observe that $\csucceqstrong{a}{b}{\varphi}$ clearly implies
  $\csucceqstrong{a\delta}{b\delta}{\varphi\delta}$ as well.
\item For \semicons,\headvar, we observe that $\csucceqstrong{f\ s_1 \cdots s_n}{s_i}{\psi}$ for any
  $\psi$ because $\rhd$ is included in $\succ$.
\item All other cases are immediately obvious.
  \qedhere
\end{itemize}
\end{proof}

We also observe that, as stated in \autoref{rem:alter}, the typical ways of
using \alter\ preserve both boundedness and strong boundedness without a need
for additional checks:

\begin{lem}\label{lem:altersafe}
Suppose \( \eqcon{\sterm}{s}{t}{\tterm}{\psi} \) is a bounded equation context,
and \( \eqcon{\sterm'}{s'}{t'}{\tterm'}{\psi'} \) alters it by either method
\ref{alter:equisat} or \ref{alter:substitute}.
Then also \( \eqcon{\sterm'}{s'}{t'}{\tterm'}{\psi'} \) is a bounded equation
context.
If \( \eqcon{\sterm}{s}{t}{\tterm}{\psi} \) is strongly bounded, then so is
\( \eqcon{\sterm'}{s'}{t'}{\tterm'}{\psi'} \).
\end{lem}

\begin{proof}
First consider case \ref{alter:equisat}: in this case \( \eqcon{\sterm'}{s'}{t'}{\tterm'}{\psi'} \)
is 
\( \eqcon{\sterm}{s}{t}{\tterm}{\varphi} \), where:
\begin{itemize}[label=$\triangleright$]
\item \( \{x_1,\dots,x_n\} = \Vars{\psi} \setminus \Vars{\sterm,s,t,\tterm} \) and
  \( \{y_1,\dots,y_m\} = \Vars{\varphi} \setminus \Vars{\sterm,s,t,\tterm} \);
\item \( \{z_1,\dots,z_k\} = \Vars{\varphi,\psi} \cap \Vars{\sterm,s,t,\tterm} \);
\item \( (\exists x_1 \dots x_n.\psi) \Leftrightarrow (\exists y_1 \dots y_m.\varphi) \) is
  logically valid; that is, \\
  for all appropriately typed values \( v_1,\dots,v_k \), \\
  there exist values \( u_1,\dots,u_n \) s.t.~\( \termInterpret{\psi[z_1:=v_1,\dots,z_k:=v_k,
  x_1:=u_1,\dots,x_n:=u_n]} = \true \), \\
  if and only if \\
  there exist values \( w_1,\dots,w_m \) s.t.~\( \termInterpret{\varphi[z_1\!:=\!v_1,\dots,
  z_k\!:=\!v_k,y_1\!:=\!w_1,\dots,y_m\!:=\!w_m]} = \true \).
\end{itemize}

Now, let \( \gamma \) be a ground substitution that respects \( \eqcon{\sterm}{s}{t}{\tterm}{\varphi} \).
From ``respects'', we know that \( \termInterpret{\varphi\gamma} = \true \), so for the fixed
values \( v_1 = \gamma(z_1), \dots,v_k = \gamma(z_k) \), if we choose \( w_1 = \gamma(y_1),\dots,
w_m = \gamma(y_m) \) then the second part of the ``if and only if'' above is satisfied, so we
can find values \( u_1,\dots,u_n \) that satisfy the first part.
Now, let \( \delta \) be the substitution that maps each \( x_i \) to \( u_i \) and all other
variables \( a \) to \( \gamma(a) \).  Then indeed \( \termInterpret{\psi\delta} =
\termInterpret{\psi[z_1:=v_1,\dots,z_k:=v_k,x_1:=u_1,\dots,x_n:=u_n]} = \true \), so \( \delta \)
respects \( \psi \), and since \( \delta \) is identical to \( \gamma \) on all variables in
\( \sterm,s,t,\tterm \), we have \( \sterm\gamma = \sterm\delta \succeq s\delta = s\gamma \) and
\( \tterm\gamma = \tterm\delta \succeq t\delta = t\gamma \) as required.

As for strong boundedness: if \( \sterm = s \) for the original, then clearly
the same holds for the altered context.  Otherwise \( \csucc{\sterm}{s}{\psi}
\) so, as we saw above, for all \( \gamma \) we find \( \delta \) such that
\( \sterm\gamma = \sterm\delta \succ s\delta = s\gamma \).  The same holds for
\( \tterm \) and \( t \).

Now consider case \ref{alter:substitute}: in this case
\( \eqcon{\sterm'}{s'}{t'}{\tterm'}{\psi'} \) has the form
\( \eqcon{\sterm\gamma}{s\gamma}{t\gamma}{\tterm\gamma}{\psi} \) where
\( \gamma = [x_1:=v_1,\dots,x_n:=v_n] \) such that \( \psi \Longrightarrow x_1 = v_1 \wedge \dots
\wedge x_n = v_n \) is valid.  That is, for every substitution \( \delta \) that respects \( \psi \) we
have \( \delta(x_i) = v_i\delta \).  Hence, \( \sterm\gamma\delta = \sterm\delta \succeq s\delta =
s\gamma\delta \) and \( \tterm\gamma\delta = \tterm\delta \succeq t\delta = t\gamma\delta \) as
required, and similar for strong boundedness.
\end{proof}

\subsection{Multiset orderings}
\label{sec:multisets}
The proofs in this section make substantial use of the multiset extension of an ordering pair
$(\succ, \succeq)$ on $\mathcal{A} = \Terms$.
We particularly consider multisets with exactly 2 elements, which allows for a simpler
representation than the common definitions in the literature.
We present our version of the multiset extension as follows:

For an ordering pair $(\succ, \succeq)$ on $\mathcal{A}$, a \emph{size-2 multiset} over \( \mathcal{A} \) is an unordered pair \( \multiset{a,b} \) with \(a,
b \in \mathcal{A} \) (it is allowed to have $a = b$).  We say that:
\begin{itemize}
\item \( \multiset{a_1,a_2} = \multiset{b_1,b_2} \) if either \( a_1 = b_1 \) and \( a_2 = b_2 \),
  or \( a_1 = b_2 \) and \( a_2 = b_1 \)
\item \( \multiset{a_1,a_2} \succeqmul \multiset{b_1,b_2} \) if one of the following holds:
  \begin{itemize}
  \item there is some \( i \in \{1,2\} \) such that both \( a_i \succ b_1 \) and \( a_i \succ b_2 \)
  \item either \( a_1 \succeq b_1 \) and \( a_2 \succeq b_2 \), or \( a_1 \succeq b_2 \) and
    \( a_2 \succeq b_1 \)
  \end{itemize}
\item \( \multiset{a_1,a_2} \succmul \multiset{b_1,b_2} \) if one of the following holds:
  \begin{itemize}
  \item there is some \( i \in \{1,2\} \) such that both \( a_i \succ b_1 \) and \( a_i \succ b_2 \)
  \item there exist \( i,j \in \{1,2\} \) such that \( a_i \succ b_j \) and \( a_{3-i} \succeq
    b_{3-j} \)
  \end{itemize}
\end{itemize}

Note that \( \multiset{a_1,a_2} \succmul \multiset{b_1,b_2} \) implies
\( \multiset{a_1,a_2} \succeqmul \multiset{b_1,b_2} \), because \( a_i \succ b_j \) implies
\( a_i \succeq b_j \).  Transitivity also follows easily, and termination was proved
in~\cite{thi:all:nag:12}.
Hence, \( (\succmul,\succeqmul) \) is itself an ordering pair on size-2 multisets over
\( \mathcal{A} \).

\subsection{Bounded ground convertibility}\label{subsec:boundedconvertibility}

In principle, to prove \autoref{theorem:soundnessRI}, it suffices to show $\rules$-ground
convertibility for every equation in $\eqs$.  However, as we will see, bounded rewriting induction
actually proves a stronger property, called \emph{bounded} ground convertibility.

\begin{defi}[$\rules$/$\hs/\eqs$-Bounded convertibility]\label{def:boundconvert}
Let \( \rules \) be a set of rewrite rules,
$(\succ, \succeq)$ a bounding pair,
\( \hs \) a set of equations,
and \( \eqs \) a set of equation contexts.
For \(a,b \in \GroundTerms \cup \{\bullet\}\) and \(u,v \in \GroundTerms\) we define \( u \convertsingle{a}{b}{\rules}{\hs}{\eqs} v \) if one of
the following holds:
\begin{enumerate}[(A)]
\item\label{boundconvert:rule}\label{convert:rule}
  \( u = C[\ell\sigma] \) and \( v = C[r\sigma] \) for some \( \ell \to r\ [\varphi] \) or
  \( r \to \ell\ [\varphi] \) in \( \rules \cup \calcrules \), such that:
  \begin{itemize}[label=$\triangleright$]
  \item \( \sigma \) is a ground substitution that respects the rule
  \item \( a \succeq \ell\sigma \) or \( b \succeq \ell\sigma \), and
  \item \( a \succeq r\sigma \) or \( b \succeq r\sigma \)
  \end{itemize}
\item\label{boundconvert:hypo}\label{convert:hypo}
  \( u = C[\ell\sigma] \) and \( v = C[r\sigma] \) for some \( \ell \approx r\ [\varphi] \) or
  \( r \approx \ell\ [\varphi] \) in \( \hs \) such that:
  \begin{itemize}[label=$\triangleright$]
  \item \( \sigma \) is a ground substitution that respects the equation, and
  \item \( \{a,b\} \succ_{\text{mul}} \{\ell\sigma,r\sigma\} \)
  \end{itemize}
\item\label{boundconvert:eqs}\label{convert:eqs}
  \( u = C[\ell\sigma] \) and \( v = C[r\sigma] \) for some \( \eqcon{d}{\ell}{r}{e}{\varphi} \) or
  \( \eqcon{e}{r}{\ell}{d}{\varphi} \) in \( \eqs \)
  with:
  \begin{itemize}[label=$\triangleright$]
  \item \( \sigma \) is a ground substitution that respects the equation context, and
  \item \( \{a,b\} \succeq_{\text{mul}} \{d\sigma,e\sigma\} \)
  \end{itemize}
\end{enumerate}
We write \( \convert{a}{b}{\rules}{\hs}{\eqs} \) for the reflexive, transitive closure of
\( \convertsingle{a}{b}{\rules}{\hs}{\eqs} \).
\end{defi}

\begin{defi}[\( \rules/\hs/\eqs \)-Bounded ground convertibility]
\label{def:boundedGroundConv}
Assume given a fixed \lcstrs\ with rules \( \rules \), and a bounding pair $(\succ,\succeq)$.
Let $\hs$ be a set of equations and $\eqs$ be a set of equation contexts. 
An equation context \( \eqcon{\sterm}{s}{t}{\tterm}{\psi} \) is
\( \rules/\hs/\eqs \)-bounded ground convertible if
\( s\gamma \bconvert{\sterm\gamma}{\tterm\gamma}{\rules}{\hs}{\eqs} t\gamma \) for every ground
substitution \( \gamma \) that respects \( \eqcon{\sterm}{s}{t}{\tterm}{\psi} \).
\end{defi}

We will prove (in \autoref{thm:boundedRI}) that, if
\( (\eqs,\emptyset) \vdash^* (\emptyset,\hs) \), then
every equation context in \( \eqs \) is \( \rules/\emptyset/\emptyset \)-bounded ground
convertible.
This result not only implies \autoref{theorem:soundnessRI}, but will also be invaluable to use
rewriting induction for ground confluence (\autoref{cor:RIforGroundConfluence}).

\medskip
To start with the proof, we observe that both relations
\( \convertsingle{a}{b}{\rules}{\hs}{\eqs} \) and \( \bconvert{a}{b}{\rules}{\hs}{\eqs} \) are
symmetric.  We will often use this property without explicitly stating it.
We also quickly see that our relations are preserved under contexts, and under an increase of the
bounding terms \( a, b \):

\begin{lem}\label{lem:convertContext}
Let \( D \) be a context.

\begin{enumerate*}
\item\label{lem:convertContext:bcs}
  If \( u \convertsingle{a}{b}{\rules}{\hs}{\eqs} v \), then \( D[u] \convertsingle{a}{b}{\rules}{\hs}{\eqs} D[v] \).
\item\label{lem:convertContext:bc}
  If \( u \bconvert{a}{b}{\rules}{\hs}{\eqs} v \), then \( D[u] \bconvert{a}{b}{\rules}{\hs}{\eqs} D[v] \).
\end{enumerate*}
\end{lem}

\begin{proof}
Obvious by definition of the relation: in each place where we use a context \( C[] \) we can also
use the context \( D[C[]] \) instead, without affecting the ordering requirements.
\end{proof}

\begin{lem}\label{lem:convertMul}
Suppose \( \{c,d\} \succeq_{mul} \{a,b\} \).

\begin{enumerate*}
\item\label{lem:convertMul:bcs}
  If \( u \convertsingle{a}{b}{\rules}{\hs}{\eqs} v \), then \( u \convertsingle{c}{d}{\rules}{\hs}{\eqs} v \).
\item\label{lem:convertMul:bc}
  If \( u \bconvert{a}{b}{\rules}{\hs}{\eqs} v \), then \( u \bconvert{c}{d}{\rules}{\hs}{\eqs} v \).
\end{enumerate*}
\end{lem}

\begin{proof}
The second statement follows immediately from the first.  For the first, recall that
\( (\succmul,\succeqmul) \) is an ordering pair; in particular, \( X \succeqmul Y \succeqmul Z \)
implies \( X \succeqmul Z \), and \( X \succeqmul Y \succmul Z \) implies \( X \succmul Z \).
Hence, both cases \ref{convert:hypo} and \ref{convert:eqs} are preserved.
If \( u \convertsingle{a}{b}{\rules}{\hs}{\eqs} v \) by case \ref{boundconvert:rule}, then
there exists some \( e \in \multiset{a,b} \) such that \( e \succeq \ell\sigma \).
By definition of \( \succeqmul \), there must be some \( f \in \multiset{c,d} \) such that
\( f \succ e \) or \( f \succeq e \). So either way, \( f \succeq e \succeq \ell\sigma \)  and
\( f \succeq \ell\sigma \) holds by transitivity.  Similarly, some element \( g \in \multiset{c,d}
\) exists with \( g \succeq r\sigma \).
\end{proof}

In the proof of \autoref{thm:boundedRI}, we will study derivation sequences \( (\eqs_1,\emptyset)
\vdash (\eqs_2,\hs_2) \vdash \dots \vdash (\emptyset,\hs_N) \), and
show \( \rules/\hs_i/\emptyset \)-bounded ground convertibility of all elements of any \( \eqs_i \)
(since $\hs_1 = \emptyset$ this implies our desired result \( s\gamma \bconvert{\sterm\gamma}{
\tterm\gamma}{\rules}{\emptyset}{\emptyset} t\gamma \) for every
\( \eqcon{\sterm}{s}{t}{\tterm}{\psi} \in \eqs_1 \) and every ground substitution \( \gamma \) that
respects this equation context).
The most difficult deduction rule in these proofs is \induct.  To present the
proofs in a comprehensible way, we will therefore first prove a property for
every deduction rule other than \induct. 

\begin{defi}[Base Soundness Property]
\label{def:BaseSound}
A deduction rule has the Base Soundness Property if, whenever 
\( (\eqs,\hs) \vdash (\eqs',\hs') \) by this rule and all equation contexts in \( \eqs \) are
bounded, then \( \hs' = \hs \) and for all \( \eqcon{\sterm}{s}{t}{\tterm}{\psi} \in \eqs
\setminus \eqs' \), and every ground substitution \( \gamma \) that respects
\( \eqcon{\sterm}{s}{t}{\tterm}{\psi} \), we have:
\( s\gamma \bconvert{\sterm\gamma}{\tterm\gamma}{\rules}{\hs}{\eqs'} t\gamma \).
(Note the use of $\eqs'$, not $\eqs$.)
\end{defi}

Now we can see that all rules other than \induct\ indeed have this property.

\begin{lem}\label{lem:simplifybase}
\simplify\ has the Base Soundness Property.
\end{lem}

\begin{proof}
We can write \( \eqs = \eqs_0 \uplus \{ \eqconsim{\sterm}{C[\ell\delta]}{t}{\tterm}{\psi} \} \) and
\( \eqs' = \eqs_0 \cup \{ \eqconsim{\sterm}{C[r \delta]}{t}{\tterm}{\psi} \} \) for some
\( \ell \to r\ [\varphi] \in \rules \cup \calcrules\) such that \( \psi \models^\delta \varphi \).

Let \( \gamma \) be a ground substitution that respects
\( \eqconsim{\sterm}{C[\ell\delta]}{t}{\tterm}{\psi} \).  Then we have:
\begin{itemize}[label = $\triangleright$]
\item \(C[\ell\delta]\gamma \convertsingle{\sterm\gamma}{\tterm\gamma}{\rules}{\hs}{
  \eqs'} C[r\delta]\gamma \) by \ref{boundconvert:rule}. This holds because:
  \begin{itemize}
  \item From \( \psi \models^\delta \varphi \) we know that \( \gamma \) respects
    \( \varphi\delta \).  Phrased differently, the composed substitution \( \sigma :=
    \delta\gamma \) respects \( \varphi \).
    Therefore, \(s \gamma =  C[\ell\delta]\gamma = C\gamma[\ell\sigma] \rw
    C\gamma[r\sigma] = C[r\delta]\gamma \).
  \item Since \( \csucceq{\sterm}{s}{\psi} \), and \( \succeq \) includes
    \( \supterm \), we have \( \sterm\gamma \succeq  s\gamma \succeq
    \ell\sigma \).
  \item Since \( \rw \) is included in \( \succeq \), it follows that
    \( \sterm\gamma \succeq \ell\sigma \succeq r\sigma \) as well.
  \end{itemize}
\item \( C[r\delta]\gamma \convertsingle{\sterm\gamma}{\tterm\gamma}{\rules}{\hs}{\eqs'}
  t\gamma \) by \ref{boundconvert:eqs}, because
  \( \{\sterm\gamma,\tterm\gamma \} \succeq_{mul} \{\sterm\gamma,\tterm\gamma \}
  \) by reflexivity.
\end{itemize}
Hence, putting these steps together we have
\( s\gamma \bconvert{\sterm\gamma}{\tterm\gamma}{\rules}{\hs}{\eqs'} t\gamma \).
\end{proof}

\begin{lem}\label{lem:casebase}
If \( \rules \) is weakly normalizing and quasi-reductive, then \case\ has the
Base Soundness Property.
\end{lem}

\begin{proof}
We can write \( \eqs = \eqs_0 \uplus \{ \eqcon{\sterm}{s}{t}{\tterm}{\psi} \} \) and
\( \eqs' = \eqs_0 \cup \{ \eqcon{\sterm\delta}{s\delta}{t\delta}{\tterm\delta}{\psi\delta \wedge
\varphi} \mid (\delta,\varphi) \in \coverset\} \) for \( \coverset \) a cover set of \( s \approx
t\ [\psi] \).
Let \( \gamma \) be a ground substitution that respects
\( \eqcon{\sterm}{s}{t}{\tterm}{\psi} \).
Since our rules are weakly normalizing and quasi-reductive, every term has a normal
form, and every ground normal form must be a semi-constructor term.
Let \( \gamma^\downarrow \) be the substitution that maps each \( x \) to a normal form of
\( \gamma(x) \).  Since \( \gamma \) maps all variables in \( \psi \) to values,
\( \gamma^\downarrow(x) = \gamma(x) \) on those variables, so also \( \gamma^\downarrow \) respects
\( \psi \).  Thus we see:
\begin{itemize}[label = $\triangleright$]
\item \( s\gamma \convert{\sterm\gamma}{\tterm\gamma}{\rules}{\hs}{\eqs'} s\gamma^\downarrow \)
  by $0$ or more steps using \ref{boundconvert:rule}, because:
  \begin{itemize}
  \item The assumption that \( \eqcon{\sterm}{s}{t}{\tterm}{\psi} \) is a bounded equation context gives
    \( \sterm\gamma \succeq s\gamma \).
  \item Writing \( s\gamma = u_0 \rw u_1 \rw \dots \rw u_n \), the inclusion of \( \rw \) in \( \succeq \)
    (along with transitivity of \( \succeq \)) ensures that \( \sterm\gamma \succeq u_i \) for all \( i \).
  \item For \( 1 \leq i \leq n \), writing \( u_{i-1} = C[\ell\epsilon] \) and \( u_i = C[r\epsilon] \) for
    some rule $\ell \arrz r\ [\varphi]$ and substitution $\epsilon$, we observe that \( u_{i-1} \succeq
    \ell\epsilon \) and \( u_i \succeq r\epsilon \) because \( \supterm \) is included in \( \succeq \).
  \end{itemize}
\item \( s\gamma^\downarrow \convertsingle{\sterm\gamma}{\tterm\gamma}{\rules}{\hs}{\eqs'}
  t\gamma^\downarrow \) by \ref{boundconvert:eqs}:
  \begin{itemize}
  \item Since \( \gamma^\downarrow \) is a gsc substitution, the definition of cover set gives us
    a pair \( (\delta,\varphi) \in \coverset \) and a substitution \( \epsilon \) such that
    \( \epsilon \) respects \( \psi\delta \wedge \varphi \) and \( \gamma^\downarrow =
    \delta\epsilon \) on all variables in \( \domain(\gamma^\downarrow) \).
    In particular, this means that \( s\gamma^\downarrow = s\delta\epsilon \)
    and \( t\gamma^\downarrow = t\delta\epsilon \) and
    \( \sterm\gamma^\downarrow = \sterm\delta\epsilon \) and \( \tterm\gamma^\downarrow =
    \tterm\delta\epsilon \).
  \item Hence, we can use the step with \( \eqcon{\sterm\delta}{s\delta}{t\delta}{\tterm\delta}{
    \psi\delta \wedge \varphi} \in \eqs' \).
  \item We have \( \{\sterm\gamma,\tterm\gamma\} \succeq_{mul}
    \{\sterm\gamma^\downarrow,\tterm\gamma^\downarrow\} =
    \{\sterm\delta\epsilon,\tterm\delta\epsilon\} \)
    because by definition of a bounding pair
    (in particular the inclusion of \( \rw \) in \( \succeq \)),
    \( \sterm\gamma \succeq \sterm\gamma^\downarrow \) and
    \( \tterm\gamma \succeq \tterm\gamma^\downarrow \).
  \end{itemize}
\item \( t\gamma^\downarrow \bconvert{\sterm\gamma}{\tterm\gamma}{\rules}{\hs}{\eqs'} t\gamma \) by
  \ref{boundconvert:rule} by the same reasoning as the reduction from \( s\gamma \) to
  \( s\gamma^\downarrow \).
  \qedhere
\end{itemize}
\end{proof}

\begin{lem}\label{lem:deletebase}
\delete\ has the Base Soundness Property.
\end{lem}

\begin{proof}
We can write \( \eqs = \eqs' \uplus \{ \eqcon{\sterm}{s}{t}{\tterm}{\psi} \} \), where either
\( s = t \) or \( \psi \) is unsatisfiable.
Let \( \gamma \) be a ground substitution that respects \( \eqcon{\sterm}{s}{t}{\tterm}{\psi} \).
The existence of \( \gamma \) implies that \( \psi \) \emph{is} satisfiable, so necessarily \( s = t \).
But then clearly \( s\gamma \bconvert{\sterm\gamma}{\tterm\gamma}{\rules}{\hs}{\eqs'} t\gamma \),
using \( 0 \) steps.
\end{proof}

\begin{lem}\label{lem:semiconsbase}
\semicons\ and \headvar\ have the Base Soundness Property.
\end{lem}

\begin{proof}
We can write \( \eqs = \eqs_0 \uplus \{ \eqcon{\sterm}{f\ s_1 \cdots s_n}{f\ t_1 \cdots t_n}{
\tterm}{\psi} \} \) and \( \eqs' = \eqs_0 \cup \{ \eqcon{\sterm}{s_i}{t_i}{\tterm}{\psi} \mid
1 \leq i \leq n \} \) for \( f \) a variable or function symbol.
Let $\gamma$ be some ground substitution that respects the equation context.
Since \( \multiset{\sterm\gamma,\tterm\gamma} \succeqmul \multiset{\sterm\gamma,\tterm\gamma} \),
we use \ref{boundconvert:eqs} to derive
\[
(f\ s_1\ s_2 \cdots s_n)\gamma
\convertsingle{\sterm\gamma}{\tterm\gamma}{\rules}{\hs}{\eqs'}
(f\ t_1\ s_2 \cdots s_n)\gamma
\convertsingle{\sterm\gamma}{\tterm\gamma}{\rules}{\hs}{\eqs'}
\dots
\convertsingle{\sterm\gamma}{\tterm\gamma}{\rules}{\hs}{\eqs'}
(f\ t_1\ t_2 \cdots t_n)\gamma
\]
using a context \( C[\square] = f\ t_1 \cdots t_{i-1}\ \Box\ s_{i+1} \cdots s_n \) for the $i^{\text{th}}$
step.  Together, this exactly gives
\( (f\ s_1\ s_2 \cdots s_n)\gamma \bconvert{\sterm\gamma}{\tterm\gamma}{\rules}{\hs}{\eqs'}
(f\ t_1\ t_2 \cdots t_n)\gamma \) as required.
\end{proof}

\begin{lem}\label{lem:hdeletebase}
\hdelete\ has the Base Soundness Property.
\end{lem}

\begin{proof}
We can write \( \eqs = \eqs' \uplus \{ \eqconsim{\sterm}{C[\ell\delta]}{C[r\delta]}{\tterm}{\psi} \} \)
for some \( \ell \approx r\ [\varphi] \) in \( \hs \) such that \( \psi \models^\delta \varphi \), and
we have 
\innerenumA\ \( \csucc{\sterm}{\ell\delta}{\psi} \) or
\innerenumB\ \( \csucc{\tterm}{r\delta}{\psi} \).
Let \( \gamma \) be a ground substitution that respects the equation context.
The above gives
\innerenumAprime\ \( \sterm\gamma \succ \ell\delta\gamma \) or
\innerenumBprime\ \( \tterm\gamma \succ r\delta\gamma \).
Since equation contexts in \( \eqs \) are bounded, and \( \suptermeq \) is
included in \( \succeq \), we also know that
\innerenumC\ \( \sterm\gamma \succeq C[\ell\delta]\gamma \succeq \ell\delta\gamma \) and
\innerenumD\ \( \tterm\gamma \succeq C[r\delta]\gamma \succeq r\delta\gamma \).
Using \innerenumAprime\ with \innerenumD, or \innerenumBprime\ with \innerenumC, 
\( \multiset{\sterm\gamma,\tterm\gamma} \succmul \multiset{\ell\delta\gamma,r\delta\gamma} \).
So,
\( C[\ell\delta]\gamma = C\gamma[\ell\delta\gamma]
\convertsingle{\sterm\gamma}{\tterm\gamma}{\rules}{\hs}{\eqs'}
C\gamma[r\delta\gamma] = C[r\delta]\gamma \) by \ref{boundconvert:hypo}.
\end{proof}

\begin{lem}\label{lem:postulatebase}
\postulate\ has the Base Soundness Property.
\end{lem}

\begin{proof}
Since \( \hs = \hs' \) and \( \eqs \setminus \eqs' = \emptyset \) in the case of \postulate, there
is nothing to prove.
\end{proof}

\begin{lem}\label{lem:hypothesisbase}
\hypothesis\ has the Base Soundness Property.
\end{lem}

\begin{proof}
We can write \( \eqs = \eqs_0 \uplus \{ \eqconsim{\sterm}{C[\ell\delta]}{t}{\tterm}{\psi} \} \) and
\( \eqs' = \eqs_0 \cup \{ \eqconsim{\sterm}{C[r \delta]}{t}{\tterm}{\psi} \} \) for some
\( \ell \approx r\ [\varphi] \) or \( r \approx \ell\ [\varphi] \) in \( \hs \) such that
\( \psi \models^\delta \varphi \).
Let \( \gamma \) be a ground substitution that respects
\( \eqconsim{\sterm}{C[\ell\delta]}{t}{\tterm}{\psi} \).
From \( \csucc{\pi}{\ell\delta}{\psi} \wedge \csucc{\pi}{r\delta}{\psi} \) we have
\( \pi\gamma \succ \ell\delta\gamma,r\delta\gamma \), so since \( \pi \in \{\sterm,\tterm\} \)
it follows that \( \multiset{\sterm\gamma,\tterm\gamma} \succmul \multiset{\ell\delta\gamma,
r\delta\gamma} \).
Hence, \( C[\ell\delta]\gamma = C\gamma[\ell\delta\gamma] \convertsingle{\sterm\gamma}{
\tterm\gamma}{\rules}{\hs}{\eqs'} C\gamma[r\delta\gamma] = C[r\delta]\gamma \) by
\ref{boundconvert:hypo}.
Moreover, clearly \( C[r\delta]\gamma \convertsingle{\sterm\gamma}{\tterm\gamma}{\rules}{
\hs}{\eqs'} t\gamma \) by \ref{boundconvert:eqs}.
\end{proof}

\begin{lem}\label{lem:generalizebase}
If \( \rules \) is weakly normalizing and quasi-reductive, then
\generalize\ and \alter\ have the Base Soundness Property.
\end{lem}

\begin{proof}
In both cases, we have \( \eqs = \eqs_0 \uplus \{ \eqcon{\sterm}{s}{t}{\tterm}{\psi} \} \)
and \( \eqs' = \eqs_0 \cup \{ \eqcon{\sterm'}{s'}{t'}{\tterm'}{\psi'} \}\), where the latter equation
context generalizes the former.
Let \( \gamma \) be a ground substitution that respects
\( \eqcon{\sterm}{s}{t}{\tterm}{\psi} \).
As \( \rules \) is weakly normalizing we can define \( \gamma^\downarrow \) as a substitution
that maps each \( x \) to a normal form of \( \gamma(x) \).  Since \( \rules \) is quasi-reductive we
know that \( \gamma^\downarrow \) is a gsc substitution.  Hence, by definition of generalization,
there is a substitution \( \delta \) that respects \( \psi' \) such that
(a) \( \sterm\gamma^\downarrow \succeq \sterm'\delta \),
(b) \( \tterm\gamma^\downarrow \succeq \tterm'\delta \),
(c) \( s\gamma^\downarrow = s'\delta \) and
(d) \( t\gamma^\downarrow = t'\delta \).

\begin{itemize}[label = $\triangleright$]
\item
  Clearly, we have both \( s\gamma \bconvert{\sterm\gamma}{\tterm\gamma}{\rules}{\hs}{\eqs'}
  s\gamma^\downarrow \) and \( t\gamma^\downarrow \bconvert{\sterm\gamma}{\tterm\gamma}{\rules}{
  \hs}{\eqs'} t\gamma \) by \ref{boundconvert:rule}: the facts that \( \sterm\gamma \succeq
  s\gamma \), \( \tterm\gamma \succeq t\gamma \), and that both \( \rw \) and \( \supterm \) are
  included in \( \succeq \), ensure the ordering requirements.
\item 
  Since \( \succeq \) includes \( \rw \), we know that \( \sterm\gamma \succeq
  \sterm\gamma^\downarrow \) and \( \tterm\gamma \succeq \tterm\gamma^\downarrow \), so from (a) and
  (b) together we obtain \( \multiset{\sterm\gamma,\tterm\gamma} \succeqmul \multiset{\sterm'\delta,
  \tterm'\delta} \).
  Hence, \( s\gamma^\downarrow = s'\delta \convertsingle{\sterm\gamma}{\tterm\gamma}{\rules}{\hs
  }{\eqs'} t'\delta = t\gamma^\downarrow \) by \ref{boundconvert:eqs}.
  \qedhere
\end{itemize}
\end{proof}

\subsection{Soundness of Bounded Rewriting Induction}\label{subsec:soundnessProofBounded}

Now, instead of proving \autoref{theorem:soundnessRI} directly, we will obtain the following,
more general property.

\newcommand{\thmBoundedGroundConvertabilityRI}{
  Let $\alcstrs$ be a weakly normalizing, quasi-reductive \lcstrs;
  let \( (\succ,\succeq) \) be a bounding pair, and \( \eqs_1 \) a set of bounded equation contexts.
  Suppose
  $
  (\eqs_1, \emptyset)
  \vdash
  (\eqs_2, \hs_2)
  \vdash
  \dots
  \vdash
  (\eqsN, \hs_N)
  $, where each \( \vdash \) is derived using either \induct\ or a derivation rule
  that has both the Preserving Bounds and Base Soundness Properties.
  Consider some $1 \le i \le N$.
  Then every equation context in $\eqs_i$ is $\rules/\hs_i/\eqsN$-bounded ground convertible (where \( \hs_1 = \emptyset \)).
}
\begin{thm}\label{theorem:boundedGroundConvertibilityRI}
\thmBoundedGroundConvertabilityRI
\end{thm}

Note that this result implies \autoref{theorem:soundnessRI}, but with the stronger property of
bounded ground convertibility rather than merely ground convertibility.  Moreover, it allows
us to extend the system of \autoref{fig:boundedRIrules} with additional deduction rules, and
maintain soundness so long as the new deduction rules satisfy Preserving Bounds and the
Base Soundness Property.

\medskip
The proof of \autoref{theorem:boundedGroundConvertibilityRI} roughly proceeds by
observing that for \( i = N \), every equation context in \( \eqs_i \) is
\( \rules/\hs_i/\eqsN \)-bounded ground convertible, and then showing that if
this property holds for \( \eqs_i \) with \( i > 1 \), then it also holds for
\( \eqs_{i-1} \).
We first show the inductive step for the deduction rules where \( \hs \) is
unchanged; that is, all rules other than \induct.

\newcommand{\bconvs}[5]{\convertsingle{#1}{#2}{#3}{#4}{#5}}
\newcommand{\bconv}[5]{\convert{#1}{#2}{#3}{#4}{#5}}
\newcommand{\bconvsx}[5]{\convertsingle{#1}{#2}{#3}{#4}{#5}}
\newcommand{\bconvx}[5]{\convert{#1}{#2}{#3}{#4}{#5}}
\newcommand{\eqsNU}{\mathcal{U}}

\begin{lem}\label{lem:invariantSinglestepWithoutInduct} 
Consider a deduction step \( (\eqs,\hs) \vdash (\eqs',\hs) \), i.e. leaving $\hs$ unchanged.
Suppose that 
\innerenumA\ 
  for every \( \eqcon{\sterm}{s}{t}{\tterm}{\psi} \in \eqs \setminus \eqs' \)
  and ground substitution \( \gamma \) that respects it we have
  \( s\gamma \bconvx{\sterm\gamma}{\tterm\gamma}{\rules}{\hs}{\eqs'} t\gamma \);
  and
\innerenumB\ 
  for every \( \eqcon{\sterm}{s}{t}{\tterm}{\psi} \in \eqs' \) and
  ground substitution \( \gamma \) that respects it we have
  \( s\gamma \bconvx{\sterm\gamma}{\tterm\gamma}{\rules}{\hs}{\eqsNU} t\gamma \)
  for some \( \eqsNU \).
Then for every \( \eqcon{\sterm}{s}{t}{\tterm}{\psi} \in \eqs \) and
ground substitution \( \gamma \) that respects this equation context we have
\( s\gamma \bconvx{\sterm\gamma}{\tterm\gamma}{\rules}{\hs}{\eqsNU} t\gamma \).
\end{lem}

\begin{proof}
Let \( \eqcon{\sterm}{s}{t}{\tterm}{\psi} \in \eqs \) and \( \gamma \) a ground
substitution that respects this equation context \( \psi \).
To start, if \( \eqcon{\sterm}{s}{t}{\tterm}{\psi} \in \eqs' \) then we are done by assumption \innerenumB.
Otherwise, \( \eqcon{\sterm}{s}{t}{\tterm}{\psi} \in \eqs \setminus \eqs' \).
We will show that \( s\gamma \bconvx{\sterm\gamma}{\tterm\gamma}{\rules}{\hs}{\eqsNU} t\gamma \).

By assumption \innerenumA, we may write \( s\gamma = c_0 \), \( t\gamma = c_k \), and there are
\( c_1,\dots,c_{k-1} \) such that \( c_i \bconvsx{\sterm\gamma}{\tterm\gamma}{\rules}{\hs}{\eqs'}
c_{i+1} \) for \( 0 \leq i < k \).  We are done if we can show that in fact
\( c_i \bconvx{\sterm\gamma}{\tterm\gamma}{\rules}{\hs}{\eqsNU} c_{i+1} \) for all such \( i \).
To see that this is the case, consider the definition of \( \bconvsx{\sterm\gamma}{\tterm\gamma}{
\rules}{\hs}{\eqs'} \); that is, \autoref{def:boundconvert}.
If the step from \( c_i \) to \( c_{i+1} \) is by \ref{boundconvert:rule} or
\ref{boundconvert:hypo}, then this is also a \( \bconvsx{\sterm\gamma}{\tterm\gamma}{\rules}{\hs}{
\eqsNU} \) step, so we are immediately done.
Otherwise, it is by \ref{boundconvert:eqs}, so \( c_i = C[s'\sigma] \) and \( c_{i+1} = C[t'\sigma]
\) for some \( \eqconsim{\sterm'}{s'}{t'}{\tterm'}{\varphi} \in \eqs' \) such that \( \sigma \)
respects \( \varphi \) and \( \multiset{\sterm\gamma,\tterm\gamma} \succeqmul \multiset{\sterm'
\sigma,\tterm'\sigma} \).
By assumption \innerenumB, 
\( s'\sigma \bconv{\sterm'\sigma}{\tterm'\sigma}{\rules}{\hs}{
\eqsNU} t'\sigma \).  By \autoref{lem:convertContext}, this also implies that
\( c_i = C[s'\sigma] \bconv{\sterm'\sigma}{\tterm'\sigma}{\rules}{\hs}{\eqsNU} C[t'\sigma] =
c_{i+1} \).
From \( \multiset{\sterm\gamma,\tterm\gamma} \succeqmul \multiset{\sterm'\sigma,
\tterm'\sigma} \) we obtain \( c_i \bconv{\sterm\gamma}{\tterm\gamma}{\rules}{\hs}{\eqsNU}
c_{i+1} \) by \autoref{lem:convertMul}, as required.
\end{proof}

Next, we obtain a lemma that will be very useful for when the \induct\ rule \emph{is} used.

\begin{lem}\label{lem:multisetInduction}
Suppose \( \csucceq{\sterm}{s}{\psi} \) and \( \csucceq{\tterm}{t}{\psi} \),
let \( \hs \) be a set of equations and \( \hs' = \hs \cup \{ s \approx t\ [\psi] \} \).
Let
\( s\gamma \bconv{s\gamma}{t\gamma}{\rules}{\hs'}{\eqsNU} t\gamma \) hold
for all ground substitutions \( \gamma \) that respect
\( \eqcon{\sterm}{s}{t}{\tterm}{\psi} \).

Then for all ground substitutions \( \gamma \) that respect
\( \eqcon{\sterm}{s}{t}{\tterm}{\psi} \):
\( s\gamma \bconv{\sterm\gamma}{\tterm\gamma}{\rules}{\hs}{\eqsNU} t\gamma \).

(So using \( \hs \) instead of \( \hs' \), and \( \sterm\gamma,\tterm\gamma \) instead of
\( s\gamma,t\gamma \).)
\end{lem}

\begin{proof}
Let \( \gamma \) be a ground substitution that respects
\( \eqcon{\sterm}{s}{t}{\tterm}{\psi} \).
We prove the lemma by induction on the multiset \( \{s\gamma,t\gamma\} \), ordered with
\( \succ_{mul} \).

Since \( \bconvx{s\gamma}{t\gamma}{\rules}{\hs'}{\eqsNU} \) is a transitive closure, we have
a sequence \( c_0,c_1,\dots,c_k \) with \( c_0 = s\gamma \), \( c_k = t\gamma \) and
\( c_i \bconvs{s\gamma}{t\gamma}{\rules}{\hs'}{\eqsNU} c_{i+1} \) for \( 0 \leq i < k \).
Each step from \( c_i \) to \( c_{i+1} \) by \( \leftrightarrow_\rules \),
\( \leftrightarrow_\hs \) or \( \leftrightarrow_{\eqsNU} \) is also a
\( \bconvsx{s\gamma}{t\gamma}{\rules}{\hs}{\eqsNU} \) step, hence a
\( \bconvsx{\sterm\gamma}{\tterm\gamma}{\rules}{\hs}{\eqsNU} \) step by \autoref{lem:convertMul}.
So, suppose \( c_i \bconvs{s\gamma}{t\gamma}{\rules}{\hs'}{\eqsNU} c_{i+1} \) using
\( \hs' \setminus \hs \): there is a ground substitution \( \delta \) that respects \( s \approx t\ [\psi] \), such that \( \multiset{s\gamma,t\gamma}
\succmul \multiset{s\delta,t\delta} \), and either \( c_i = C[s\delta] \),\( c_{i+1} = C[t\delta] \),
or \( c_i = C[t\delta] \),\( c_{i+1} = C[s\delta] \).
Either way, 
\( s\delta \bconv{s\delta}{t\delta}{\rules}{\hs}{\eqsNU} t\delta \), by the induction hypothesis,
and therefore by Lemmas~\ref{lem:convertContext} and~\ref{lem:convertMul},
\( c_i \bconv{\sterm\gamma}{\tterm\gamma}{\rules}{\hs}{\eqsNU} c_{i+1} \).
\end{proof}

With this, we have all the preparations needed to prove our primary soundness
result for bounded rewriting induction.  We recall the theorem:

\smallskip\noindent\fbox{\parbox{\textwidth}{\textbf{\autoref{theorem:boundedGroundConvertibilityRI}.}
\thmBoundedGroundConvertabilityRI}}

\begin{proof}
By definition of Preserving Bounds (\autoref{def:boundPreserve}), and the
fact that \induct\ also preserves bounds (\autoref{lem:boundedPreserve}), we
obtain by induction on \( i \) that
(**) each \( \eqs_i \) contains only bounded equation contexts.
Now, for \( 1 \leq i \leq N \), let \( \eqcon{\sterm}{s}{t}{\tterm}{\psi} \in \eqs_i \),
and let \( \gamma \) be a ground substitution that respects \( \psi \).  We will show,
by induction on \( N - i \):
\( s\gamma \bconv{\sterm\gamma}{\tterm\gamma}{\rules}{\hs_i}{\eqsN} t\gamma \).

If \( N - i = 0 \) this is obvious:
we have
\( s\gamma \bconvs{\sterm\gamma}{\tterm\gamma}{\rules}{\hs_i}{\eqsN} t\gamma \)
by a single step using \ref{boundconvert:eqs}.
If \( N - i > 0 \), so \( 1 \leq i < N \), consider the step
$(\eqs_i, \hs_i) \vdash (\eqs_{i+1}, \hs_{i+1})$.
If this step has the Base Soundness Property, then by (**), property \innerenumA\ of
\autoref{lem:invariantSinglestepWithoutInduct} is satisfied, while property \innerenumB\ is satisfied by the
induction hypothesis (for \( \eqsNU = \eqsN \)).  Thus,
\( s\gamma \bconv{\sterm\gamma}{\tterm\gamma}{\rules}{\hs_i}{\eqsN} t\gamma \)
by \autoref{lem:invariantSinglestepWithoutInduct}.

The only remaining case is that the deduction step was performed by \induct.
That is, \( \eqs_i = \eqs \uplus \{\eqcon{a}{u}{v}{b}{\varphi}\} \) and
\( \eqs_{i+1} = \eqs \cup \{\eqcon{u}{u}{v}{v}{\varphi}\} \) and
\( \hs_{i+1} = \hs_i \cup \{u \approx v\ [\varphi] \} \).
We distinguish two cases.
\begin{itemize}[label = $\triangleright$]
\item \( \eqcon{\sterm}{s}{t}{\tterm}{\psi} \in \eqs \).
  Then this equation context is in \( \eqs_{i+1} \), so by induction hypothesis we have
  \( s\gamma \bconv{\sterm\gamma}{\tterm\gamma}{\rules}{\hs_{i+1}}{\eqsN} t\gamma \).
  That is, there exist \( c_0,\dots,c_k \) with \( k \geq 0 \) such that
  \( s\gamma = c_0 \), \( t\gamma = c_k \) and for \( 0 \leq j < k \) we have
  \( c_j \bconvs{\sterm\gamma}{\tterm\gamma}{\rules}{\hs_{i+1}}{\eqsN} c_{j+1} \).
  It suffices if each
  \( c_j \bconv{\sterm\gamma}{\tterm\gamma}{\rules}{\hs_i}{\eqsN} c_{j+1} \) as well.
  So fix \( j \in \{0,\dots,k-1\} \), and consider how the step from \( c_j \) to \( c_{j+1} \)
  is produced.
  If the step from \( c_j \) to \( c_{j+1} \) is by \ref{boundconvert:rule} (so using
  \( \leftrightarrow_\rules \)) or \ref{boundconvert:eqs} (so using \( \leftrightarrow_{\eqsN} \)),
  or if it is by \ref{boundconvert:hypo} with an equation in \( \hs_i \), then we immediately
  have \( c_j \bconvs{\sterm\gamma}{\tterm\gamma}{\rules}{\hs_i}{\eqsN} c_{j+1} \) as well.

  The only alternative is a step by \ref{boundconvert:hypo} using \( \hs_{i+1} \setminus \hs_i =
  \{ u \approx v\ [\varphi] \} \).  That is, there exists a ground substitution \( \sigma \)
  that respects \( \varphi \) and for which \( \multiset{\sterm\gamma,\tterm\gamma} \succmul
  \multiset{u\sigma,v\sigma} \), such that either \( c_j = C[u\sigma] \) and \( c_{j+1} =
  C[v\sigma] \), or \( c_j = C[v\sigma] \) and \( c_{j+1} = C[u\sigma] \).
  Now, observe that for all ground substitutions \( \delta \) that respect
  \( \varphi \) we have \( u\delta \bconv{u\delta}{v\delta}{\rules}{\hs_{i+1}}{\eqsN}
  v\delta \) by the induction hypothesis, because \( \eqcon{u}{u}{v}{v}{\varphi} \in \eqs_{i+1} \).
  Hence we may apply \autoref{lem:multisetInduction} (with \( \csucceq{u}{u}{\varphi} \) and
  \( \csucceq{v}{v}{\varphi} \)) to obtain that
  \( u\delta \bconv{u\delta}{v\delta}{\rules}{\hs_i}{\eqsN} v\delta \).
  By \autoref{lem:convertContext}, this implies
  \( c_j \bconv{u\delta}{v\delta}{\rules}{\hs_i}{\eqsN} c_{j+1} \),
  and by \autoref{lem:convertMul} we have
  \( c_j \bconv{\sterm\gamma}{\tterm\gamma}{\rules}{\hs_i}{\eqsN} c_{j+1} \) as required.

\item \( \eqcon{\sterm}{s}{t}{\tterm}{\psi} \notin \eqs \), so
  \( \eqs_i = \eqs \uplus \{ \eqcon{\sterm}{s}{t}{\tterm}{\psi} \} \) and
  \( \eqs_{i+1} = \eqs \cup \{ \eqcon{s}{s}{t}{t}{\psi} \} \)
  Then the conditions to apply \autoref{lem:multisetInduction} are satisfied:
  \begin{itemize}
  \item by (**) and definition of a bounded equation context, we have
    \( \csucceq{\sterm}{s}{\psi} \) and \( \csucceq{\tterm}{t}{\psi} \)
  \item \( \hs_{i+1} = \hs_i \cup \{ s \approx t\ [\psi] \} \)
  \item for all ground substitutions \( \delta \) that respect \( \psi \) we have
    \( s\delta \bconv{s\delta}{t\delta}{\rules}{\hs_{i+1}}{\eqsN} t\delta \) by the induction
    hypothesis because \( \eqcon{s}{s}{t}{t}{\psi} \in \eqs_{i+1} \)
  \end{itemize}
  Hence, we immediately conclude that \( s\gamma \bconv{\sterm\gamma}{\tterm\gamma}{\rules}{
  \hs_i}{\eqsN} t\gamma \).
\qedhere
\end{itemize}
\end{proof}

With this, we can easily obtain a core soundness claim of bounded rewriting induction:

\begin{thm}\label{thm:boundedRI}
  Let $\alcstrs$ be a weakly normalizing, quasi-reductive \lcstrs\ and let \( \eqs \) be a set of
  bounded equation contexts.
  Let $(\succ, \succeq)$ be a bounding pair, and suppose $(\eqs,\emptyset) \vdash^* 
  (\emptyset, \hs)$, for some $\hs$ using the derivation rules in \autoref{fig:boundedRIrules}.
  
  Then for every equation context \( \eqcon{\sterm}{s}{t}{\tterm}{\psi} \in \eqs \), and every
  ground substitution \( \gamma \) that respects \( \psi \), we have:
  \( s\gamma \convert{\sterm\gamma}{\tterm\gamma}{\rules}{\emptyset}{\emptyset} t\gamma \).
\end{thm}

\begin{proof}
By Lemmas
\ref{lem:simplifybase}--\ref{lem:generalizebase},
all derivation rules of \autoref{fig:boundedRIrules} other than \induct\ 
have the Base Soundness Property.
By \autoref{lem:boundedPreserve} they also all preserve bounds.
The theorem follows by applying \autoref{theorem:boundedGroundConvertibilityRI}, with
\(i = 1\) (and $\eqs_1=\eqs$, $\eqsN = \emptyset$).
\end{proof}

We obtain the original soundness result as a corollary:

\smallskip\noindent\fbox{\parbox{\textwidth}{\textbf{\autoref{theorem:soundnessRI}.}
\theoremsoundnessRI}}

\begin{proof}
Clearly, equation contexts of the form \( \eqcon{\bullet}{s}{t}{\bullet}{\psi} \) are bounded.
Hence, the result immediately follows by \autoref{thm:boundedRI}.
\end{proof}

\subsection{Completeness}\label{sec:completenessProofRI}
To prove \autoref{thm:completeness}, we start with a proposition. 
\begin{prop}\label{prop:completenessBaseCase}
Let $\rules$ be a quasi-reductive, ground confluent \lcstrs\ and let $\eqs$ be a set of equation contexts.
If \( (\eqs,\hs,\cflag) \vdash \bot \)
then there exist an equation context \( \eqcon{\sterm}{s}{t}{\tterm}{\psi} \in \eqs \) and a
ground substitution $\gamma$ that respects \( \eqcon{\sterm}{s}{t}{\tterm}{\psi}
\) for which $\rweq{s \gamma}{t \gamma}{\rules}$ does not hold.
\end{prop}

\begin{proof}
We first observe that, by definition of ground confluence, we have
\( \rweq{s\gamma}{t\gamma}{\rules} \) exactly if \( \samereduct{s\gamma}{t\gamma}{\rules} \),
where we say \( \samereduct{u}{v}{\rules} \) if there exists \( w \) such that both
\( u \rw^* w \) and \( v \rw^* w \).

The only derivation rule that allows us to deduce \( \bot \) is \disprove, and by definition of
this rule, \( \eqs \) contains an equation context \( \eqcon{\sterm}{s}{t}{\tterm}{\psi} \) such
that \( s \approx t\ [\psi] \) is contradictory.  We consider all possibilities of being
contradictory.

\textbf{(1)}: there exist $\symb{f}, \symb{g} \in \Sig$ with $\symb{f} \ne
\symb{g}$ and a ground substitution $\gamma$ that respects $\psi$ such that
$s\gamma=\symb{f}\ s_1 \cdots s_n$ and
$t\gamma = \symb{g}\ t_1 \cdots t_m$ with $n < \arity(\symb{f})$ and
$m < \arity(\symb{g})$.
\emph{Claim}: $\gamma$ is a substitution for which $\samereduct{s\gamma}{t
\gamma}{\rules}$ does not hold (and therefore neither does $\rweq{s\gamma}{t
\gamma}{\rules}$).
To see this, assume that there exists \( u \) such that both
\( s\gamma \rw^* u \) and \( t\gamma \rw^* u \).
By the arity restriction, $s\gamma$ cannot reduce at the root, and nor can any
of its reducts.  Therefore, \( \head(u) = \symb{f} \).
Similarly, we conclude \( \head(u)=\symb{g} \).
This gives the desired contradiction.

\textbf{(2)}: $s, t \in T(\thSig, \Var)$ and there is a ground substitution 
$\gamma$ that respects $\psi$, with $\termInterpret{(\psi \wedge s \neq
t)\gamma}=\top$.  Hence, $\gamma$ maps every variable in $\Vars{s} \cup
\Vars{t}$ to a ground theory term.
Using only $\to_{\calcrules}$ we reduce $s \gamma$ and $t \gamma$ to their normal forms (which must be values).
We have $\termInterpret{s\gamma} \ne \termInterpret{t \gamma}$, and since calculation rules preserve interpretation, this implies $\termInterpret{(s
\gamma)\!\downarrow} \ne \termInterpret{(t \gamma)\!\downarrow}$. 
But then also $(s
\gamma)\!\downarrow \ne (t
\gamma)\!\downarrow$, as there is a one-to-one correspondence between values and their interpretations.
Hence, \( \samereduct{s\gamma}{t\gamma}{\rules} \) does not hold.
\end{proof}

\autoref{prop:completenessBaseCase} only proves the desired statement for a single-step deduction. 
To obtain the full statement of \autoref{thm:completeness},
we will use an inductive reasoning.

\smallskip\noindent\fbox{\parbox{\textwidth}{\textbf{\autoref{thm:completeness}.} \thmCompleteness}}

\begin{proof}
Assume given a deduction
\[
(\eqs_1, \hs_1, \cflag_1)
\vdash 
(\eqs_2, \hs_2, \cflag_2)
\vdash^*
(\eqs_n, \hs_n, \cflag_n)
\vdash 
\bot
\]
with \( \hs_1 = \emptyset \) and \( \cflag_1 = \complete \).
We first observe that (**) for all $i$, either $\cflag_i = \complete$, or $\cflag_i = \incomplete(\eqs_k)$
for some $k < i$.  This follows easily by induction on $i$.

Now, towards a contradiction, assume that all equations in \( \eqs_1 \) are inductive theorems;
that is, that \( \leftrightarrow_{\eqs_1}\;\subseteq\;\leftrightarrow^*_{\rules} \) on ground terms.
We show that this assumption implies that \( \leftrightarrow_{\eqs_i} \cup \leftrightarrow_{\hs_i}\;\subseteq\;
\leftrightarrow^*_{\rules} \) on ground terms for all \( 1 \le i \le n \), by induction on \( i \).
The contradiction then immediately follows by considering \(i=n\) and \autoref{prop:completenessBaseCase}.

For \( i = 1 \) this is true by assumption (since \( \hs_1 = \emptyset \)).  So let \( i > 1 \).

If the step \( (\eqs_{i-1},\hs_{i-1},\cflag_{i-1}) \vdash (\eqs_i,\hs_i,\cflag_i) \) uses \stepCC, then
\( (\eqs_{i-1},\hs_{i-1}) \vdash (\eqs_i,\hs_i) \) by a deduction rule with the Completeness
Property, so by definition of this property and the induction hypothesis we immediately obtain
\( \leftrightarrow_{\eqs_i} \cup \leftrightarrow_{\hs_i}\;\subseteq\;\leftrightarrow^*_{\rules} \)
on ground terms.

If \completeness\ is used, then \( \hs_i = \hs_{i-1} \) and by (**) \( \eqs_i \subseteq \eqs_k \)
for some \( k < i \).  By the induction hypothesis we both have
\( \leftrightarrow_{\hs_{i-1}}\;\subseteq\;\leftrightarrow_\rules^* \) and
\( \leftrightarrow_{\eqs_k}\;\subseteq\;\leftrightarrow_\rules^* \), which suffices.

Otherwise, \stepCI\ or \stepII\ is used.
  First we observe that \( \leftrightarrow_{\hs_i}\;\subseteq\;\leftrightarrow_\rules^* \):
  let \( s \approx t\ [\psi] \in \hs_i \) and \( \gamma \) be a ground substitution that respects
  \( \psi \).  If \( s \approx t\ [\psi] \in \hs_{i-1} \) then
  \( s\gamma \leftrightarrow_{\rules}^* t\gamma \) by the induction hypothesis.
  Otherwise, note that the only deduction rule that adds elements to the set of induction hypotheses
  is \induct.  As such, there is some \( \eqcon{\sterm}{s}{t}{\tterm}{\psi} \in \eqs_{i-1} \), and
  therefore we also have \( s\gamma \leftrightarrow_\rules^* t\gamma \) by the induction hypothesis.

Next, we consider \( \eqs_i \).
Whether \stepCI\ or \stepII\ is used, by (**) $\cflag_i = \incomplete(\eqs_k)$ for some $k < i$.
For \( j \) with \( i < j \leq n \), if \( \cflag_{j-1} = \incomplete(\eqs_k) \), there are only two options:
\( \cflag_j = \incomplete(\eqs_k) \) as well, or the step
\( (\eqs_{j-1},\hs_{j-1},\cflag_{j-1}) \vdash (\eqs_j,\hs_j,\cflag_j) \) follows by \completeness, and
therefore \( \eqs_j \subseteq \eqs_k \).
Since \( \cflag_n = \complete \), necessarily the latter must at some point occur. Hence, there exists
\( j > i \) with \( \eqs_j \subseteq \eqs_k \), so by the induction hypothesis,
\( \leftrightarrow_{\eqs_j}\;\subseteq\;\leftrightarrow_\rules^* \).
Let \( \eqcon{\sterm}{s}{t}{\tterm}{\psi} \in \eqs_i \) and \( \gamma \) a ground substitution that
respects \( \psi \).
Then by \autoref{theorem:boundedGroundConvertibilityRI} we have
\( s\gamma \bconvert{\sterm\gamma}{\tterm\gamma}{\rules}{\hs_i}{\eqs_j} t\gamma \), and since
both \( \leftrightarrow_{\hs_i} \) and \( \leftrightarrow_{\eqs_j} \) are
included in \( \leftrightarrow_\rules^* \), we conclude \( \rweq{s\gamma}{t\gamma}{\rules} \) as
required.
\end{proof}

Note that this proof only uses that derivation rules other than \induct\ 
and \disprove\ do not change \( \hs \).  Hence, like
\autoref{theorem:boundedGroundConvertibilityRI},
the result is applicable to extensions of \autoref{fig:boundedRIrules}.

\section{The relative power of Bounded RI}\label{sec:relativePower}
The motivation for Bounded RI is the reduction of termination requirements to obtain a more powerful proof system, compared to existing systems. 
In this section, we show that this goal has been achieved: we can prove at least the same equations
as before, and obtain more liberal ordering requirements.
This section is organized as follows:  
\begin{itemize}
\item In~\autoref{sec:improvementsRI} we discuss two deduction rules of related RI systems for
constrained rewriting~\cite{fuh:kop:nis:17, hag:kop:24} and show how Bounded RI can simulate and improve them. 
We provide two examples of how Bounded RI allows for easier induction proofs in practice.   
\item In \autoref{sec:moduloRI} we show that Bounded RI is at least equally powerful as the RI systems
based on modulo-rewriting~\cite{aot:06,aot:08a,aot:08b}, that also avoid the traditional approach of orienting induction hypotheses. 
\end{itemize}

\subsection{Improvements on related RI systems}
\label{sec:improvementsRI}
While Bounded RI is designed as an improvement on existing RI methods, it is not immediately obvious
that it captures all the same equations.  Traditional RI systems have a set of \emph{equations}
rather than equation contexts.  Hence, the use of \hypothesis\ (or similarly named derivation
rules) is not restricted by bounding terms.  Instead, elements of \( \hs \) are oriented as rules,
and termination of \( \rules \cup \hs \) is required.
Moreover, such systems have some different rules, two of which we discuss below.

Bounded RI is defined so that every step in traditional systems for constrained RI,
like~\cite{fuh:kop:nis:17,hag:kop:24}, can be replaced by one or more steps in Bounded RI.
This uses a bounding pair with \( \to_{\rules \cup \hs}\,\subseteq\,\succ \), and replaces every
equation \( s \approx t\ [\psi] \) in the original derivation by a strongly bounded
equation context \( \eqcon{\sterm}{s}{t}{\tterm}{\psi} \) which has the additional property
\textsf{(strictbound)}: for all \( u \in \{s,t\} \) there is some \( \pi \in \{\sterm,\tterm\} \)
such that \( \csucc{\pi}{u}{\psi} \).

\paragraph{\expand} 
Most existing RI systems in the literature include a deduction rule \expand,
used to perform induction proofs.
In the definition from RI for higher-order constrained rewriting~\cite{hag:kop:24}, this rule is
expressed roughly as follows:
\[
  \AxiomC{$(\eqs \uplus \{C[\afun\ s_1 \cdots s_k] \simeq t\ [\psi]\}, \hs)$}
  \RightLabel{\scriptsize$\begin{array}{l}
              \\
              \afun\ \in \Defined,\ k = \arity(\afun), \\
              \text{each}\ s_i\ \text{is a semi-} \\
              \text{constructor term,} \\
              \rules \cup \hs \cup \{
              C[\afun\ s_1 \\
              \cdots s_k] \arrz t\ [\psi] \} \\
              \text{is terminating}
              \end{array}$}
  \UnaryInfC{$\left(
              \begin{array}{c}
              \eqs \cup \{C[r]\delta \simeq t\delta\ [(\psi \wedge \varphi)\delta] \mid
              \ell \arrz r\ [\varphi] \in \rules, \delta = \mgu{\afun\ s_1 \cdots s_k}{\ell} \\
              \phantom{ABCDEFGHIJKLM}\delta(\Vars{\varphi,\psi}) \subseteq \Val \cup \Var\}, \\
              \hs \cup \{ C[\afun\ s_1 \cdots s_k] \arrz t\ [\psi] \}
              \end{array}
              \right)$}
  \DisplayProof
\]

At a first encounter, the \expand\ rule can be difficult to grasp, since it combines three
conceptual steps simultaneously: saving an induction hypothesis, starting a case analysis
\emph{and} doing a simplification.
In Bounded RI we have split this rule up to create a more modular approach.
\expand\ can then be seen as the consecutive execution of three deduction rules:

\begin{enumerate}[(1). ]
\item \induct.
  Given a proof state $\mathcal{P}_0 = (\eqs \uplus
  \{\eqcon{\sterm}{s}{t}{\tterm}{\psi}\}, \hs)$ we apply \induct\ to obtain
  \( \mathcal{P}_1 = (\eqs \cup \{\eqcon{s}{s}{t}{t}{\psi}\}, \hs \cup \{s
  \approx t\ [\psi]\}) \).
  We model the termination requirement by
  the choice of \( \succ \) discussed above.
\item\label{expand:case} \case.
  Due to the assumption of quasi-reductivity, the following is a cover set:
  \[
  \coverset_{\texttt{EXP}}
  =
  \left\{ 
  (\delta, \varphi \delta)
  \middle|
  \begin{array}{l}
  \ell \to r \ [\varphi] \in \rules,\
  \delta = \text{mgu}(\f\ s_1 \cdots s_k, \ell), \\
  \delta(\Vars{\varphi,\psi}) \subseteq \Val \cup \Var
  \end{array}
  \right\}
  \]
  We apply \case\ to $\mathcal{P}_1$, using $\coverset_{\texttt{EXP}}$ to obtain
  the following proof state
  \[
  \mathcal{P}_2
  =
  (\eqs \cup \{\eqcon{s\delta}{s\delta}{t\delta}{t\delta}{\psi\delta \wedge \varphi\delta} \mid
  (\delta,\varphi\delta) \in \coverset_{\texttt{EXP}}\}, \hs \cup \{s \approx t\ [\psi]\})
  \]
  If $s = C[\afun\ s_1 \cdots s_k]$, then each \( s\delta = C[\ell]\delta \) for some
  \( \ell \arrz r\ [\varphi] \in \rules \).
\item \simplify\ on each of the equations obtained in (2), using the
  corresponding rule $\ell \to r \ [\varphi] \in \rules$.  Due to the choice of
  cover set, the conditions of \simplify\ are satisfied.
  Assuming (wlog) that the subterm \( \f\ s_1 \cdots s_k \) occurs in the left-hand
  side of the equation,
  this yields \( \mathcal{P}_3 = (\eqs \cup \eqs^{\rw}_{\texttt{EXP}},
  \hs \cup \{s \approx t\ [\psi]\}) \) with \( \eqs^{\rw}_{\texttt{EXP}} = \)
  \[
  \left\{ 
  \eqcon{s\delta}{C[r] \delta}{t\delta}{t\delta}{\psi \delta
              \wedge
              \varphi \delta}\ 
  \middle|
  \begin{array}{l}
  \ell \to r \ [\varphi] \in \rules,
  \delta = \text{mgu}(\f\ s_1 \cdots s_k, \ell), \\
  \delta(\Vars{\varphi, \psi}) \subseteq \Val \cup \Var
  \end{array}
  \right\}
  \]
\end{enumerate}
Hence, we can simulate the traditional \expand\ rule in our work by using multiple steps.
The \textsf{(strictbound)} property is satisfied in the resulting equations: we have
\( \csucc{s\delta}{C[r] \delta}{\psi\delta\wedge \varphi\delta} \)
because \( \rw\,\subseteq\,\succ \), and
\( \csucc{s\delta}{t\delta}{(\psi\delta \wedge \varphi\delta)} \)
because all oriented induction hypotheses, including \( s \to t\ [\psi] \), are ordered by
\( \succ \).
The other deduction rules of~\cite{hag:kop:24} translate directly to Bounded RI steps, and preserve
\textsf{(strictbound)}.  When a ``simplification'' step with an element of \( \hs \) is done,
\textsf{(strictbound)} ensures that the ordering restrictions of \hypothesis\ are satisfied.

\paragraph{\eqdeletion}
In earlier, first-order work on constrained rewriting~\cite{fuh:kop:nis:17,sak:nis:sak:sak:kus:09},
another derivation rule appears that has no direct counterpart in \autoref{fig:boundedRIrules}.
\newcommand{\RULEeqdeletion}{
  \DEDUCRULE{(\eqs \uplus \{C[s_1, \ldots, s_n] \approx C[t_1, \ldots, t_n]\ [\varphi]\}, \hs)}
            {\begin{aligned}
             & s_i, t_i \in T(\thSig, \Vars{\varphi}) \\
             & \text{for all } 1 \le i \le n
             \end{aligned}}
            {(\eqs \cup \{C[s_1, \ldots, s_n] \approx C[t_1, \ldots, t_n]\ [\varphi \wedge \neg (\wedge_{i=1}^n s_i = t_i)]\}, \hs)}}
\RULEeqdeletion

This deduction rule can also be obtained in Bounded RI:
\begin{enumerate}[(1). ]
\item Start with \case\ to split $C[s_1,\ldots,s_n] \approx C[t_1,\ldots,t_n]\ [\varphi]$ into two equations:
  one where $\wedge_{i=1}^n s_i = t_i$ holds and one where $\neg(\wedge_{i=1}^n s_i = t_i)$ holds. 
\item The case $\wedge_{i=1}^n s_i = t_i$ is followed by some simplifications and \alter\ to obtain
  $C[x_1,\ldots,x_n] \approx C[x_1,\ldots,x_n]\ [\varphi \wedge (\wedge_{i=1}^n x_i = s_i \wedge x_i = t_i)]$,
  which can be removed by \delete.
\item The case $\neg (\wedge_{i=1}^n s_i = t_i)$ yields the equation $C[s_1,\ldots,s_n] \approx C[t_1,\ldots,t_n]\ [\varphi \wedge \neg (\wedge_{i=1}^n s_i = t_i)]$, exactly as in the \eqdeletion\ rule.
\end{enumerate}

\paragraph{Easier induction proofs}
Hence, proofs from our earlier work can be replayed in Bounded RI, by choosing a specific bounding
pair.  However, we can potentially do much better, since Bounded RI gives us the freedom to choose
a \emph{different} bounding pair.  In particular, we do not require the induction hypotheses
themselves to be orientable.
This allows us to entirely avoid challenges we encountered
in our previous work.  We provide two examples.

\begin{enumerate}[(i) ]
\item The paper~\cite{hag:kop:23} proves $\symb{sum2}\ x \approx \symb{sum3}\ x$ in the \lcstrs  
\[
\begin{aligned}
\symb{sum2}\ x & \to \symb{add}\ x\  (\symb{sum2}\ (x - \symb{1})) && [x > \symb{0}] 
&
\symb{sum3}\ x & \to \symb{v}\ x\ \symb{0}
&&
\\
\symb{sum2}\ x & \to \symb{return}\ \symb{0} && [x \le \symb{0}]
&
\symb{v}\ x\ a & \to \symb{v}\ (x - \symb{1})\ (a + x) && [x > \symb{0}]
\\
\symb{add}\ x\ (\symb{return}\ y)
&
\to
\symb{return}\ (x+y) 
&&&
\symb{v}\ x\ a & \to \symb{return}\ a && [x \le \symb{0}]
\end{aligned}
\] 
As explained in~\cite{hag:kop:23}, it is easy to find the lemma equation $\symb{add}\ x\  (\symb{v}\ y\ z) \approx \symb{v}\ y\ a\ [a = x+ z]$
that is needed for a rewriting induction proof to succeed.
In the RI system of~\cite{hag:kop:23}, however, this is not a suitable induction hypothesis, as
neither the rule $\symb{v}\ y\ a \to \symb{add}\ x\ (\symb{v}\ y\ z) \  [a = x+ z]$ nor
$\symb{add}\ x\  (\symb{v}\ y\ z) \to \symb{v}\ y\ a \  [a = x+ z]$ is terminating.
Further effort needs to be spent to find invariants $x>0$ and $z \ge 0$ to obtain a terminating induction rule. 
In Bounded RI this is not necessary: we can directly use the equation as induction hypothesis.
The \hdelete\ rule is especially beneficial, as it allows us to complete a proof with minimal
termination requirements (e.g., requirements that are immediately satisfied if \( \rw\;\subseteq\;\succ \)).
But even if we only use \hypothesis\ and \delete, it is easy to find a suitable ordering.

\item The paper~\cite{hag:kop:24} proves 
$
\sumfun \ f \ n
\approx
\fold \ \prefix{+}\ \symb{0} \ (\map \ f\ (\init \ n))\ 
[n \ge \symb{0}]
$, using 
\[
\begin{aligned}
{}\quad\quad
& \map \ f \ \nil \to \nil & & &
& \map \ f \ (\cons\ h\ t) \to \cons\ (f \ h)\ (\map \ f \ t) \\
& \fold \ g \ v \ \nil \to v & & &
& \fold \ g \ v \ (\cons\ h\ t) \to \fold \ g \ (g \ v \ h) \ t \\
& \init \ n \to \nil & [n < \symb{0}] & \quad &
& \init \ n \to \cons\ n\ (\init \ (n-\symb{1})) & [n \ge \symb{0}] \\
& \sumfun\ f\ x \to f\ x & [x \leq \symb{0}] & &
& \sumfun\ f\ x \to (f\ x) + (\sumfun\ f\ (x - \symb{1})) & [x > \symb{0}]
\end{aligned}
\]
The proof requires the induction rule 
\(
  x + (\fold\ \prefix{+}\ y\ l) \arrz  \fold \ \prefix{+}\ z \ l\ [z = x + y]
\) which is non-standard, having a theory symbol $+$ as root symbol on the left. 
Termination \emph{can} be proven but requiring an  advanced method, not easy to
use automatically.
In practice, it is preferred to avoid such induction rules as much as possible. 
Again, Bounded RI provides a solution, as the application of \hdelete\ with induction hypotheses \(
  x + (\fold\ \prefix{+}\ y\ l) \approx \fold \ \prefix{+}\ z \ l\ [z = x + y]
\) does not impose such a requirement.
\end{enumerate}

\subsection{Comparison to modulo rewriting with induction hypotheses}
\label{sec:moduloRI}
As mentioned in \autoref{sec:RI}, improvements on the termination requirements for the basic RI system~\cite{red:90} were already introduced in~\cite{aot:06,aot:08a,aot:08b}, based on modulo rewriting. 
Similar to Bounded RI, these works avoid orienting induction hypotheses, and explore ways to have more flexibility in constructing a well-founded ordering $\succ$, either fixed beforehand (e.g., the lexicographic path ordering), or constructed
during or after the proof.
In these works, instances of hypotheses may be oriented with a quasi-ordering \( \succeq \) rather than the default \( \succ \).

However, the imposed restrictions on the ordering are more severe than in Bounded RI:
\( \succ \ = \ \succeq \setminus \preceq \) must be a \emph{reduction order}~\cite{aot:06}
(a well-founded partial order closed under substitution and context).
In the later works~\cite{aot:08a, aot:08b}, $\succ$ is also required to be total on
ground terms.
This disallows, for instance, a choice \( \succ\ = (\to_{A} \cup\ \supterm)^+ \), or the
constructions based on dependency pairs we will discuss in \autoref{sec:howToFindOrdering}.
It would be an even more severe limitation in higher-order rewriting, where very few
ground-total reduction orderings exist.

The approach also differs in the way the orderings are employed during the RI process.
While rewriting modulo a set of induction hypotheses avoids the need to orient individual
induction hypotheses, it does so by sequentializing a number of hypothesis steps and
rewrite steps, and requiring the combination to be oriented.  This causes additional
bureaucracy, and makes it harder to use when the relation \( \succ \) is constructed
on the fly rather than fixed beforehand.  In Bounded RI, these steps can safely be done one by
one instead.

To illustrate these aspects, we will provide a bit more detail on these RI systems below.

\paragraph{Enhanced RI}
Enhanced RI (eRI)~\cite{aot:06} uses proof states $(E, H, G)$ where $H$ is used to store oriented induction hypotheses (rewrite rules), and $G$ is used to store non-oriented induction hypotheses (equations). This approach requires $R \cup H \subseteq \ \succ$ and $G \subseteq \ \succeq \cap \preceq$. 
The corresponding \emph{Delete} rule simultaneously expresses the \hdelete\ and \delete\ rule in Bounded RI,
allowing any equation $s \approx t$ to be removed if it satisfies $s\leftrightarrow^*_{G}t$. 
Additionally, the set $G$ can be used in \emph{Simplify}, to perform reductions with $\to_{(R \cup H)/G} \ = \ \leftrightarrow^*_{G} \cdot \to_{R \cup H} \cdot \leftrightarrow^*_{G}$. 

We will, however,  illustrate that under the above assumptions on $(\succ, \succeq)$, this does not yield a more powerful proof system: every eRI proof can also be performed in Bounded RI
using the bounding pair \( ((\succ \cup\;\supterm)^+,(\succeq\cup\;\supterm)^+ \)).

\medskip\noindent
\emph{Expansion2}.
The modulo $G$ approach in eRI is facilitated by the deduction rule \emph{Expansion2}, which
includes an equation $s \approx t \in E$ in $G$, and then replaces $s \approx t$ by a set of
equations, each obtained from $s \approx t$ by reducing both $s$ and $t$ with $\to_{R}$. 
Since $R$ is assumed to be included in $\succ$, every resulting equation
$s' \approx t'$ will then satisfy $s \succ s'$ and $t \succ t'$.

We will not repeat the precise definition of \emph{Expansion2} here, but assert that it can be
simulated in Bounded RI much like before.  This results in a set of strongly bounded equation
contexts \( \eqconwc{\sterm}{s'}{t'}{\tterm} \) that satisfy \textsf{(strictbound)}, as
\( \sterm \succ s' \) and \( \tterm \succ t' \).  There is a separate \emph{Expansion}
rule to include oriented equations in $H$, which also yields equation contexts with the
\textsf{(strictbound)}
property.
The only other deduction rule (aside from \emph{Delete}) is \emph{Simplify}, which allows an
equation $s \approx t$ (or $t \approx s$) to be rewritten to $s' \approx t$ (resp.~$t \approx
s'$) if $s \leftrightarrow_G^* \cdot \to_{R \cup H} s'$.
Due to \textsf{(strictbound)}, this rule can always be simulated by successive
\hypothesis\ and \simplify\ steps.
The \textsf{(strictbound)} property is preserved because, if we rewrite using an equation
from \( H \cup G \) in the direction \( a \to b \), then necessarily \( a \succeq b \).

We illustrate this idea with an example from~\cite{aot:06}: the equation $\add\ x\ y \approx \add\ y\ x$
in the system $R_{\add} = \{ \add\ \symb{0}\ y \to y,\ \add\ (\symb{s}\ x)\ y \to \symb{s}\ (\add\ x\ y) \}$.
So, $\rules = R_{\add}$ and we start with a state $\eqs = \{\eqconwc{\bullet}{\add\ x\ y}{\add\ y\ x}{\bullet} \}$, $\hs = \emptyset$. 
We first do an \emph{Expansion2} step.
\begin{enumerate}[(i)]
\item To model the addition of $\add\ x\ y \approx \add\ y\ x$ to $G$, we introduce it into $\hs$, using \induct\ to obtain
$\eqs = \{\eqconwc{\add\ x\ y}{\add\ x\ y}{\add\ y\ x}{\add\ y\ x}\}$, $\hs = \{\add\ x\ y \approx \add\ y\ x\}$. 
\item We then apply \case\ with cover set \( \{ [x:=\symb{0},y:=\symb{0}], [x:=\symb{0},y:=\symb{s}(z)],
  [x:=\symb{s}(z),y:=\symb{0}], [x:=\symb{s}(x'),y:=\symb{s}(y')] \} \), and in each of the resulting equation
  contexts, follow up the step with a use of \simplify\ both on the left-hand side and the right-hand side.
\[
\begin{array}{ll}
\eqlabel{E}{0} & \eqconwc{\add\ \symb{0}\ \symb{0}}{\symb{0}}{\symb{0}}{\add\ \symb{0}\ \symb{0}} \\
\eqlabel{E}{1} & \eqconwc{\add\ \symb{0}\ (\symb{s}\ z)}{\symb{s}\ z}{\symb{s}\ (\add\ z\ \symb{0})}{\add\ (\symb{s}\ z)\ \symb{0}} \\
\eqlabel{E}{2} & \eqconwc{\add\ (\symb{s}\ z)\ \symb{0}}{\symb{s}\ (\add\ z\ \symb{0})}{\symb{s}\ z}{\add\ \symb{0}\ (\symb{s}\ z)} \\
\eqlabel{E}{3} & \eqconwc{\add\ (\symb{s}\ x')\ (\symb{s}\ y')}{\symb{s}\ (\add\ x'\ (\symb{s}\ y'))}{\symb{s}\ (\add\ y'\ (\symb{s}\ x'))}{\add\ (\symb{s}\ y')\ (\symb{s}\ x')} \\
\end{array}
\]
\end{enumerate}
Like~\cite{aot:06} we delete \( \eqlabel{E}{0} \), and observe that \( \eqlabel{E}{1} \) and \( \eqlabel{E}{2} \)
represent the same equation context (just with sides reversed), so will ignore the latter.
In eRI, the remaining equations are removed by \emph{Simplify} steps followed by \emph{Delete}.
For example, $\eqlabel{E}{1}$ is removed using $\symb{s}\ (\add\ z\ \symb{0}) \to_{R_{\add}/G}
\symb{s}\ z$, after which the two sides of the equation are equal, hence can be removed. 
In Bounded RI, we model this by chaining \hypothesis\ (with $\add\ x\ y \approx \add\ y\ x \in \hs$) and
\simplify, yielding the equation \( \symb{s}\ z \approx \symb{s}\ z \), which can be deleted.
Since \( \supterm \subseteq\ \succ \) the \hypothesis\ step imposes only an inequality
$\add\ (\symb{s}\ z)\ \symb{0} \succeq \symb{s}\ (\add\ \symb{0}\ z)$,
which is implied by $\rules \subseteq \ \succ$ and $\{\add\ x\ y \approx \add\  x\ y \} \subseteq\ \succeq \cap
\preceq$ because then $\add\ (\symb{s}\ z)\ \symb{0} \succ \symb{s}\ (\add\ z\ \symb{0}) \succeq \symb{s}\ 
(\add\ \symb{0}\ z)$. 

\( \eqlabel{E}{3} \), in eRI, is addressed by two \emph{Simplify} steps, which we model by successive
applications of \hypothesis\ and \simplify\ on the left, then \hypothesis, \simplify\ and \hypothesis\ on the
right, to obtain an equation \( \symb{s}\ (\symb{s}\ (\add\ y'\ x')) \approx \symb{s}\ (\symb{s}\ (\add\ y'\ x')) \)
which may be removed.
This yields ordering requirements \( \add\ (\symb{s}\ x')\ (\symb{s}\ y') \succeq
\symb{s}\ (\add\ (\symb{s}\ y')\ x') \) and
\( \add\ (\symb{s}\ y')\ (\symb{s}\ x') \succeq \symb{s}\ (\add\ (\symb{s}\ x')\ y'),\ 
\symb{s}\ (\symb{s}\ (\add\ y'\ x')) \).  Both of these
are implied by the requirements in~\cite{aot:06}, and satisfied by a recursive path ordering with multiset
comparison.

\paragraph{Extensions of eRI}
Aoto introduced several extensions of eRI. 
For example, in BRI~\cite{aot:08a} 
the requirement $s (\succeq \cap \preceq) t$, to introduce an induction hypothesis $s \approx t$ in $G$, is replaced by $s \not \succ t$ and $t \not \succ s$, but the corresponding simplify rule can only use $\sigma$-instances for which $s \sigma (\succeq \cap \preceq) t \sigma$. 
This is for instance useful when $s=\add\ x\ (\add\ y\ z) $ and $t=\add\ y\ (\add\ x\ z)$, for which $s \not \succ t$ and $t \not \succ s$, and the inductive theorem $s \approx t$ cannot be proven by eRI.
In BRI one can prove $s \approx t$, which eventually requires one induction rule $\symb{s}\ (\symb{s}\ (\add\ x\ (\symb{s}\ z))) \to \symb{s}\ (\symb{s}\ (\symb{s}\ (\add\ x\ z)))$ in $H$. 
However, BRI employs the assumption that $\succeq$ restricts to a total ordering on ground terms. 
Furthermore, any BRI-proof can also be performed in Bounded RI,
which does not require orientation of induction hypotheses in the first place.

Another variation is \emph{incremental RI} (iRI)~\cite{aot:06}, where one can apply simplify steps with other equations in $E$ that are yet to be proven.
This allows the authors to prove $E = \{\mul\ x\ y \approx \mul\ y\ x,\ \add\ x\ y \approx \add\ y\ x,\ \add\ x\ (\add\ y\ z) \approx \add\ (\add\ x\ y)\ z\}$, where $\mul\ \symb{0}\ y \to \symb{0}$, $\mul\ (\symb{s}\ x)\ y \to \add\ (\mul\ x\ y)\ y$, and where the $\add$-equations are used as hypotheses to prove the $\mul$-equation.
This is also possible in Bounded RI, where we can use \induct\ to introduce the two $\add$-equations to $\hs$, before addressing the $\mul$-equation.  

\section{Finding a Bounding Pair}\label{sec:howToFindOrdering}

The bounding pair $(\succ,\succeq)$ is an important component of bounded
rewriting induction.  Yet, how to find it in practice?  As we have seen, we do
not have to fix the pair in advance, but can use the process of bounded
rewriting induction to accumulate requirements.
Let us say the process delivers a set $\reqs$ of requirements of the form
$\csucc{s}{t}{\varphi}$ or $\csucceq{s}{t}{\varphi}$ with $s \neq t$ (if
$s = t$, then we already know \( \csucceq{s}{t}{\varphi} \) is true and
\( \csucc{s}{t}{\varphi} \) is false if \( \varphi \) is satisfiable).
But having done so, we are still left with the difficulty of proving that a
bounding pair for those requirements exists.  We discuss two different
approaches to define such a pair.

\subsection{Using a reduction ordering}\label{sec:redord}
A \emph{reduction ordering} is a monotonic (i.e. $s \sqsupset t$ implies $C[s]
\sqsupset C[t]$ for every context $C[\square]$) well-founded partial ordering
on the set of terms.
Monotonicity does not imply that \( s \sqsupset t \) whenever \( s \supterm t \)
but we \emph{can} use a reduction ordering that includes \( \rw \) to construct
a bounding pair, as follows:
\begin{itemize}[label = $\triangleright$]
\item define $a \succ_1 b$ if $a \sqsupset b$ or $a \supterm b$;
  define $a \succ b$ if $a \succ_1^+ b$, and $a \succeq b$ if $a \succ_1^* b$;
\item require that $\ell\gamma \sqsupset r\gamma$ for each $\ell \arrz r\ 
  [\varphi] \in \rules \cup \calcrules$ and substitution $\gamma$
  that respects $\varphi$. \\
  (Note: traditionally a reduction ordering must be closed under substitutions.
  In constrained rewriting, we use this implicit quantification over suitable
  substitutions instead.)
\end{itemize}
The monotonicity requirement ensures that well-foundedness of $\sqsupset$ is
preserved in $\succ$: if $s \sqsupset t = C[t']$ and $t' \sqsupset u$, then $s
\sqsupset C[t'] \sqsupset C[u]$, so any infinite sequence of $\succ$ steps can
be converted into an infinite sequence of $\sqsupset$ steps.

There are various methods to define a reduction ordering in traditional term
rewriting, with recursive path orderings and polynomial interpretations being
among the most well-known.  Unfortunately, there are not as many methods for
higher-order or constrained term rewriting, and even fewer for the combination
(although a variant of the recursive path ordering for \lcstrss\ exists
\cite{guo:kop:24}).
That being said, unconstrained first-order term rewriting systems can be seen as
a special case of \lcstrss, so Bounded RI can also be used in this setting.
Then, we have many termination methods at our disposal.

\paragraph{Using a rewrite relation}
It is worth noting that if $\Q$ is a set of rules such that $\rules
\cup \Q$ is a terminating rewrite system, then $\arrz_{\rules \cup \Q}^+$ is by
definition a reduction ordering.  So, if every $\csucc{s}{t}{\varphi}$ or
$\csucceq{s}{t}{\varphi}$ in $\reqs$ gives a valid rewrite rule $s \arrz t\ 
[\varphi]$, we can simply let $\Q$ be this set of rules, and use any method to
prove termination.  Thus, we could in particular use the dependency pair
framework \cite{art:gie:00,gie:thi:sch:05}, which allows reduction orderings to
be used in a more liberal way by requiring \emph{weak} instead of full
monotonicity.

A downside of this approach is that it is possible to encounter requirements
that are not valid rules.  This is especially common in higher-order
rewriting since the left-hand side of a rule may not have a variable as the
head, and must have the same type as the right-hand side.
For example, suppose $\reqs \ni \csucc{\afun\ (F\ x)\ y}{y-\symb{1}}{x > y}$,
with $\afun :: \int \arrtype \int \arrtype \lijst$.  We cannot include a rule
$\afun\ (F\ x)\ y \arrz y - \symb{1}\ [x > y]$ in $\Q$, because the two sides
have a different type ($\lijst$ versus $\int$).  Even if they had the same
type, the only dependency pair framework for LCSTRSs that has thus far been
defined \cite{guo:hag:kop:val:24} does not support rules whose left-hand side
contains applied variables, so the subterm $F\ x$ is problematic.

A solution to this issue is to not derive $\Q$ from the set $\reqs$ directly,
but merely to ensure that $s\gamma\ (\arrz_{\rules \cup \Q} \cup\ 
\supterm)^+\ t\gamma$ for all $\csucc{s}{t}{\varphi}$ or
$\csucceq{s}{t}{\varphi}$ in $\reqs$ and $\gamma$ that respect
$\varphi$.  For example, if we introduce a new function symbol
$\symb{inttolist} :: \int \arrtype \lijst$, then the above requirement is
handled by including a rule $\afun\ z\ y \arrz \symb{inttolist}\ 
(y-\symb{1})\ [\strue]$ in $\Q$.
While it is not immediately obvious how to build such an abstraction in
general, at least it is easy to avoid the issue of different types through
the introduction of constructors like $\symb{inttolist}$.

\medskip
An advantage of using a reduction relation is that it generates a bounding pair
that satisfies the premise of \autoref{lem:stronglyBoundedPreserve}.  Hence,
by taking some care with the use of \hypothesis, \alter\ and \generalize, we
can ensure that equation contexts are always \emph{strongly} bounded.  This
yields substantially fewer ordering requirements -- and thus an easier proof
search.

\subsection{Separating root and inner steps}

The idea explored above --  deriving a bounding pair from a single
reduction ordering -- is quite powerful and allows us to reuse existing termination
methods.  In particular the approach of letting $\sqsupset$ be a terminating
relation $\arrz_{\rules \cup \Q}$ is reminiscent of traditional ways of
using RI~\cite{red:90,fal:kap:12,fuh:kop:nis:17,hag:kop:24}, where
the induction hypotheses in $\hs$ are oriented as rules, and we prove
termination of $\arrz_{\rules \cup \hs}$.

However, in these approaches we do not fully take advantage of the weaker requirements in the
present setting.  In particular, the elements of $\reqs$ do not need to be oriented with a
\emph{monotonic} relation. For example, if $\rules = \{ \afun\ \symb{a} \arrz \afun\ \symb{b} \}$
and $\reqs = \{ \symb{b} \succ \symb{a} \}$, there is no reduction ordering $\sqsupset$
that orients $\reqs$, since then we would need both $\afun\ \symb{a} \sqsupset \afun\ 
\symb{b}$ and $\afun\ \symb{b} \sqsupset \afun\ \symb{a}$ (by monotonicity).  Yet, there \emph{are}
bounding pairs for these requirements: e.g., the relation $(\arrztop{\{\symb{b} \arrz \symb{a}\ 
[\strue]\}} \cup \arrz_\rules \cup\ \supterm)^+$, whose subrelation $\arrztop{\Q}$ is defined by:

\begin{defi}\label{def:arrztop}
For a set $\Q$ of triples $(\ell,r,\varphi)$ of two terms and a constraint, let $s \arrztop{\Q} t$
if there exist $(\ell,r,\varphi) \in \Q$ and a substitution $\gamma$ that respects $\varphi$, such that $s = \ell\gamma$ and
$t = r\gamma$.
\end{defi}

In terms of the (first-order) dependency pair framework~\cite{art:gie:00,gie:thi:sch:05,kop:13},
the difference is \emph{roughly}: termination of $\arrz_{\rules \cup \Q}$ coincides with finiteness
of the DP problem $(\DP(\rules \cup \Q),\rules \cup \calcrules \cup \Q)$, while termination of
$\arrztop{\Q} \cup \arrz_\rules \cup\ \rhd$ coincides with finiteness of
$(\DP(\rules \cup \Q),\rules \cup \calcrules)$.
In higher-order rewriting, there are multiple dependency pair frameworks, and their full introduction
requires some theoretical development that is beyond the scope of this paper.  However, we will
show how the question whether a suitable bounding pair exists can be reduced to the existence of
something much like a dependency pair chain.

\paragraph{Dependency-pair-like chains}
In the following, we assume that \( \reqs \) contains only requirements of the
form \( \afun\ \vec{\ell} \succ r\ [\varphi] \) or \( \afun\ \vec{\ell} \succeq
r\ [\varphi] \) with \( \afun \in \Sig \).
We define:
\begin{itemize}
\item \( \reqs_\succeq = \{ \ell \arrz r\ [\varphi] \mid \csucceq{\ell}{r}{\varphi} \in
  \reqs \}\) and \( \reqs_\succ = \{ \ell \arrz r\ [\varphi] \mid \csucc{\ell}{r}{\varphi}
  \in \reqs \}\).
\item $\rules^\prime$ is the set $\{ \ell\ \avar_1 \cdots \avar_i \arrz r\ \avar_1
  \cdots \avar_i\ [\varphi] \mid \ell \arrz r\ [\varphi] \in \rules \cup \calcrules
  \wedge \ell :: \atype_1 \arrtype \dots \arrtype \atype_n \arrtype \asort$
  ($\asort \in \Sorts$) $\wedge\ 0 \leq i \leq n \wedge \avar_1 :: \atype_1,
  \dots,\avar_i :: \atype_i$ fresh variables$\}$.
\item \( \Defs \) is the set of pairs \( (\afun,n) \) such that there exists
  \( \afun\ \ell_1 \cdots \ell_n \arrz r\ [\varphi] \in \rules' \cup \reqs_\succ \cup
  \reqs_\succeq \).
\item A \emph{candidate} is a term of the form \( f\ s_1 \cdots s_n \)
  with either \( f \in \Var \) and \( n > 0 \), or \( (f,n) \in \Defs \).
\item Let $\Pweak = \{ \ell \arrz r\ [\varphi] \in \reqs_\succeq \cup
  \rules^\prime \mid r$ a candidate$\}$ and
  $\Pstrong = \{ \ell \arrz r\ [\varphi] \in \reqs_\succ \mid$ \linebreak $r$ a candidate$\}
  \cup \{ \ell \arrz p\ [\varphi] \mid \ell \arrz r\ [\varphi] \in \reqs_\succ
  \cup \reqs_\succeq \cup \rules^\prime \wedge r \supterm p \wedge p$ a
  candidate$\}$.
\item A \emph{strong chain} is an infinite sequence $s_i \arrztop{\Pstrong}
  t_i\ (\arrz_\rules \cup \arrztop{\Pweak})^*\ s_{i+1}$ for all $i$.
\end{itemize}

The elements of $\Pweak \cup \Pstrong$ are essentially the \emph{dynamic
dependency pairs} of $\reqs \cup \rules^\prime$ (see \cite{kop:raa:12}),
except that we did not mark the head symbol of the left-hand side; such a
marking can always be added at a later stage in a termination proof.
Correspondingly, a strong chain is a dependency pair chain in
which the DPs in $\Pstrong$ occur infinitely often.

\begin{lem}\label{lem:strongchain}
For ground terms \( s,t \), let \( s \succeq t \) if
$s\ (\arrztop{\reqs_\succ \cup \reqs_\succeq} \cup \arrz_\rules \cup\ 
\supterm)^*\ t$ and \( s \succ t \) if there exist \( u,v \) such that
\( s \succeq u \arrztop{\reqs_\succ} v \succeq t \).
Suppose no strong chain exists.

Then \( (\succ,\succeq) \) is a bounding pair that
orients all requirements in \( \reqs \).
\end{lem}

\begin{proof}
Clearly $\succeq$ is transitive and reflexive, as it is the
transitive-reflexive closure of the relation $\succeq_{\text{base}} :=
\arrztop{\reqs_\succ \cup \reqs_\succeq} \cup \arrz_\rules \cup\ \supterm$.
We have $\succ\;\subseteq\;\succeq$ because $\arrztop{\reqs_\succ}\;
\subseteq\;\succeq_{\text{base}}\;\subseteq\;\succeq$ and $\succeq$ is
transitive, while transitivity of $\succ$ follows because $a \succeq \cdot
\arrztop{\reqs_\succ} \cdot \succeq b \succeq \cdot \arrztop{\reqs_\succ}
\cdot \succeq c$ implies $a \succeq \cdot \arrztop{\Q_1} \cdot \succeq
\cdot \succeq \cdot \succeq_{\text{base}} \cdot \succeq c$ and $\succeq$ is
transitive.  The properties $\succ \cdot \succeq\;\subseteq\;\succeq$ and
$\succeq \cdot \succ\;\subseteq\;\succeq$ follow similarly.
It is also obvious that all requirements in \( \reqs \) are oriented.
All that remains to be shown is well-foundedness of \( \succ \).

Towards a contradiction, suppose that \( \succ \) is not well-founded; that is,
there is a sequence of \( \arrztop{\reqs_\succ} \cup \arrztop{\reqs_\succeq}
\cup \arrz_\rules \cup\ \supterm \) steps in which \( \arrztop{\reqs_\succ} \)
appears infinitely often.
We define that a term \( s \) is \emph{non-terminating} if there is an
infinite such sequence starting in \( s \), \emph{terminating} if there is
not, and \emph{minimal non-terminating} (MNT) if \( s \) is non-terminating
but all its strict subterms terminate (i.e., \( t \) is terminating whenever
\( s \supterm t \)).

Since $\succeq$ includes $\supterm$, we first make the following important
observation:

\emph{(OBS1) if $u$ is terminating and $u \supterm q$ then also $q$ is
terminating.}

Using this property along with an induction on the size of $s$, we easily
observe:

\emph{(OBS2) if $s\gamma$ is non-terminating but for all $x \in \Vars{s}$ the
term $\gamma(x)$ is terminating, then there is a non-variable term $t$ with $s
\suptermeq t$ such that $t\gamma$ is minimal non-terminating.}

If an MNT term $s = f\ u_1 \cdots u_n$ is reduced with $\arrztop{\reqs_\succ
\cup \reqs_\succeq}$ then $(f,n) \in \Defs$.  If it is reduced with
$\arrz_\rules$ at the root or head, then it could be reduced with
$\arrztop{\rules'}$, so also $(f,n) \in \Defs$.  If $s \supterm t$ then by
definition $t$ is no longer MNT.  And if $s \arrz_\rules t$ by some reduction
not at the head, then $t = f\ u_1 \cdots u_i' \cdots u_n$ with $u_i \arrz_\rules
u_i'$. So $t$ has the same outer shape, and is terminating or MNT.
This last point holds because $u_i'$ and each $u_j$ are terminating, as is
$f\ u_1 \cdots u_i' \cdots u_{n-1}$ (being equal to, or a reduct of, the
terminating subterm $f\ u_1 \cdots u_{n-1}$ of $s$).
Hence we conclude:

\emph{(OBS3) if $s = f\ u_1 \cdots u_n$ is MNT then $(f,n) \in \Defs$.}

Finally, again using (OBS1) we have:

\emph{(OBS4) for $\afun\ \ell_1 \cdots \ell_k \arrz r\ [\varphi] \in \reqs_\succ
\cup \reqs_\succeq \cup \rules^\prime$ and substitution $\gamma$: if
$\ell_i\gamma$ is terminating, then so is each $\gamma(\avar)$ with $\avar
\in \Vars{\ell_i}$.}

\smallskip
Towards obtaining the required contradiction, we will construct a strong
chain.  From the assumption that \( \succ \) is not well-founded, we know that
a non-terminating term exists, and use (OBS2) with $\gamma = []$ to obtain a
MNT term $t_0$.
In the following, denote $\arrzin{\rules}$ for a step \emph{not} at the root or
head of a term, i.e., $f\ s_1 \cdots s_j \cdots s_n \arrzin{\rules} f\ s_1
\cdots s_j' \cdots s_n$ if $s_j \arrz_\rules s_j'$.

Now, for natural number $i$, assume given a minimal non-terminating term $t_i$.
Then, denoting $\leadsto$ for the relation $(\arrztop{\reqs_\succeq} \cup
\arrz_\rules \cup\ \supterm)$, non-termination of $t_i$ implies that there is a
reduction $t_i = a_1 \leadsto a_2 \leadsto \dots \leadsto a_m
\arrztop{\reqs_\succ} b$, with $b$ still non-terminating.
We prove, by induction on $m-j$ for $j \in \{1,\dots,m\}$, that if $a_j$ is MNT
then there exist MNT terms  $s_{i+1}$ and $t_{i+1}$ such that $a_j\ 
(\arrzin{\rules} \cup \arrztop{\Pweak})^*\ s_{i+1}$ and $s_{i+1}
\arrztop{\Pstrong} t_{i+1}$.
\begin{itemize}[label=$\triangleright$]
\item If $j = m$ then $a_j = \ell\gamma$ for some $\ell \arrz r\ [\varphi] \in
  \reqs_\succ$; write $\ell = \afun\ \ell_1 \cdots \ell_k$.  Then $(\afun,k)
  \in \Defs$, and each $\ell_n\gamma$ is terminating by minimality.  By (OBS4),
  also all $\gamma(\avar)$ are terminating.  As $b = r\gamma$, by (OBS2),
  there is a non-variable subterm $r \suptermeq p$ such that $p\gamma$ is MNT.
  By (OBS3), $p\gamma$ is a candidate, and because $p$ is not itself a
  variable, this means $p$ must be a candidate.  Hence, we can let $s_{i+1} :=
  a_j = \ell\gamma$ and $t_{i+1} := p\gamma$.
\item If $j < m$ and $a_j \supterm a_{j+1}$, then $a_j' \suptermeq a_{j+1}$ for
  some immediate argument term $a_j'$ of $a_j$, and therefore $a_{j+1}$ would be
  terminating by definition of $a_j$ being MNT.  This is not possible as
  $a_{j+1}$ reduces to the non-terminating term $b$.
\item If $j < m$ and $a_j \arrz_\rules a_{j+1}$ by a step \emph{not} at the
  head, then we can write $a_j = \afun\ u_1 \cdots u_n \cdots u_m$ and
  $a_{j+1} = \afun\ u_1 \cdots u_n' \cdots u_m$.
  All arguments $u_l$ are terminating, and therefore so is $u_n'$.
  Also, $\afun\ u_1 \cdots u_{m-1}$ is terminating (as a strict subterm of
  $a_j$), so if $n' < m$ then $\afun\ u_1 \cdots u_n' \cdots u_{m-1}$ is
  terminating as well.  Hence, all strict subterms of $a_{j+1}$ are
  terminating, and $a_{j+1}$ is still MNT.
  Hence, by the induction hypothesis, $a_j \arrzin{\rules} a_{j+1}\ 
  (\arrzin{\rules} \cup \arrztop{\Pweak})^*\ s_{i+1}$ and $s_{i+1}
  \arrztop{\Pstrong} t_{i+1}$ for suitable $s_{i+1},t_{i+1}$.
\item Finally, if $j < m$ and either $a_j \arrz_\rules a_{j+1}$ at the head, or
  $a_j \arrztop{\reqs_\succeq} a_{j+1}$, then there is some $\ell \arrz r\ [\varphi]
  \in \rules' \cup \reqs_\succeq$ such that $a_j = \ell\gamma$ and $a_{j+1} =
  r\gamma$ and $\gamma$ respects $\varphi$.  Then by (OBS4) and minimality of
  $a_j$, all $\gamma(x)$ are terminating.  Since $a_{j+1}$ is necessarily
  non-terminating (as it reduces to the non-terminating term $b$), clearly
  $r$ is not a variable.  By (OBS2) there is a non-variable term $p$ with $r
  \suptermeq p$ such that $p\gamma$ is MNT.  By (OBS3) and the fact that $p$
  is not a variable, $p$ is a candidate.  There are two options:
  \begin{itemize}
  \item If $r = p$, then $\ell \arrz r\ [\varphi] \in \Pweak$.  Hence, using the
    induction hypothesis, $a_j \arrztop{\Pweak} a_{j+1}\ (\arrzin{\rules} \cup
    \arrztop{\Pweak})^*\ s_{i+1}$ and $s_{i+1} \arrztop{\Pstrong} t_{i+1}$ for
    some MNT $s_{i+1},t_{i+1}$.
  \item If $r \supterm p$, then $\ell \arrz p\ [\varphi] \in \Pstrong$.  Hence, we let
    $s_{i+1} := \ell\gamma$ and $t_{i+1} := p\gamma$.
  \end{itemize}
\end{itemize}
In particular (for $j = 1$), $t_i\ (\arrz_\rules \cup \arrztop{\Pweak})^*
s_{i+1}$ and $s_{i+1} \arrztop{\Pstrong} t_{i+1}$.  We have thus constructed a
strong chain, which contradicts the lemma's assumption that no strong chain
exists.  We conclude a contradiction with the assumption that $\succ$ is not
well-founded.
\end{proof}

Overall, if all elements of $\reqs$ have a shape $\csucc{\afun\ s_1 \cdots s_n}{
t}{\varphi}$ or $\csucceq{\afun\ s_1 \cdots s_n}{t}{\varphi}$ with $\afun$ not a
variable, then a suitable bounding pair for $\reqs$ exists if there is no strong
chain.

If the elements of $\reqs$ do not all have this form, we may still be able to
apply a similar approach (by changing $\reqs$ so the original requirements are
captured by the pair $(\succ,\succeq)$ defined in
\autoref{lem:strongchain}), but the techniques to do so are beyond the scope of
this section.

\paragraph*{Reduction triples}
To directly investigate the presence of strong chains we can use an extension of
the notion of reduction ordering:

\newcommand{\redgr}{>}
\newcommand{\redgeq}{\geq}
\newcommand{\redsim}{\gtrsim}

\begin{defi}\label{def:redtriple}
A reduction triple is a triple $(\redgr,\redsim,\redgeq)$ of three relations such that:
\begin{itemize}
\item $\redgr$ is a well-founded partial ordering on the set of terms
\item $\redsim$ is a quasi-ordering on the set of terms, such that $s \redgr t \redsim u$ implies
  $s \redgr u$
\item $\redgeq$ is a monotonic quasi-ordering on the set of terms, such that $s \redgeq t$ implies
  $s \redsim t$
\end{itemize}
\end{defi}

\begin{cor}\label{cor:redtriple}
There exists no infinite strong chain if there is a reduction triple
such that:
\begin{itemize}
\item for all $(\ell,p,\varphi) \in \Pstrong$ and all $\gamma$ that respect $\varphi$:
  $\ell\gamma \redgr p\gamma$
\item for all $(\ell,p,\varphi) \in \Pweak$ and all $\gamma$ that respect $\varphi$:
  $\ell\gamma \redsim p\gamma$
\item for all $\ell \arrz r\ [\varphi] \in \rules$ and all $\gamma$ that respect $\varphi$:
  $\ell\gamma \redgeq r\gamma$
\end{itemize}
\end{cor}

\begin{proof}
In a strong chain \( s_i \arrztop{\Pstrong} t_i\ (\arrz_\rules \cup \arrztop{\Pweak})^*\ s_{i+1}
\mid i \in \mathbb{N} \), each \( s_i > t_i \), and each \( t_i \redsim s_i \) (since each
\( \arrz_\rules \) step is included in \( \redgeq \) and therefore in \( \redsim \) by monotonicity,
and each \( \arrztop{\Pweak} \) step is included by definition).  Hence, if a strong chain exists,
we obtain an infinite decrease: \( s_1 > s_2 > \dots \), contradicting well-foundedness of
\( > \).
\end{proof}

The advantage of a reduction triple compared to a reduction ordering is that the
monotonicity requirement is replaced by \emph{weak} monotonicity: only $\geq$
needs to be monotonic.  It also allows us to more easily take advantage of the
difference between $\succ$ and $\succeq$ requirements.

Reduction triples can be constructed using for instance argument filterings or
weakly monotonic algebras.  While there has not yet been extensive research in
this direction for higher-order constrained rewriting, this is a natural
direction for future work.

\section{Obtaining new deduction rules more easily}\label{sec:newDedEasy}

In \autoref{sec:proofs}, we set out to prove soundness and completeness of
Bounded RI in a modular way: rather than merely proving that the system as given
in \autoref{fig:boundedRIrules} satisfies \autoref{theorem:soundnessRI}, we have
shown the stronger \autoref{theorem:boundedGroundConvertibilityRI}.  This means
that any extension of the system with derivation rules that satisfy both the
Preserving Bounds Property (\autoref{def:boundPreserve}) and the Base Soundness
Property (\autoref{def:BaseSound}), is still sound.  If these rules satisfy
the Completeness Property (\autoref{def:complete}), they can also be used in a
non-equivalence proof.

Let us discuss two naturally arising situations for which we can exploit this
approach to introduce a new deduction rule, and establish its soundness relatively easily.

\paragraph{Calculations}
In the definition of Bounded RI, we follow the approach of~\cite{hag:kop:24} for
the definition of the \simplify\ (and \hypothesis) rules: these rules are much
more basic than the corresponding rule in the first-order RI definition for
LCTRSs~\cite{fuh:kop:nis:17} which uses a more sophisticated definition of
``constrained reduction''.  Essentially, a slightly weaker version of the
\alter\ step is included in the \simplify\ rule of~\cite{fuh:kop:nis:17}.

The price for this simplicity, however, is that we often require extra steps,
especially when it comes to calculation rules.  For example, in
\autoref{sec:example} we encountered the equation (for the sake of clarity we
omit the bounding terms of the equation context):
\[
f\ i\ (\recdown\ f\ n\ (i-\symb{1})\ a) 
\approx 
\tailup \ f\ (n+\symb{1})\ i\ (f\ n\ a)\ 
[i \ge n]
\]
Here, we first had to apply \alter\ to obtain the equation
\[
f\ i\ (\recdown\ f\ n\ (i-\symb{1})\ a) 
\approx 
\tailup \ f\ (n+\symb{1})\ i\ (f\ n\ a)\ 
[i' = i - \symb{1} \wedge n' = n + \symb{1} \wedge i \ge n]
\]
before we could apply two \simplify\ steps, using calculation rules $i-\symb{1} \to i'\ [i' = i-\symb{1}]$ and $n+\symb{1} \to n'\ [n' = n+\symb{1}]$,
turning the equation into 
\[
f\ i\ (\recdown\ f\ n\ i'\ a)
\approx 
\tailup \ f\ n'\ i\ (f\ n\ a)\ 
[i' = i - \symb{1} \wedge n' = n + \symb{1} \wedge i \ge n]
\]
With constrained reductions, we could do the \simplify\ steps directly to obtain
the same result.
Fortunately, we can regain this facility by observing that the following deduction rule

\medskip\noindent
\textbf{(Calc)}
\newcommand{\RULEcalc}{
  \DEDUCRULE{(\eqs \uplus \{ \eqcon{\sterm}{C[s_1, \ldots, s_n]}{t}{\tterm}{\psi} \}, \hs)}
            {\begin{aligned}
            &n > 0 \\
            &s_1, \ldots, s_n \text{ theory terms}\\
            &x_1, \ldots, x_n \text{ distinct fresh}\\
            &\text{variables}
        \end{aligned}}
            {(\eqs \cup \{ \eqcon{\sterm}{C[x_1, \ldots, x_n]}{t}{\tterm}{\psi \! \wedge \! (x_1 = s_1) \! \wedge \! \ldots \! \wedge \! (x_n = s_n)} \}, \hs)}
}
\RULEcalc

\medskip\noindent
can be considered as a shortcut for an \alter\ step, followed by some \simplify\ steps, possibly followed again by some \alter\ steps. (Provided all variables in
each $s_i$ are in $\Vars{\psi}$, or theory sorts are inextensible.)
Hence, it immediately satisfies the two invariants. 

\begin{exa}
We may use \calc\ to replace an equation \( \afun\ (x+(y+\symb{1})) \approx t \)
by \( \afun\ z \approx t\ [z=x+(y+\symb{1})] \). To do so without \calc\ would
require four steps:
\[
\begin{aligned}
&\eqlabel{E}{0}\quad 
&\afun\ (x+(y+\symb{1}))
\approx t 
\quad 
&& &
\\
&\eqlabel{E}{1}\quad 
&\afun\ (x+(y+\symb{1}))
\approx t 
\quad 
&[q=y+\symb{1} \wedge z=x+q] 
&&\quad \text{by }\alter
\\
&\eqlabel{E}{2}\quad 
&\afun\ (x+q)
\approx t 
\quad 
&[q=y+\symb{1} \wedge z=x+q] 
&&\quad \text{by }\simplify
\\
&\eqlabel{E}{3}\quad 
&\afun\ z
\approx t 
\quad 
&[q=y+\symb{1} \wedge z=x+q] 
&&\quad \text{by }\simplify
\\
&\eqlabel{E}{4}\quad 
&\afun\ z
\approx t 
\quad 
&[z=x+(y+\symb{1})] 
&&\quad \text{by }\alter
\end{aligned}
\] 
This cannot be done with only one \alter\ step because the rules in $\calcrules$
do not have composite terms on their right-hand sides; there is no rule
$x+(y+1) \to z\ [z=x+(y+\symb{1})]$.
\end{exa}

\paragraph{Axioms}
We can use the two properties to go beyond just composite rules, though.
A particular example, from a rule that is implemented in
\emph{incremental Rewriting Induction (iRI)}~\cite[Section 5]{aot:06}, is the
use of \emph{axioms}: equations that are known to be consequences of $\rules$,
but for which we may not have a proof using Bounded RI.
Let $\mathcal{A}$ be a set of such axioms.
We distinguish the following two cases:

\begin{itemize}[label = $\triangleright$]
\item $\rules$ is ground confluent, and all equations in \( \mathcal{A} \) are
  ground convertible; that is, \( s\gamma \leftrightarrow_\rules^* t\gamma \)
  for all \( s \approx t\ [\varphi] \in \mathcal{A} \) and ground substitutions
  \( \gamma \) that respect the equation.
\item $\mathcal{A}$ is a set of equations such that every equation context in
  $\eqs_{\mathcal{A}} = \{ \eqcon{\ell}{\ell}{r}{r}{\varphi}\mid \ell \approx
  r\ [\varphi]\in \mathcal{A}\}$ is $\rules/\emptyset/\emptyset$-bounded ground
  convertible with respect to the \emph{same} ordering pair $(\succ, \succeq)$
  that is used in the RI derivation under consideration.
\end{itemize}

In either case, we introduce the following deduction rule:

\medskip\noindent
\newcommand{\RULEaxiom}{
  \DEDUCRULE{(\eqs \uplus \{\eqconsim{\sterm}{C[\ell\delta]}{t}{\tterm}{\psi}\}, \hs)}
            {
            \begin{aligned}
            &
            (\ell \approx r\ [\varphi]\in \mathcal{A} \text{ or } r \approx \ell\ [\varphi] \in \mathcal{A}) \text{ and } \psi \models^\delta \varphi\\
            &\sterm \succeq C[r \delta]\ [\psi]
            \end{aligned} 
            }
            {(\eqs \cup \{\eqconsim{\sterm}{C[r \delta]}{t}{\tterm}{\psi}\}, \hs)}
}
\textbf{(Axiom)}
\RULEaxiom

\begin{lem}
The \axiom\ rule preserves bounds.
\end{lem}

\begin{proof}
If $\eqconsim{\sterm}{C[\ell\delta]}{t}{\tterm}{\psi}$ is a bounded equation
context, then $\tterm \succeq t\ [\psi]$.  Since the derivation rule requires
$\sterm \succeq C[r \delta]\ [\psi]$, also 
$\eqconsim{\sterm}{C[r\delta]}{t}{\tterm}{\psi}$ is bounded.
\end{proof}

\begin{lem}
\axiom\ satisfies the Base Soundness Property.
\end{lem}

\begin{proof}
To see this, we will show that if
$\eqconsim{\sterm}{C[\ell\delta]}{t}{\tterm}{\psi}$ is a bounded equation
context and $\gamma$ a gsc substitution which respects it, then
\( C[\ell\delta] \gamma \bconvert{\sterm\gamma}{\tterm\gamma}{\rules}{\hs}{
\eqs'}
C[r\delta] \gamma \convertsingle{\sterm\gamma}{\tterm\gamma}{\rules}{\hs}{
\eqs'} t \gamma \) holds for any \( \eqs, \hs \) and
\( \eqs' = \eqs \cup \{\eqconsim{\sterm}{C[r \delta]}{t}{\tterm}{\psi}\} \).

The latter step is obvious: by \autoref{def:boundconvert}.\ref{boundconvert:eqs}
we only have to show that $\{\sterm \gamma, \tterm \gamma\} \succeq_{\text{mul}}
\{\sterm \gamma, \tterm \gamma\}$, which is trivially true.  
So consider the former step.

In the first case, where \( \rules \) is ground confluent and
\( \rweq{\ell\chi}{r\chi}{\rules} \) for all ground substitutions \( \chi \)
that respect the equation, we have in particular that
\( \ell \delta \gamma \rw^* \cdot \rwleftstar r \delta \gamma\).
Since \( \rw \) is included in \( \succeq \) by definition of a bounding pair,
\( \ell\delta\gamma \bconvert{\ell\delta\gamma}{r\delta\gamma}{\rules}{
\emptyset}{\emptyset} r\delta\gamma \) holds by
\autoref{def:boundconvert}.\ref{boundconvert:rule}.
In the second case, 
we have this property by definition.
Either way, \( C[\ell\delta]\gamma \bconvert{\ell\delta\gamma}{r\delta\gamma}{
\rules}{\emptyset}{\emptyset} C[r\delta]\gamma \) holds by
\autoref{lem:convertContext}, and since this derivation can only use steps with
\( \rules \),
\( C[\ell\delta]\gamma \bconvert{\ell\delta\gamma}{r\delta\gamma}{\rules}{\hs}{
\eqs'} C[r\delta]\gamma \) follows regardless of \( \hs \) and \( \eqs' \).
Observe that \( \sterm\gamma \succeq C[\ell\delta]\gamma \succeq
\ell\delta\gamma \) since the original equation context is bounded and \(
\succeq \) includes \( \supterm \); and that \( \tterm\gamma \succeq C[r\delta]
\gamma \succeq r\delta\gamma \) by definition of the rule.  Hence, we apply
\autoref{lem:convertMul} to increase the bounds and obtain
\( C[\ell\delta]\gamma \bconvert{\sterm\gamma}{\tterm\gamma}{\rules}{\hs}{
\eqs'} C[r\delta]\gamma \).
\end{proof}

Hence, both \calc\ and \axiom\ can be used in equivalence proofs.  In addition,
to see that they can also be used in non-equivalence proofs, we must prove that
they satisfy the Completeness Property as well.  This is easy.
In the case of \calc\ this holds because the step is a concatenation of rules
with the Completeness Property.
In the case of \axiom, we observe that if \( \gamma \) respects \( \psi \),
then \( C[r\delta]\gamma \leftrightarrow_{\mathcal{A}} C[\ell\delta]\gamma
\leftrightarrow_{\eqs} t\gamma \), and by the assumption on \( \mathcal{A} \)
and the premise of the Completeness Property, both steps are included in
\( \leftrightarrow^*_\rules \).

\section{Ground Confluence} 
\label{sec:GroundConfluence} 

We now turn our eye to \emph{ground confluence}.  This property is both needed to deduce that
an equation is \emph{not} an inductive theorem using rewriting induction
(using \autoref{subsec:completeness}),
and interesting to study in its own right -- since, in a
terminating and ground confluent system, reduction essentially computes a unique result for every
ground term.

Existing work on traditional rewriting~\cite{aot:toy:16} has shown that RI can be used as a tool
to prove ground confluence of a terminating TRS: a first-order, unconstrained term
rewriting system is ground confluent if all its critical pairs are bounded ground convertible, and
RI can be used as a method to prove bounded ground convertibility.
The existing work uses a slightly different definition of bounded ground convertibility than the one employed in
this paper.
We extend this result to our more liberal definition of bounded ground convertibility and to higher-order constrained
rewriting.
First, in \autoref{sec:CriticalPairs}, we give a definition of critical pairs and the closely
related notion of critical \emph{peaks} for \lcstrss, and establish the
relation between (ground) confluence and critical peaks.  In \autoref{sec:RIforGroundConfluence}
we generalize the first-order result~\cite{aot:toy:16}, showing that, for terminating \lcstrss, ground confluence is implied by bounded ground convertibility of all critical pairs.
Since Bounded RI proves bounded ground convertibility, it can therefore be used to obtain
ground confluence.

\subsection{Critical pairs}\label{sec:CriticalPairs}

\emph{Critical pairs} are a standard tool for proving confluence and have also been defined for
first-order \lctrss~\cite{kop:nis:13,sch:mid:23}. The latter definition relies on the notion of
\emph{positions}, which we here adapt to our (higher-order) definition of terms.

\begin{defi}\label{def:positions}
For a term $t = a\ t_1 \cdots t_n$ with $a \in \Sig \cup \Var$ and $n \geq 0$, 
the set of positions $\Pos(t)$ is defined as
$\Pos(t) = (\bigcup_{i=0}^n \{ \star i \}) \cup (\bigcup_{i=1}^n \{ i \cdot p \mid p \in \Pos(t_i) \})$.

Define the subterm $t|_p$ of $t$ at position $p \in \Pos(t)$ as follows:
\begin{itemize}[label=$\triangleright$]
\item $(a\ t_1 \cdots t_n)|_{\star i} = (a\ t_1 \cdots t_{n-i})$
\item $(a\ t_1 \cdots t_n)|_{i \cdot p} = t_i|_p$
\end{itemize}
If $t$ is a term, $p \in \Pos(t)$, and $s$ is a term of the same type as $t|_p$, then let $t[s]_p$ be
the term obtained from $t$ by replacing $t|_p$ by $s$.
We will denote the position \( \star 0 \) as \( \epsilon \).
\end{defi}

\autoref{def:positions} differs from the definition for first-order rewriting by including positions
for partial applications: in first-order terms, we only have positions of the form \( i_1 \cdots i_n
\star 0 \).\footnote{
  This definition also differs from the usual presentation of positions in higher-order rewriting,
  which does not use the \( \star \) notation but rather sets \( (s\ t)|_{1p} = s|_p \) and
  \( (s\ t)|_{2p} = t|_p \).  This difference is not significant.  Rather, the presentation above is
  chosen because in practice (and in our tool) it is convenient if the positions used for
  first-order rewriting refer to the same subterm when viewing the term as higher-order.
}

We now recall the first-order definition of critical pairs~\cite{sch:mid:23} (only adapting it to
omit the set $\mathcal{E}\!\mathit{Var}$ which is redundant here since we have required
\( \Vars{r_i} \subseteq \Vars{\ell_i} \cup \Vars{\varphi_i} \)).

\begin{defi}[Critical pairs]\label{def:critPair}
An \emph{overlap} in an \lctrs\ $\rules$ is a triple $\langle \rho_1, p, \rho_2 \rangle$ with rules  
$\rho_1 := \ell_1 \to r_1 \ [\varphi_1]$ and $\rho_2 := \ell_2 \to r_2 \ [\varphi_2]$ and position
\( p \), satisfying 
\begin{itemize}[label=$\triangleright$]
\item $\rho_1$ and $\rho_2$ are variable-disjoint renamings of rewrite rules in $\rules \cup \calcrules$
\item $p \in \Pos(\ell_2)$ with $\ell_2|_p$ not a variable
\item $\sigma = mgu(\ell_1, \ell_2|_p)$ 
and $\sigma(\Vars{\varphi_1} \cup \Vars{\varphi_2}) \subseteq \Val \cup \Var$
\item $\varphi_1 \sigma \wedge \varphi_2 \sigma$ is satisfiable 
\item If $p=\epsilon$ and $\rho_1$ and $\rho_2$ are renamings of the same rule, then $\Vars{r_1} \setminus \Vars{\ell_1} \ne \emptyset$.
\end{itemize}
The critical pair of $\langle \rho_1, p, \rho_2 \rangle$ is the equation 
\[
\ell_2[r_1]_p\sigma \approx r_2 \sigma \ [\varphi_1 \sigma \wedge \varphi_2 \sigma]
\] 
\end{defi} 

The ``overlap'' here refers to the fact that $\ell_2$ has a subterm that unifies with $\ell_1$. 
Therefore,  $\ell_2 \sigma$ can be reduced in two ways: if $\sigma$ respects $\varphi_1 \wedge \varphi_2$ then, as $\ell_2 \sigma = (\ell_2 \sigma)[\ell_2|_p \sigma]_p=(\ell_2 \sigma)[\ell_1 \sigma]_p$, we apply $\rho_1$ to obtain $\ell_2 \sigma \rw (\ell_2 \sigma)[r_1 \sigma]_p$. 
Of course, we also have   
$\ell_2 \sigma \rw r_2 \sigma$. 

\begin{exa}\label{ex:critPairFirstOrder}
Let us apply \autoref{def:critPair} to \autoref{example:nonEquivalence}. 
We have the following variable-disjoint copies of $\symb{H}$-rules, which are not renamings of each other 
\[
\begin{aligned}
\rho_1
:=
\symb{H} \ f \ n\ m\ x 
&\to 
\symb{H} \ f \ (n-\symb{1})\ m\ (f\ x) 
&& [n > \symb{0}]\\
\rho_2
:=
\symb{H} \ g \ i\ j\ y 
&\to 
\symb{H} \ g \ (j-\symb{1})\ i\ (g\ y) 
&& 
[j > \symb{0}]
\end{aligned}
\] 
There is overlap $\langle \rho_1, \epsilon, \rho_2 \rangle$ with 
$\sigma = [g:=f, i:=n, j:=m, y:=x]$.
This yields a critical pair 
\[
\symb{H} \ f \ (n-\symb{1})\ m\ (f\ x)
\approx 
\symb{H} \ f \ (m-\symb{1})\ n\ (f\ x)
\
[n>\symb{0} \wedge m>\symb{0}]
\] 
\end{exa} 

\begin{rem}
In \autoref{ex:critPairFirstOrder} we did not list \emph{all} critical pairs.
Since the overlap occurs at the root position $\epsilon$, there is also an overlap
$\langle \rho_2, \epsilon, \rho_1 \rangle$ yielding the critical pair
\(
\symb{H} \ f \ (m-\symb{1})\ n\ (f\ x)
\approx 
\symb{H} \ f \ (n-\symb{1})\ m\ (f\ x)
\
[n>\symb{0} \wedge m>\symb{0}]
\).  For confluence techniques based on symmetric joinability conditions of critical pairs -- like
the one we will consider in this paper -- we only need one of the two.  For other confluence
techniques, it is important to consider both (see, e.g.,~\cite{sch:mid:23}).

Moreover, every calculation rule $\afun\ x_1 \cdots x_m \to y \ [ y = \afun\ x_1 \cdots x_m ]$
yields a critical pair with itself due to the overlap at the root:
$y \approx z\ [y = \afun\ x_1 \cdots x_m \wedge z = \afun\ x_1 \cdots x_m]$.
However, this is merely an artefact of the definition.
In practice such critical pairs cause no issues because every substitution $\gamma$ that
respects $y = \afun\ x_1 \cdots x_m \wedge z = \afun\ x_1 \cdots x_m$ has
$\gamma(y) = \gamma(z)$.
\end{rem}

While this definition was designed for \lctrss\ -- so first-order systems where variables cannot
be applied, and partially applied function symbols do not occur -- our updated definition of
positions allows the definition to be used unaltered for (higher-order) \lcstrss!

\begin{exa}\label{ex:critPairHigherOrder}
Consider an \lcstrs\ containing the following variable-disjoint rules 
\[
\rho_1:=\symb{f} \to \symb{g}\quad\quad\quad
\rho_2:=\symb{u}\ (\symb{f}\ x) \to \symb{h}\ x
\]
We have an overlap \( \langle\rho_1,1 \star 1,\rho_2\rangle \) with \( \sigma \) the empty substitution.  This yields the
critical pair \( \symb{u}\ (\symb{g}\ x) \approx \symb{h}\ x \), because
\( \ell_2\sigma = \symb{u}\ (\symb{f}\ x) \) reduces to both sides.
\end{exa}

In some cases, it is useful to also preserve the source $\ell_2 \sigma$.
We do so using \emph{critical peaks}.

\begin{defi}[Critical peaks]
The critical peak of an overlap $\langle \rho_1, p, \rho_2 \rangle$ is the tuple
\[
    \peak{\ell_2 \sigma}
         {\ell_2[r_1]_p\sigma}
         {r_2 \sigma}
         {\varphi_1 \sigma \wedge \varphi_2 \sigma}
\]
We let $\critpairs{\rules}$ be the set of all critical peaks generated by overlaps
$\langle \rho_1,p,\rho_2 \rangle$ where $\rho_2 \in \rules$.
\end{defi}

\begin{exa}
Following \autoref{ex:critPairFirstOrder}, we have $\ell_2 \sigma = \symb{H}\ f\ n\ m\ x$,
so the corresponding critical peak is given by
\[
\peak{\symb{H}\ f\ n\ m\ x}
     {\symb{H} \ f \ (n-\symb{1})\ m\ (f\ x)}
     {\symb{H} \ f \ (m-\symb{1})\ n\ (f\ x)}
     {n>\symb{0} \wedge m>\symb{0}}
\]
We also have a critical peak
\[
\peak{\symb{H}\ f\ n\ m\ x}
     {\symb{H} \ f \ (m-\symb{1})\ n\ (f\ x)}
     {\symb{H} \ f \ (n-\symb{1})\ m\ (f\ x)}
     {n>\symb{0} \wedge m>\symb{0}}
\]
but we do \emph{not} obtain critical peaks for the calculation rules: we excluded these by requiring
that the second rule in the overlap must be from $\rules$.
(The first rule may still be from $\calcrules$, so we would for instance get a critical pair when
$\rho_2$ is a rule like $\afun\ (\symb{1}\ \symb{+}\ y) \arrz \bfun\ y$.)

In \autoref{ex:critPairHigherOrder} we obtain only one critical peak
\[
\peak{\symb{u}\ (\symb{f}\ x)}
     {\symb{u}\ (\symb{g}\ x)}
     {\symb{h}\ x}
     {\strue}
\]
\end{exa}

Critical pairs and peaks rely on positions, which we have otherwise not used in our definitions.
To relate these to our notion of reduction, we use the following lemma.

\begin{lem}
\label{lem:forCriticalPair}
Let $\gamma$ be a substitution, $X$ a finite set of variables, and $s,t \in \Terms$ such that
$s \gamma \rw t$ and $X \supseteq \Vars{s}$.
Then one of the following holds.
\begin{enumerate}[(i). ]
\item There exist a substitution $\delta$, variable-renamed rule $\myrule \in \rules \cup \calcrules$
  with variables disjoint from $X$, and $p \in \Pos(s)$ such that $\delta(x) = \gamma(x)$ for all $x
  \in X$, $\delta$ respects $\varphi$, $s|_p \notin \Var$, $s|_p\delta = \ell\delta$ and
  $t=s[r]_p\delta$.
\item There exist $x \in \Vars{s}$ and substitution $\gamma'$ such that $\gamma(x) \rw \gamma'(x)$,
  and $\gamma'(y)=\gamma(y)$ for all other $y$ and $t \rw^* s \gamma'$.
\end{enumerate}
\end{lem}
\begin{proof}
Assume $s \gamma \rw t$. 
By induction on the term-shape of $s$ we prove that either (i) or (ii) holds.
Any $s \in \Terms$ can be written as $s=a\ s_1 \cdots s_n$ with $a \in \Sigma \cup \Var$ and $n \ge 0$. 

\begin{itemize}[label=$\triangleright$]
\item If $a \in \Var$ and $\gamma(a) \rw u$ and $t = u\ (s_1\gamma) \cdots (s_n\gamma)$, we
  define $\gamma'$ as follows:
  \[
  \gamma'(a) = u 
  \qquad\qquad
  \text{and}
  \qquad \qquad 
  \gamma'(y) = \gamma(y)
  \text{ if } y \neq a
  \]
  Then clearly $t = u\ (s_1\gamma) \cdots (s_n\gamma) \rw^* u\ (s_1\gamma') \cdots (s_n\gamma') = s\gamma'$;
  thus, (ii) holds.

\item If the reduction occurs at the root or head of $s\gamma$, then $s \gamma = (\ell \sigma)\ 
  (s_{i+1} \gamma) \cdots (s_n \gamma) \rw (r \sigma)\ (s_{i+1} \gamma) \cdots (s_n \gamma)=t$
  for some $i \in \{0,\dots,n\}$, some $\myrule \in \rules \cup \calcrules$, and some $\sigma$ which
  respects this rule.
  If $i = 0$ and $a$ is a variable we are in the first case and are done.
  Otherwise, let $p = \star (n-i)$.
  As we can safely assume (by renaming) that the variables in \( \myrule \) do not occur in \( X \), we can
  define $\delta(x) = \gamma(x)$ for $x \in X$ and $\delta(x) = \sigma(x)$ for all other $x$.
  Then (i) holds.

\item The only remaining option is a reduction in some subterm $s_i$ with $1 \le i \le n$:
  $s_i \gamma \rw t_i$ and $t=(a\gamma)\ (s_1\gamma) \cdots t_i \cdots (s_n\gamma)$.
  By the induction hypothesis there are two possibilities
  \begin{enumerate}[(i). ]
  \item There exist \( \delta,\myrule \) and \( p_i \in \Pos(s_i) \) such that \( \delta \) respects \( \varphi \),
    $\delta(x) = \gamma(x)$ for all $x \in X$ (and therefore $s\gamma = s\delta$), $s_i|_{p_i} \notin \Var$,
    $s_i|_{p_i} \delta = \ell\delta$ and $t_i = s_i[r]_{p_i}\delta$.
    Let $p=i\ p_i$.
    Then (i) also holds for $s \gamma \rw t$, since $s|_p = s_i|_{p_i}$, and $t = (a\gamma)\ (s_1\gamma) \cdots t_i
    \cdots (s_n\gamma) = (a\delta)\ (s_1\delta) \cdots (s_i[r]_{p_i}\delta) \cdots (s_n\delta) =
    s[r]_p\delta$.
  \item There exist $x \in \Vars{s_i}$ and substitution $\gamma'$ such that $\gamma(x) \rw \gamma'(x)$,
    $\gamma'(y)=\gamma(y)$ for all $y \in \domain(\gamma)\setminus \{x\}$ and $t_i \rw^* s_i \gamma'$.
    We have $s_j \gamma \rw^* s_j \gamma'$ for all $j \ne i$.
    Conclude $t = (a\gamma)\ (s_1\gamma) \cdots t_i \cdots (s_n\gamma) \rw^* 
    (a\gamma)\ (s_1\gamma) \cdots (s_i\gamma') \cdots (s_n\gamma) \rw^*
    (a\gamma')\ (s_1\gamma') \cdots (s_i\gamma') \cdots (s_n\gamma') = s\gamma'$.
    \qedhere
  \end{enumerate}
\end{itemize}
\end{proof}

In analogy to the Critical Pair Lemma~\cite[Lemma 6.2.3]{baa:nip:98}
we introduce the Critical Peak Lemma.
\begin{lem}[Critical Peak Lemma]\label{lem:CriticalPair}
If both $s \rw t_1$ and $s \rw t_2$, then either 
\begin{enumerate}[(i). ]
\item There exists a term $t$ such that $t_1 \rw^* t \rwleftstar t_2$. 
\item 
There exist a critical peak 
\(
  \peak{a}{b}{c}{\varphi}
\), context $C$ and substitution $\delta$ that respects  
$\varphi$
such that $s = C[a\delta]$ and either
$t_1 = C[b\delta]$, $t_2 = C[c\delta]$, or
$t_1 = C[c\delta]$, $t_2 = C[b\delta]$.
\end{enumerate}
\end{lem}  

\begin{proof}
By induction on the term-shape of $s$ we prove that either (i) or (ii) holds.
Any $s \in \Terms$ can be written as $s=f\ s_1 \cdots s_n$ with $f \in \Sigma \cup \Var$ and $n \ge 0$. 
We consider all possibilities of $t_1 \rwleft s \rw t_2$.
\begin{enumerate}[(1). ]
\item There exist $1 \le i,j \le n$ and terms $v_i, w_j$ such that $s_i \rw v_i$, $s_j \rw w_j$ and 
\[
t_1=s[v_i]_i
\qquad \text{ and }\qquad 
t_2=s[w_j]_j
\]
\begin{itemize} 
\item 
If $i\neq j$ then the reductions are in parallel positions; (i) holds with
  $t=t_1[w_j]_j=t_2[v_i]_i$.
\item If $i=j$ then $v_i \rwleft s_i \rw w_i$. 
By induction hypothesis we either have 
\begin{itemize}
\item There is a term $t'$ such that $v_i \rw^* t' \rwleftstar w_i$. 
Then (i) holds: take $t=s[t']_i$. 
\item There exist a critical peak 
\(
  \peak{a}{b}{c}{\varphi}
\), context $C'$ and substitution $\delta$ that respects
$\varphi$
such that 
$s_i = C'[a\delta]$, $v_i = C'[b\delta]$ and $w_i = C'[c\delta]$.
Then (ii) holds: take $C[\square]=s[C'[\square]]_i = f\ s_1 \cdots s_{i-1}\ 
C'[\square]\ s_{i+1} \cdots s_n$.
\end{itemize} 
\end{itemize}
\item One of the reductions is at the head.  Without loss of generality, we assume it is the step $s \rw t_1$.
  That is, there exist a rule $\rho_1:=\myrule \in \rules \cup \calcrules$ with $\ell = \afun \ v_1 \cdots
  v_k$ and a substitution $\gamma$ which respects $\varphi$, and $\vec{s}=s_{k+1} \cdots s_n$ such that
  $k=\arity(\afun) \le n$ and  
\[
s
=
(\ell \gamma)\vec{s}
=
\afun\ 
(v_1 \gamma) \cdots (v_{k} \gamma) \
\vec{s}
\rw 
(r \gamma)
\vec{s}
=
t_1 
\]
We consider the possibilities for $s \rw t_2$
\begin{itemize} 
\item $s_i \rw s'$, for some $i > k$, and $t_2 = s[s']_i$.
  Then (i) holds, as we can see by choosing
  $t=(r\gamma)\ s_{k+1} \cdots s_{i-1}\ s'\ s_{i+1} \cdots s_n$.
\item $s_i = v_i \gamma \rw s'$, with $1 \le i \le k$, and $t_2 = s[s']_i$. 
By \autoref{lem:forCriticalPair}, applied with $X = \Vars{\ell} \cup \Vars{r} \cup \Vars{\varphi} \supseteq
\Vars{v_i}$, there are two possibilities: 
\begin{itemize}
\item $s_i = v_i\gamma = v_i\delta$ for some $\delta$ that respects a renamed rule $\rho_2:=\ell' \to r'\ [\varphi']$
  with $\ell'=\symb{g}\ l'_1 \cdots l'_m$ and there is a $p' \in \Pos(v_i)$ such that $v_i|_{p'} \notin \Var$,
  $v_i|_{p'} \delta = \ell'\delta$ and $s'=v_i[r']_{p'}\delta$; moreover, $\delta(x) = \gamma(x)$ on all $x \in
  X$.
  Then we observe that $\rho_1 \notin \calcrules$ (as $v_i$ is not a variable), and
  $t_2 = \afun\ (v_1\gamma) \cdots (v_{i-1}\gamma)\ (v_i[r']_{p'}\delta)$
  $(v_{i+1}\gamma) \cdots (v_k\gamma)\ 
  \vec{s} = (\afun\ v_1 \cdots v_{i-1}\ v_i[r']_{p'}\ v_{i+1} \cdots v_k)\delta\ \vec{s}$.
  We will see that (ii) holds: let $C= \square\ \vec{s}$ and $p=i\ p'$.
  We have
\[
t_2 
= 
C\left[\ell[r']_p \delta\right] 
\rwleft
s
=
C[\ell \delta]
\rw
t_1=C[r \delta] 
\]

Since \( \delta \) is a unifier of \( \ell|_p = v_i|_{p'} \) and \( \ell' \), there exist a most
general unifier \( \sigma \) and a substitution \( \chi \) so that \( \delta = \sigma\chi \).
Since \( \delta \) respects both \( \varphi \) (because it corresponds with \( \gamma \) on \( X
\supseteq \Vars{\varphi} \)) and \( \varphi' \), clearly \( (\varphi' \wedge \varphi)\sigma \) is
satisfiable (by the substitution \( \chi \)).
Hence, the peak above is an instance of the critical peak
$\peak{\ell\sigma}{\ell[r']_p\sigma}{r\sigma}{\varphi'\sigma \wedge \varphi\sigma}$
for the overlap $\langle \rho_2,p,\rho_1 \rangle$.

\item There exist $x \in \Vars{v_i}$ and substitution $\gamma'$ such that $\gamma(x) \rw \gamma'(x)$,
$\gamma'(y)=\gamma(y)$ for all $y \neq x$ and $s' \rw^* v_i \gamma'$.
Then (i) holds: choosing $t=(r \gamma')\vec{s}$ we have
\begin{center}
\begin{tikzcd}[column sep=small,
row sep=1.0em, cells={scale=1.0, transform shape}]
& 
s= (\ell \gamma)\vec{s}
\arrow[dl, " ", "\rules "', " " near end] 
\arrow[dr, "\rules ", ""', "" near end]  
& 
\\
t_1=(r \gamma)\vec{s} 
\arrow[ddr, " ", "\rules"', "*"]  
&
& 
t_2=(\ell \gamma)[s']_i \vec{s} = \afun\ (v_1\gamma) \cdots s' \cdots (v_n\gamma)\ \vec{s}
\arrow{d} 
\arrow[d, "*", "\rules"', "" near end] 
\\
&
& 
(\ell \gamma')\vec{s}=\afun\ (v_1\gamma') \cdots (v_i\gamma') \cdots (v_n\gamma')\ \vec{s}
\arrow[dl, "\rules", " "', " " near end]
\\
& 
t=(r \gamma')\vec{s} 
\end{tikzcd}
\end{center}
\end{itemize}
\item The other reduction is also at the head. 
Since each function symbol has a fixed arity, this
  means there exist a rule $\rho_2:=\ell'\to r' \ [\varphi']$ with $\ell'=\afun \ v'_1 \cdots v'_k$
  and substitution $\gamma'$ which respects $\varphi'$ such that 
$s = (\ell' \gamma')\ \vec{s}$ and 
$t_2 = (r' \gamma')\ \vec{s}$.
There are two cases.
\begin{itemize}
\item
  If $r\gamma = r'\gamma'$, then (i) holds with $t = t_1 = t_2$.
  Note that this is always the case if $\rho_2$ is a calculation rule: then $\afun$ is a calculation symbol
  and each $v_i\gamma = v_i'\gamma'$ a value.  Since left-hand sides of rules in $\rules$ may not
  be theory terms, necessarily $\rho_1 \in \calcrules$ and therefore we can only have $\rho_1 =
  \rho_2$.  Since calculation is deterministic, $r\gamma = r'\gamma'$.
\item
  Otherwise, (ii) holds with $C[\square] = \square\ \vec{s}$.
By renaming variables we can safely assume that $\domain(\gamma) \cap \domain(\gamma') = \emptyset$,
so $\delta = \gamma \cup \gamma'$ is well-defined. Thus
\[
t_1=C[r \delta] 
=C[\ell'[r]_\epsilon \delta]
\rwleft
s=C[\ell' \delta]
\rw
t_2 = C[r' \delta] 
\]
Let $\sigma = \mgu{\ell}{\ell'}$, and let $\chi$ such that $\delta = \sigma\chi$.
Then the peak above corresponds to the critical peak $\peak{\ell'\sigma}{r\sigma}{r'\sigma}{
\varphi\sigma \wedge \varphi'\sigma}$.
(If $\rho_1$ and $\rho_2$ are renamings of the same rule, then $\ell\sigma = \ell'\sigma$ implies
$r\sigma = r'\sigma$ \emph{unless} $\Vars{r} \subseteq \Vars{\ell}$ does not hold, so the last
requirement of an overlap is also satisfied.)
\qedhere
\end{itemize}
\end{itemize} 
\end{enumerate}
\end{proof}

The critical peak lemma is valuable both for proving (local) confluence and ground (local)
confluence.  In this paper, we focus on ground confluence.  The following lemma gives a way to
prove this property.

\begin{lem}\label{lem:GroundConfluenceExistsMinimal}
Assume given an \lcstrs\ $\alcstrs=(\Sigma, \rules)$ which is not ground locally confluent, and  
let $(\succ, \succeq)$ be a bounding pair on $T(\Sig,\emptyset)$. 
Then 
\begin{enumerate}[(i). ]
\item There is a $\succ$-minimal ground term $s$ for which there exist terms $t_1, t_2$ with
  $t_1 \rwleft s \rw t_2$ but such that there is no $t$ with $t_1 \rw^* t \rwleftstar t_2$. 
\\
Moreover, there exist a critical peak
$\peak{u}{v_1}{v_2}{\varphi}$ and ground substitution $\gamma$ that respects $\varphi$ 
such that $s=u \gamma$, $t_1=v_1 \gamma$ and $t_2=v_2\gamma$. 
\item If in addition $\rw$ is terminating, then for all ground terms $r$ with $s \succ r$:
  if there are $q_1, q_2$ with
  $q_1 \rwleftstar r \rw^* q_2$ then there is a term $q$
  with $q_1 \rw^* q \rwleftstar q_2$. 
\end{enumerate} 
\end{lem}

\begin{proof}
We will introduce a well-founded ordering $\sqsupset$ that includes $\succ$ and prove a slightly
stronger result than required. 
For (i), we will show the existence of a $\sqsupset$-minimal element (rather than merely a
$\succ$-minimal element). 
For (ii), we will prove the result for every $r$ satisfying $s\sqsupset r$ (rather than for $r$ satisfying $s \succ r$). 
We define $\sqsupset$ as follows:

If $\rw$ is terminating, let $\sqsupset$ be the relation with $a \sqsupset b$ if $a \succ b$ or
$a\ (\rw^+ \cdot \suptermeq)\ b$ or $a \supterm b$.  
If $\rw$ is not terminating, then let $\sqsupset$ be the relation with $a \sqsupset b$ if $a \succ
b$ or $a \supterm b$.
Either way, this relation is well-founded, because $a \rw b$ and $a \supterm b$ both imply $a
\succeq b$, and if $\rw$ is terminating then so is the union of $\rw$ and $\supterm$.

\begin{enumerate}[(i). ]
\item
  By definition of not being ground locally confluent, there are ground terms $s, t_1, t_2$ with
  $t_1 \rwleft s \rw t_2$, for which there is no $t$ such that $t_1
  \rw^* t \rwleftstar t_2$.
  Hence, we can take $s$ to be a term with this property that is minimal with respect to $\sqsupset$.

  By \autoref{lem:CriticalPair} there exist a critical peak $\peak{u}{v_1}{v_2}{\varphi}$,
  context $C$ and substitution $\gamma$ that respects $\varphi$ 
  such that $s=C[u \gamma]$, $t_1=C[v_1 \gamma]$ and $t_2=C[v_2\gamma]$.
  Since $u\gamma$ is ground -- and therefore its reducts $v_1\gamma$ and $v_2\gamma$ as well --
  necessarily $\gamma$ is ground on $\Vars{u,v_1,v_2}$.  
  Thus, we can safely assume that $\gamma$ is ground
  overall.
  Furthermore,  $C=\square$ because otherwise $s \sqsupset u \gamma$, contradicting
  minimality of $s$.  

\item
  Assume $\rw$ is terminating.
  Let $r, q_1, q_2$ be ground terms with $s \sqsupset r$ and
  $q_1 \rwleftstar r \rw^* q_2$.
  We show that there is a term $q$ such that $q_1 \rw^* q \rwleftstar q_2$, using induction on $\rw^+$. 
  
  If $q_1=r$ or $q_2=r$ then there is nothing to prove, so we might assume there are terms $w_1, w_2$ such that
  $q_1 \rwleftstar w_1 \rwleft r \rw w_2 \rw^* q_2$. 
  Since $s \sqsupset r$ and $s$ is a minimal term, there exists a term $w$ such that
  $w_1 \rw^* w \rwleftstar w_2$. 

  But then $q_1 \rwleftstar w_1 \rw^* w$ and $s \sqsupset w_1$
  (because $\rw \ \subseteq\ \succeq$), so by induction hypothesis there is a term $a$
  such that $q_1 \rw^* a \rwleftstar w$.
  In a similar way, we can apply the induction hypothesis to find a term $b$ such that $w \rw^* b \rwleftstar q_2$. 
  We complete the proof like in the diagram below: since both $w \rw^* a$ and $w \rw^* b$, one additional application of
  the induction hypothesis will give us the term $q$.

\begin{center}
\begin{tikzcd}[column sep=small, 
row sep=1.6em, 
cells={scale=1.0, transform shape}]
&  
&
r
\arrow[dl, "\rules", ""', " " near end] \arrow[dr, "", "\rules"', "" near end]
& 
& 
\\
&
w_1
\arrow[dl, "\rules", "*"', " " near end]
\arrow[dr, dashed, "*", "\rules"', "" near end]
& 
&
w_2
\arrow[dr, "*", "\rules"', "" near end]
\arrow[dl, dashed, "\rules", "*"', " " near end]
&
\\
q_1
\arrow[dr, dashed, "*", "\rules"', "" near end]
&    
& 
w
\arrow[dr, dashed, "*", "\rules"', "" near end]
\arrow[dl, dashed, "\rules", "*"', " " near end]
&
&
q_2
\arrow[dl, dashed, "\rules", "*"', " " near end]
\\
&
a
\arrow[dr, dashed, "*", "\rules"', "" near end]   
& 
&
b
\arrow[dl, dashed, "\rules", "*"', " " near end]
&
\\
&     
&
q
&
&
\end{tikzcd}
\end{center}  
\end{enumerate}
\end{proof}

\subsection{RI for ground confluence}\label{sec:RIforGroundConfluence}

Recall that our goal, following the procedure of~\cite{aot:toy:16}, is to show that (1) an \lcstrs\ 
is ground confluent when all its critical pairs are bounded ground convertible, and that (2) RI can
be used as a method to prove bounded ground convertibility.
For point (1), we introduce the Ground Critical Peak Theorem, analogously to the well-known
Critical Pair Theorem~\cite[Theorem 6.2.4]{baa:nip:98}.

\begin{thm}[Ground Critical Peak Theorem]\label{thm:groundCriticalPair}
Let $\rules$ be a terminating \lcstrs\ and let $(\succ, \succeq)$ be a bounding pair.
Then $\rules$ is ground confluent if for all its critical peaks \( \peak{u}{v_1}{v_2}{\varphi} \),
we have $\csucc{u}{v_1}{\varphi}$, $\csucc{u}{v_2}{\varphi}$, and
the equation context \( \eqcon{v_1}{v_1}{v_2}{v_2}{\varphi} \) is
\( \rules/\emptyset/\emptyset \)-bounded ground convertible using $(\succ,\succeq)$.
\end{thm}

Note that $\csucc{u}{v_1}{\varphi}$ and $\csucc{u}{v_2}{\varphi}$ are certainly satisfied if
\( \rw\ \subseteq\ \succ \), but that this is not \emph{required}.  This provides a bit more
freedom in the choice of the relation $\succ$.

In \autoref{thm:boundedRI} we showed that bounded rewriting induction proves
\( \rules/\emptyset/\emptyset \)-bounded ground convertibility. 
Therefore, we already have point (2), and conclude the following:

\begin{cor}[Bounded RI for ground confluence]
\label{cor:RIforGroundConfluence}
Let $\rules$ be a terminating, quasi-reductive \lcstrs\ and $(\succ, \succeq)$ a bounding pair such
that \( \csucc{u}{v_1}{\varphi} \) and \( \csucc{u}{v_2}{\varphi} \) for all
\( \peak{u}{v_1}{v_2}{\varphi} \in \critpairs{\rules} \).
Let $\eqs=\{ \eqcon{v_1}{v_1}{v_2}{v_2}{\varphi} \mid \peak{u}{v_1}{v_2}{\varphi} \in
\critpairs{\rules} \}$.
If $(\eqs,\emptyset) \vdash^* (\emptyset,\hs)$ for some set $\hs$, then $\rules$ is ground
confluent.
\end{cor}

\begin{exa}\label{example:groundConfluent}
In the \lcstrs\ from \autoref{example:nonEquivalence} we have    
$ 
\critpairs{\rules}
=
\{
\peak{\symb{H}\ f\ n\ m\ x}
     {\symb{H} \ f \ (n-\symb{1})\ m\ (f\ x)}
     {\symb{H} \ f \ (m-\symb{1})\ n\ (f\ x)}
     {n>\symb{0} \wedge m>\symb{0}}
\}
$. 
For \(\eqs = \{ \eqcon{\symb{H} \ f \ (n-\symb{1})\ m\ (f\ x)}{\symb{H} \ f \ (n-\symb{1})\ m\ (f\ x)}
{\symb{H} \ f \ (m-\symb{1})\ n\ (f\ x)}{\symb{H} \ f \ (m-\symb{1})\ n\ (f\ x)}{n>\symb{0} \wedge m>\symb{0}}\}\),
it is easy to show that
$
(\eqs, \emptyset)
\vdash^*
(\emptyset, \hs)
$ for some set $\hs$, using a bounding pair with \( \rw\ \subseteq\ \succ \).

By \autoref{cor:RIforGroundConfluence}, $\rules$ is ground confluent. 
\end{exa}

Towards a proof of \autoref{thm:groundCriticalPair} we introduce two lemmas. 

\begin{lem}\label{lem:contextShape}
Assume $C_1[a_1]=C_2[a_2]$. 
Then at least one of the following cases holds 
\begin{enumerate}[(i). ]
\item $C_1[\square]=D[\square, a_2]$ and $C_2[\square]=D[a_1, \square]$ for some context $D[\square, \square]$
\item $a_1=D[a_2]$ for some context $D[\square]$
\item $a_2=D[a_1]$ for some context $D[\square]$
\end{enumerate}
\end{lem}

\begin{proof}
By induction on the shape of $C_1$.
Any term can be written as $f\ s_1 \cdots s_n$ with $f \in \Sigma \cup \Var$ and $n \ge 0$.
Therefore, any context $C[\square]$ has one of the following shapes
\begin{enumerate}[(A). ]
\item $C[\square ] = \square\ s_{i+1} \cdots s_n$, for some $1 \le i \le n$
\item $C[\square] = f\ s_1 \cdots s_{i-1}\ C'[\square]\ s_{i+1} \cdots s_n$, for some $1 \le i \le n$ and context $C'$
\end{enumerate} 
Now, assume $C_1[a_1]=C_2[a_2]$. 
We consider the four combinations for $C_1[\square]$, $C_2[\square]$.
\begin{description}
\item[(AA)]  
$C_1[\square] = \square \ s_{i+1} \cdots s_n$ and 
$C_2[\square] = \square \ s_{j+1} \cdots s_n$
\begin{itemize}[label = $\triangleright$]
\item If $i=j$ then $C_1=C_2$ and $a_1=a_2=f\ s_1 \cdots s_i$. In particular (ii) holds with $D = \square$.
\item  If $i>j$ then $a_1 = f\ s_1 \cdots s_j \cdots s_i$ and $a_2=f\ s_1 \cdots s_j$. 
Then (ii) holds with $D[\square] = \square\ s_{j+1} \cdots s_i$ (then $a_1=D[a_2]$).
\item Similarly, if $i<j$ then (iii) holds by a symmetrical reasoning.
\end{itemize}
\item[(AB)] $C_1[\square] = \square\ s_{i+1} \cdots s_n$ and $C_2[\square] = f\ s_1 \cdots s_{j-1}\ C'[\square]\ s_{j+1} \cdots s_n$. 
In particular $a_1=f\ s_1 \cdots s_i$ and $C'[a_2]=s_j$.
\begin{itemize}[label = $\triangleright$]
\item If $i \ge j$ then (ii) holds: $a_1 = D[a_2]$ with $D[\square] = f\ s_1 \cdots s_{j-1}\ C'[\square]\ s_{j+1} \cdots s_i$. 
\item If $j>i$ then (i) holds: $D[\square_1, \square_2]=\square_1\ s_{i+1} \cdots s_{j-1} \ C'[\square_2] \ s_{j+1}\  \cdots s_n$.   
\end{itemize}
\item[(BA)] Symmetrical to $\textbf{(AB)}$. 
\item[(BB)] $C_1[\square] = f\ s_1 \cdots s_{i-1}\ C'_1[\square]\ s_{i+1} \cdots s_n$ and $C_2[\square] = f\ s_1 \cdots s_{j-1}\ C'_2[\square]\ s_{j+1} \cdots s_n$. 
In particular $C'_1[a_1]=s_i$ and $C'_2[a_2]=s_j$.
\begin{itemize}[label = $\triangleright$]
\item If $i = j$ then $C'_1[a_1]=C_2'[a_2]$ and we apply the induction hypothesis; in case
  (ii) and (iii) we are done, and if the induction hypothesis gives us $D'$
  we let $D[\square_1,\square_2] = f\ s_1 \cdots s_{i-1}\ D'[\square_1,\square_2]\ s_{i+1} \cdots s_n$.
\item If $j>i$ then (i) holds, choosing $D[\square_1, \square_2]=
  f\ s_1 \cdots s_{i-1} \ C'_1[\square_1] \ s_{i+1}\ \cdots s_{j-1} \ C'_2[\square_2]$\ 
  $\ s_{j+1}\ \cdots s_n$.  Similarly, if $i > j$ we have (i) by taking
  $D[\square_1, \square_2]=
  f\ s_1 \cdots s_{j-1} \ C'_2[\square_2]$\ $s_{j+1}\ \cdots s_{i-1} \ C'_1[\square_1] \ s_{i+1}\  \cdots s_n$.
\qedhere
\end{itemize}
\end{description}
\end{proof}

In the following, let $u \brightRp{s} v$ if we can write $u = C[a]$, $v = C[b]$,
$s \succ a$, $s \succ b$, and $a \rw^* b$.
Let $u \bleftRp{s} v$ if $v \brightRp{s} u$.  We use $u \bconvertRp{s} v$ if either
$u \brightRp{s} v$ or $u \bleftRp{s} v$.

\begin{lem}\label{lem:commonReduct}
Let $\rules$ be a terminating, non-ground locally confluent \lcstrs\ and
$(\succ, \succeq)$ a bounding pair.
Let $s$ be a $\succ$-minimal term that exists by \autoref{lem:GroundConfluenceExistsMinimal}.(i).
Suppose $v \bleftRp{s} u \brightRp{s} w$.
Then there is a term $u'$ such that $v \brightRp{s} u' \bleftRp{s} w$.
\end{lem}

\begin{proof}
There are contexts $C_1$, $C_2$ and terms $a_1$, $a_2$, $b_1$, $b_2$ such that 
\begin{itemize}[label = $\triangleright$]
\item $u = C_1[a_1]=C_2[a_2]$, $v=C_1[b_1]$ and $w=C_2[b_2]$
\item $a_1 \rw^* b_1$, $a_2 \rw^* b_2$ 
\item $s \succ a_1$, $s \succ b_1$, $s \succ a_2$, $s \succ b_2$.
\end{itemize}
By \autoref{lem:contextShape}
there are three options:
\begin{enumerate}[(i). ]
\item $C_1[\square]=D[\square, a_2]$ and $C_2=D[a_1, \square]$ for some context $D[\square, \square]$.   
Then 
\[
D[b_1, a_2] = v \bleftRp{s} u = D[a_1, a_2] \brightRp{s} w = D[a_1, b_2]
\]
We can take $u'=D[b_1, b_2]$. 
We have $s \succ b_1$ and $s \succ b_2$, as required. 

\item $a_1=D[a_2]$ for some context $D[\square]$. Then 
\[
C_1[b_1] = v \bleftRp{s} u = C_1[a_1] = C_1[D[a_2]]
\brightRp{s} w =
C_1[D[b_2]]
\]
Since $s \succ a_1=D[a_2]$ and 
\(
b_1 \rwleftplus D[a_2] \rw^+ D[b_2]
\)
we apply \autoref{lem:GroundConfluenceExistsMinimal}.(ii) to obtain a term $t$ such that
\(
b_1 \rw^* t \rwleftstar D[b_2]
\).
We claim that with $u'=C_1[t]$ we obtain the desired result 
\(
v = C_1[b_1] \brightRp{s} C_1[t] =u' \bleftRp{s} C_1[D[b_2]] = w \). 
The reductions within the context $C_1$ are already established, so it remains to verify the necessary inequalities $s \succ b_1, s \succ t$ and $s \succ D[b_2]$. 
The inequality $s \succ b_1$ holds by assumption. 
Combining this with $b_1 \rw^* t$
(and $\rw\;\subseteq\;\succeq$) it follows that $s \succ t$.
For the last inequality, we combine the assumption $s \succ a_1 = D[a_2]$ with $D[a_2] \rw^* D[b_2]$ to conclude $s \succ D[b_2]$.
\item $a_2=D[a_1]$ for some context $D[\square]$: similar to (ii). 
  \qedhere
\end{enumerate}
\end{proof}

With this, we are ready to prove the Ground Critical Peak Theorem:

\begin{proof}[\textbf{Proof of \autoref{thm:groundCriticalPair}}]
In a terminating system, ground confluence is implied by ground local confluence.  So, towards a
contradiction, assume all equation contexts corresponding to critical pairs are bounded ground
convertible but that $\alcstrs$ is not ground locally confluent. 

By \autoref{lem:GroundConfluenceExistsMinimal}.(i) there is a $\succ$-minimal
ground term $s$ and terms $t_1, t_2$ for which both
\begin{enumerate}[label = (\arabic*).]
\item $t_1 \rwleft s \rw t_2$ and there is no $t$ such that $t_1 \rw^* t \rwleftstar t_2$ 
\item there is a critical peak $\peak{u}{v_1}{v_2}{\varphi}$ and ground substitution
$\gamma$ that respects $\varphi$ such that  $s=u\gamma$, $t_1=v_1 \gamma$ and $t_2=v_2\gamma$
\end{enumerate}
Combining (2) with the assumption that $\eqcon{v_1}{v_1}{v_2}{v_2}{\varphi}$ is
$\rules/\emptyset/\emptyset$-bounded ground convertible, we have $t_1 \bconvert{t_1}{t_2}{\rules}{\emptyset}{\emptyset} t_2$. 
So there exist $c_0, \ldots, c_k$ such that $t_1 = c_0$ and $t_2 = c_k$, and 
$c_i \convertsingle{t_1}{t_2}{\rules}{\emptyset}{\emptyset} c_{i+1}$ for all $0 \le i < k$.  
In particular, every such step can only be performed by case \ref{boundconvert:rule} of \autoref{def:boundconvert} because $\hs = \eqs = \emptyset$. 
That means that for all $0 \le i < k$ we have $c_i=C[\ell \sigma] \leftrightarrow_{\rules} C[r \sigma]=c_{i+1}$, for some $\myrule \in \rules$ or $r \to \ell\ [\varphi] \in \rules$, where both $\ell \sigma$ and $r \sigma$ are $\succeq$-dominated by either $t_1$ or $t_2$. 
By assumption, we also have $\csucc{u}{v_1}{\varphi}$ and $\csucc{u}{v_2}{\varphi}$, which implies $s \succ t_1$ and $s \succ t_2$. 
Since $\succ \cdot \succeq \ \subseteq \ \succ$, we conclude that both $\ell \sigma$ and $r \sigma$ are $\succ$-dominated by $s$. 
Therefore, $c_i \bconvertRp{s} c_{i+1}$ for all $0 \le i < k$.

To obtain the contradiction, we show that there is a
$z$ such that $t_1 \rw^* z \rwleftstar t_2$.
We do so using induction on $k$.

\begin{itemize}[label = $\triangleright$]
\item $k = 0$.  Then $t_1 = t_2$ and we can choose $z = t_1$.
\item $k = 1$.  Then $t_1 \rw^* t_2$ or $t_2 \rw^* t_1$, so it trivially holds.
\item $k \geq 2$. We distinguish the following cases:
  \begin{itemize}
  \item If $t_1 = c_0 \brightRp{s} c_1$, the induction hypothesis yields $z$ such that
    $c_1 \rw^* z$ and $t_2 \rw^* z$.
  \item If $c_{k-1} \bleftRp{s} c_k = t_2$, the induction hypothesis yields $z$ with
    $t_1 \rw^* z$ and $c_{k-1} \rw^* z$.
  \item Otherwise, there exists $m > 0$ such that $c_0 \bleftRp{s} \dots \bleftRp{s} c_{m-1}
    \bleftRp{s} c_m \brightRp{s} c_{m+1}$.  We use a second induction on $m$.
    Observe that, by \autoref{lem:commonReduct}, there exists $c_m'$ such that
    $c_{m-1} \brightRp{s} c_m' \bleftRp{s} c_{m+1}$.  Then the sequence
    $c_0,\dots,c_{m-1},c_m',c_{m+1},\dots,c_k$ still satisfies the properties above.
    If $m = 1$, we have $t_1 \rw^* c_1'$ and we complete by the first induction hypothesis. 
    If $m > 1$ then we observe that now $c_0 \bleftRp{s} \dots \bleftRp{s} c_{m-1}
    \brightRp{s} c_m'$, so we complete by the second induction hypothesis.
    \qedhere
  \end{itemize}
\end{itemize}
\end{proof}

\section{Implementation}\label{sec:implementation}

We have implemented bounded rewriting induction in our tool \cora~\cite{cora} as
an interactive proving procedure.  Here, the user has to supply the proof steps,
while the tool keeps track of all equation contexts, ensures that deduction
rules are applied correctly, and searches for a bounding pair.
We choose the strategy from \autoref{sec:redord}: we maintain \emph{strongly}
bounded equation contexts by viewing each occurrence of
\( \csucceq{s}{t}{\psi} \) in \autoref{fig:boundedRIrules} as
\( s \succeq^! t\ [\psi] \), and use the results of
\cite{guo:hag:kop:val:24} to, essentially, prove termination of the LCSTRS
\( \reqs \cup \rules \).

Validity and satisfiability checks are delegated to an external SMT-solver (by
default \texttt{Z3}~\cite{z3}, but the user can set other preferences),
on the theory of (first-order) integer arithmetic (\texttt{QF\_NIA})~\cite{smtlib}.  For the two
highly undecidable derivation rules \generalize\ and \alter, only limited
versions are available to handle the most common cases (which allows us to avoid
universally quantified constraints, and hence stay in QF-NIA):
\begin{itemize}
\item For \generalize, \cora\ allows the user to change the constraint to an
  implied one, or to supply an equation that the current goal is an instance
  of.
\item For \alter, \cora\ allows the user to supply an equivalent constraint, or
  add definitions \( x = u \) into the constraint with \( x \) a fresh variable.
  In addition, case \ref{alter:substitute} combined with a \delete\ step can be
  accessed through the \texttt{eq-delete} command. The name refers to the
  \eqdeletion\ rule discussed in \autoref{sec:improvementsRI}, though it
  actually combines \eqdeletion\ with \delete, for the common use case
  that
  \( \varphi \wedge \neg (\wedge_{i=1}^n s_i = t_i) \) is unsatisfiable.
\end{itemize}

\begin{figure}[tp]
\includegraphics[width=\textwidth]{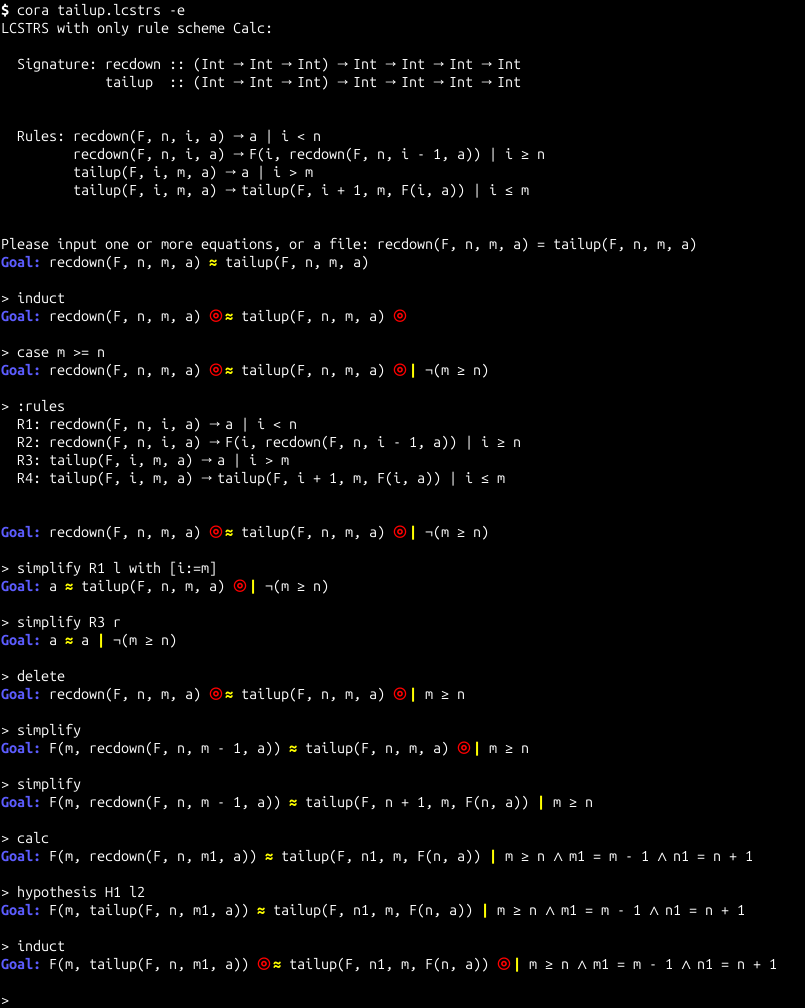}
\caption{Basic usage of \cora}
\label{fig:corabase}
\end{figure}

The use of \cora\ is illustrated in \autoref{fig:corabase}.  By default, \cora\ 
shows only the equation part of an equation context, coupled with a red
$\circledcirc$ if a side of the equation is the same as its bounding term.
(The full equation context can be queried using the command \texttt{:equations full}.)
The deduction rules are implemented through the commands \texttt{simplify},
\texttt{case}, \texttt{delete}, \texttt{induct}, \texttt{hypothesis},
\texttt{hdelete}, \texttt{generalise}, \texttt{alter}, \texttt{postulate},
\texttt{disprove},
\texttt{calc} (simplifying with a calculation rule, see \autoref{sec:newDedEasy}),
\texttt{eq-delete} (explained above),
\texttt{semiconstructor} and \texttt{context} (these two handle
\semicons\ and \headvar\ respectively).
A \texttt{:help} command explains their usage, and lists other commands like
\texttt{undo} or \texttt{:save}.
Termination checks are done once the proof is complete, but can also be forced
earlier by the user executing \texttt{:check}, or
\texttt{:check horpo} to force a proof by a recursive path ordering.

While \cora\ does require user guidance, there is some automation.
As shown in \autoref{fig:corabase}, commands like \texttt{simplify} can be called
with more or less information: \texttt{simplify} without arguments rewrites with
an arbitrary rule at some innermost position, but the user can also supply the
rule, position and substitution to be used.
An \texttt{auto} command is supplied to automatically do \texttt{simplify},
\texttt{calc}, \texttt{delete}, \texttt{hdelete}, \texttt{eq-delete},
\texttt{disprove} and \texttt{semiconstructor} steps as far as possible.
This use of automation is shown in \autoref{fig:coraauto}.

\begin{figure}[tp]
\includegraphics[width=\textwidth]{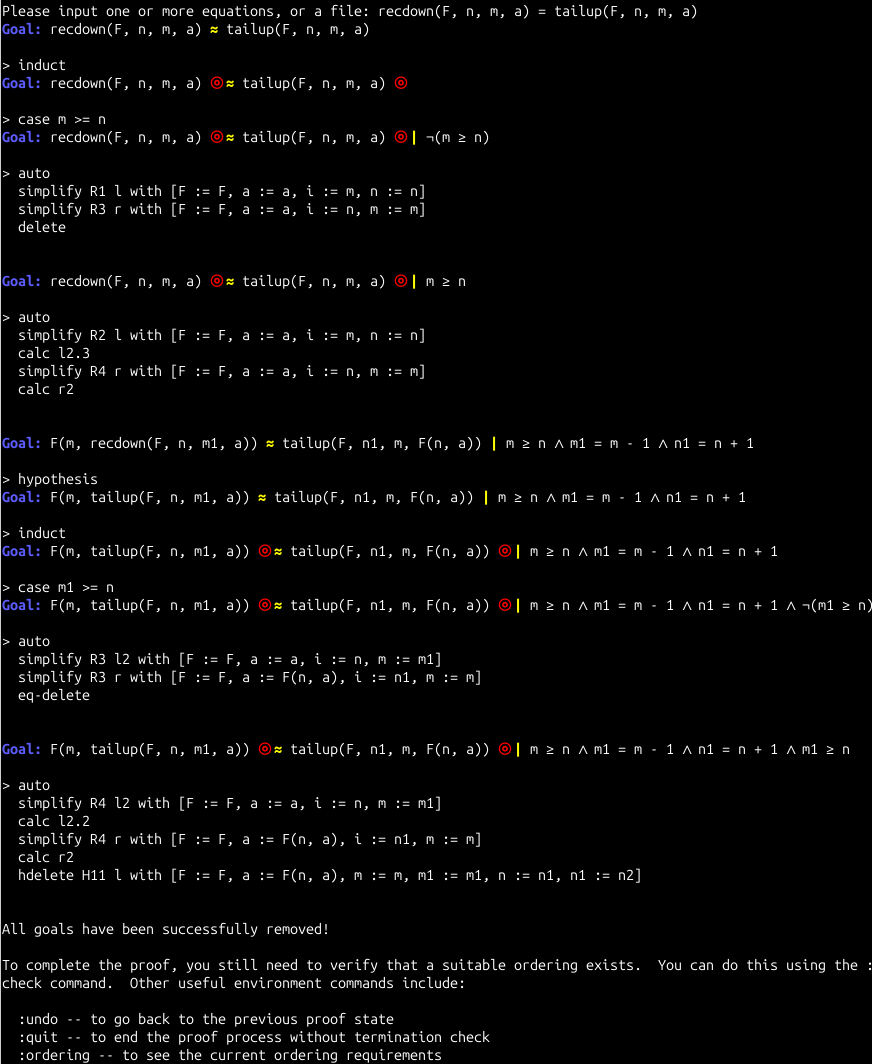}
\caption{Using the \texttt{auto} command in \cora}
\label{fig:coraauto}
\end{figure}

In future work, we hope to further automate the rewriting induction process.
However, a fully automatic process will need to include a powerful lemma
generation engine for automatic application of \postulate\ and \generalize, so this is a
large project.

An evaluation page with save files that provide full proofs for the inductive
theorems in this paper, a pre-compiled version of \cora\ and instructions on how
to start the proof process for an interactive proof, is provided online.  We have
also included Bounded RI proofs for the (unconstrained, first-order) benchmarks
in~\cite{aot:06,aot:08a}.
(Though with a disclaimer: we could not find the definitions of the rules in these
benchmarks, so we guessed at the implementation of for instance the
$\symb{times}$ and $\symb{max}$ rules).
You can find the website at:

\begin{center}
\url{https://cs.ru.nl/~cynthiakop/experiments/lmcs26/}
\end{center}

\section{Conclusion \& Future Work}
\label{sec:conclusionFuture}
With the introduction of Bounded Rewriting Induction, we revised the existing
RI system for \lcstrss, making fundamental changes with the aim of reducing termination requirements. 
We replaced a well-founded order $\succ$ or terminating relation
$\to_{\rules \cup \hs}$ with a \emph{bounding pair} $(\succ, \succeq)$ allowing
for less strict ordering requirements with $\succeq$.  
We replaced equations by \emph{equation contexts}, which contain \emph{bounding terms}, keeping track of induction bounds efficiently. 
As a byproduct, we have obtained an intuitive proof system with a high degree of modularity,
which makes it easy to add new deduction rules to the system.
Moreover, we have defined critical pairs and peaks for \lcstrss\ and shown how Bounded RI can be used to prove
ground confluence.
Finally, as a natural application of this result, we showed that for ground confluent \lcstrss,
Bounded RI can be extended to also disprove inductive theorems. 

\medskip
We would like to conclude with some practically motivated directions for future research.

\paragraph{Global rewriting induction}
A property not discussed in this work, but a major topic in~\cite{aot:yam:chi:11, hag:kop:24}, is \emph{extensibility}.
This means that if an equation is an inductive theorem in $\alcstrs$, it remains an inductive theorem in any ``reasonable''
extension of $\alcstrs$ (intuitively, a reasonable extension is one that represents a real-world program, in which $\alcstrs$ is a stand-alone module).
Such inductive theorems are called \emph{global inductive theorems}.
In terms of functional programming, if two functions are equivalent this should not
change when these functions are used inside a larger program.
Extensibility is a way to express that local reasoning extends globally, at least for such reasonable \lcstrss. 
From a perspective of software verification this is a desirable property: to prove properties about a small part of a larger system, we only need to consider the rules that are directly related.

\emph{Global Rewriting Induction}~\cite{hag:kop:24} is an extension of RI, designed to prove global inductive theorems. 
An interesting question is whether we can apply a similar construction to Bounded RI.
The main challenge is to guarantee that a bounding pair $(\succ, \succeq)$ on $T(\Sigma,\emptyset)$ can
be extended to a bounding pair $(\succ', \succeq')$ on $T(\Sigma', \emptyset)$, for any reasonable extension
$(\Sigma',\rules')$ of $(\Sigma,\rules)$. 
In particular, we will have to reconsider every deduction rule in \autoref{fig:boundedRIrules} that has an ordering
requirement, including \hypothesis, \hdelete, \generalize\ and \alter, and observe that the meaning of
``every gsc substitution'' in Definitions~\ref{def:coverset} and~\ref{def:generalize} changes.

\paragraph{Ground convertibility modulo axioms}
In practice, we encounter inductive theorems which are only provable once we are allowed to utilize our knowledge about the instantiation of certain variables. 
For example, an equivalence $\symb{P}_1\ f \approx \symb{P}_2\ f$ may hold only
under the assumption that  $f :: \int \to \int \to \int$ is instantiated by a
commutative function.
So instead of considering convertibility under \emph{all} ground substitutions, we restrict to those ground substitutions that instantiate $f$ as a commutative function.  

Designing a corresponding RI system is non-trivial, because this new kind of equivalence is no longer based on the standard notion of inductive theorem, but on a restricted form of ground convertibility. 
This likely brings us into a setting similar to term rewriting modulo a set $\mathcal{A}$ of axioms.  
A natural first step is therefore to examine how the soundness proofs from \autoref{sec:proofs} can be adapted to the setting of rewriting modulo $\mathcal{A}$. 

\paragraph{Automating RI proofs}
Currently, any proof by Bounded RI is (almost) completely human driven. 
Ideally, we would like to incorporate some tactics in \texttt{Cora} for automatically proving inductive theorems with Bounded RI, building on what has been implemented in \texttt{Ctrl}~\cite{fuh:kop:nis:17} for first-order \lctrss.
As observed in \autoref{sec:implementation}, a major obstacle in this is finding
lemmas automatically. 
Here, existing generalization methods can be employed, but it remains a challenge to determine when and how to apply them effectively.
As program equivalence is undecidable in general, we cannot expect them to be exhaustive,
but combining the strengths of these traditional techniques with machine learning
may be very powerful.

\paragraph{Unbounded RI}
The definition of Bounded RI takes advantage of the restriction that all our
equation contexts are bounded (and remain so because our derivation rules have
the Preserving Bounds property).  However, this restriction is not fundamental:
as shown by the proofs in \autoref{sec:proofs}, what matters most is that the
right multiset ordering conditions are imposed.  To obtain these conditions we
currently use the property that equation contexts are bounded, but we could
alternatively require it directly in the derivation rules.  This would for
instance replace the \hypothesis\ rule by
\DEDUCRULE{(\eqs \uplus \{\eqconsim{\sterm}{C[\ell\delta]}{t}{\tterm}{\psi}\}, \hs)}
          {(\ell \approx r\ [\varphi]\in \hs \text{ or } r \approx \ell\ [\varphi] \in \hs) \text{ and } \psi \models^\delta \varphi \text{ and} \\
            &\multiset{\sterm,\tterm} \succmul \multiset{\ell\delta,r\delta}\ [\psi]}
          {(\eqs \cup \{\eqconsim{\sterm}{C[r \delta]}{t}{\tterm}{\psi}\}, \hs)}
And \induct\ would become:
\DEDUCRULE{(\eqs \uplus \{\eqcon{\sterm}{s}{t}{\tterm}{\psi}\}, \hs)}
          {\multiset{\sterm,\tterm} \succeqmul \multiset{s,t}\ [\psi]}
          {(\eqs \cup \{\eqcon{s}{s}{t}{t}{\psi}\}, \hs \cup \{s \approx t\ [\psi]\})}
On the upside, we could then weaken the restrictions for \alter/\generalize\ 
(since it is no longer necessary to preserve bounds).  
At a first pen-and-paper check, this altered system works to prove ground
convertibility of equations: adaptations of all the proofs in this paper appear
to go through.  However, it does not give \emph{bounded} ground convertibility.

The advantage of such a change is that we can further weaken the ordering
requirements: instead of requiring that \( \rw \) is included in \( \succeq \)
it would suffice if \( s \succeq s\downarrow_\rules \) for all ground terms
\( s \).
However, a downside is that it is not obviously usable for proving ground
confluence.

The most important complication is that six different combinations of $\succ$
and $\succeq$ are admissible when realizing a multiset ordering.  This is not a
problem if the bounding pair \( (\succ,\succeq) \) is fixed in advance, but
\emph{is} challenging if the search for a bounding pair is done afterwards.
Naively exploring the entire search space -- i.e. considering all options at
every deduction step that imposes requirements on the ordering -- would lead to
exponential growth in the length of the deduction sequence.
From this perspective, Bounded RI can be seen as a heuristic that provides
guidance on which ordering requirements to impose.
While Bounded RI still gives \emph{some} choice on the ordering requirements,
the choice is usually straightforward (in \hypothesis, choose \( \pi = \sterm \)
unless this causes immediate problems) or low-impact (in \hdelete, if neither
\( \sterm \) or \( \tterm \) has a variable or calculation symbol at the head,
just pick one).

In future work, it would be interesting to explore Unbounded RI and search for
strategies to take advantage of the greater generality, without sacrificing the
ease of finding \( (\succ,\succeq) \).

\bibliographystyle{alphaurl}
\bibliography{references}
\end{document}